\def\conference{0}
\def\shownotes{1}

\ifnum\conference=1

\documentclass[sigconf]{acmart}

\copyrightyear{2026}
\acmYear{2026}
\setcopyright{cc}
\setcctype{by}
\acmConference[CCS '26]{Proceedings of the 2026 ACM SIGSAC Conference on Computer and Communications Security}{November 15--19, 2026}{The Hague, Netherlands}
\acmBooktitle{Proceedings of the 2026 ACM SIGSAC Conference on Computer and Communications Security (CCS '26), November 15--19, 2026, The Hague, Netherlands}
\acmDOI{10.1145/3830454.3846678}
\acmISBN{979-8-4007-2871-6/2026/11}
\newcommand{\myparab}[1]{\subsubsection{#1}}
\else
\documentclass[twocolumn]{article}
\usepackage[margin=.75in]{geometry}
\newcommand{\myparab}[1]{\paragraph{#1}}
\usepackage{mathtools}
\usepackage{cite}
\usepackage{pgfplots}
\pgfplotsset{compat=1.17}
\usepackage{tikz}
\usepackage{balance}
\usepackage{authblk}
\usepackage{hyperref}
\usepackage{color}
\usepackage{boldline}
\fi

\usepackage{amsthm}
\usepackage{amssymb}
\usepackage{url}
\usepackage{framed}

\usepackage{makecell}

\usepackage{subcaption}
\usepackage{multirow}

\usepackage[inline]{enumitem}
\setlist[enumerate,1]{label={(\arabic*)}}

\newtheorem{theorem}{Theorem}
\newtheorem{definition}{Definition}
\newtheorem{lemma}{Lemma}

\newtheorem{corollary}{Corollary}

\DeclareMathOperator*{\Exp}{\mathbb{E}}

\newcommand{\id}{\operatorname{id}}
\newcommand{\tab}{\mathrm{tab}}

\newcommand{\cvr}{\mathrm{cvr}}
\newcommand{\dup}{\mathrm{dup}}
\newcommand{\coarse}{\mathrm{coarse}}

\newcommand{\act}{\mathrm{act}}
\newcommand{\obs}{\mathrm{obs}}

\newcommand{\size}{\mathsf{S}}

\newcommand{\winner}{\mathsf{W}}
\newcommand{\batchnum}{\beta}
\newcommand{\batches}{m}
\newcommand{\batch}{\mathsf{batch}}
\newcommand{\manifest}{\mathsf{Manifest}}
\providecommand{\TV}{d_{\mathrm{TV}}}
\newcommand{\loser}{\mathsf{L}}
\newcommand{\disc}{\mathsf{D}}
\newcommand{\obsdisc}{\mathsf{D}^{\Duplicates}}

\newcommand{\iter}{\mathsf{iter}}
\newcommand{\ballot}{\mathbf{b}}
\newcommand{\Ballots}{\mathbf{B}}
\newcommand{\sample}{\mathsf{Sample}}
\newcommand{\inAcc}{\Delta}
\newcommand{\inAccSmall}{\delta}

\newcommand{\Duplicates}{\mathsf{Duplicate}}

\newcommand{\Labels}{{L}}

\newcommand{\consistent}{\textsc{Consistent}}
\newcommand{\inconclusive}{\textsc{Inconclusive}}
\newcommand{\RLA}{\mathsf{RLA}}
\newcommand{\Stop}{\mathsf{Stop}}

\newcommand{\Stopdeltaalp}{\mathsf{Stop}_{\sigma, \alpha}}
\newcommand{\FirstStop}{\tau_{\Stop}}
\newcommand{\risk}{\alpha}
\newcommand{\Criterion}{\mathsf{R}}

\newcommand{\Criteriondeltaalp}{\mathsf{R}_{\sigma, \alpha}}

\newcommand{\Auditor}{\mathcal{C}}
\newcommand{\Adversary}{\mathcal{A}}
\newcommand{\N}{\mathbb{N}}

\newcommand{\procedure}[1]{\mathtt{#1}}
\newcommand{\DetectDuplicates}{\procedure{DetectDuplicates}}
\newcommand{\BoundSize}{\procedure{BoundSize}}
\newcommand{\BasicExperiment}{\procedure{BasicExperiment}}
\newcommand{\CheckManifest}{\procedure{CheckManifest}}
\newcommand{\Test}{\mathcal{T}}

\newcommand{\Bad}{\mathsf{Bad}}
\newcommand{\R}{\mathbb{R}}

\newcolumntype{L}[1]{>{\raggedleft\arraybackslash}p{#1}}

\ifnum\shownotes=1
\newcommand{\authnote}[2]{{\marginpar{\tiny \textcolor{red}{\textsf{#1 notes: }\textcolor{blue}{ #2}}}}}
\else
\newcommand{\authnote}[2]{}
\fi

\newcommand{\fullversionref}{the full version~\cite{fuller2026sublinear}}
\newcommand{\appref}[1]{\ifnum\conference=1 \fullversionref\else Appendix~\ref{#1}\fi}
\newcommand{\Appref}[1]{%
  \ifnum\conference=1
    The full version~\cite{fuller2026sublinear}%
  \else
    Appendix~\ref{#1}%
  \fi
}

\begin{document}

\title{Sublinear Risk-Limiting Audits from Direct Ballot Selection and Statistical Ballot Manifests}
\ifnum\conference=0
\author{Benjamin Fuller\thanks{\textcopyright{} Benjamin Fuller, Abigail Harrison,
Alexander Russell 2026. This is the author's version of The Work. It is posted
here for your personal use. Not for redistribution. The definitive version was
published in CCS 2026,
\href{https://doi.org/10.1145/3830454.3846678}{\texttt{https://doi.org/10.1145/3830454.3846678}}.}}
\author{Abigail Harrison}
\author{Alexander Russell}

\affil{
  {\small \texttt{\{benjamin.fuller,abigail.harrison,acr\}@uconn.edu}}\\
  Voting Technology Research Laboratory\\
  University of Connecticut}
\else
\author{Benjamin Fuller}
\email{benjamin.fuller@uconn.edu}
\affiliation{%
  \department{Voting Technology Research Laboratory}
  \institution{University of Connecticut}
  \city{Storrs}
  \state{Connecticut}
  \country{USA}}
\author{Abigail Harrison}
\email{abigail.harrison@uconn.edu}
\affiliation{%
  \department{Voting Technology Research Laboratory}
  \institution{University of Connecticut}
  \city{Storrs}
  \state{Connecticut}
  \country{USA}}
\author{Alexander Russell}
\email{acr@uconn.edu}
\affiliation{%
  \department{Voting Technology Research Laboratory}
  \institution{University of Connecticut}
  \city{Storrs}
  \state{Connecticut}
  \country{USA}}
\renewcommand{\shortauthors}{Fuller et al.}
\fi

\ifnum\conference=0
\maketitle
\fi

\begin{abstract}
  Risk-limiting audits (RLAs) are post-election auditing procedures that rigorously guarantee a specified maximum probability that an incorrect electoral outcome will not be detected. Aside from ready access to physical ballots, efficient RLA methods require a software-independent 
  count of the ballots in each batch, called a \emph{ballot manifest}. While electoral procedures can efficiently provide rough estimates for batch sizes, even slight inaccuracies (commensurate with the margin of the contest under audit) can invalidate conventional RLAs (Lindeman et al., EVT 2012). Thus, establishing a sufficiently accurate manifest often requires handling every ballot in the election and can dominate the cost of conducting an RLA.

With the goal of providing audits with sublinear effort across the full spectrum of election architectures,
we propose two new risk-limiting techniques. The first is a statistical test that checks a trusted, coarse manifest 
against an untrusted, tabulator-supplied manifest to certify that the aggregate error is small; this bounds both the error in the reported ballot total and the distortion of the ballot-sampling distribution. The second is a new approach for election architectures that do not provide efficient indexing of ballots by identifier, as is typical of voter-facing tabulators. We call this approach \emph{direct ballot selection}: it reverses the traditional comparison procedure by selecting physical ballots uniformly and comparing them to their corresponding cast vote records. This method also incorporates a new statistical test to check for identifier duplication.

Taken together, these techniques provide a striking reduction in the effort required to conduct RLAs across a wide swath of election configurations.
Our two main findings are as follows:
  \begin{enumerate}
  \item Manifest creation time can be reduced with only a modest increase in the number of ballots sampled in the audit. For California (the largest state by population in the United States) at a 3\% margin, our model indicates that the overall audit time for comparison, polling, and direct selection audits is reduced by factors of approximately $438$, $27$, and $10$, respectively.
  \item Direct ballot selection improves over state-of-the-art polling for small margins. For Connecticut (29th-largest by population) at a 1\% margin, it requires 55\% fewer ballots than the Minerva (USENIX Security 2021) and Providence (USENIX Security 2023) ballot polling methods.
  \end{enumerate}
\end{abstract}

\ifnum\conference=1

\ccsdesc[500]{Theory of computation~Cryptographic protocols}
\ccsdesc[500]{Mathematics of computing~Probabilistic algorithms}
\ccsdesc[500]{Security and privacy~Human and societal aspects of security and privacy}

\keywords{risk-limiting audits, election security, ballot manifests, direct ballot selection}

\maketitle
\fi

\section{Introduction}

Elections are fundamental civic processes that are challenging to analyze from a security perspective, 
as they rely on complex heterogeneous systems combining software, hardware, paper records, 
and human procedures at scale. Among the most important advances in election security has been the development of \emph{risk-limiting audits} (RLAs)---
post-election auditing procedures that provide rigorous statistical guarantees of election outcomes in a natural security model~\cite{lindeman2012gentle,RLAWorkbook,rla-working-group,bernhard2021risk,stark2009auditing,stark2010super,higgins2011sharper,lindeman2012bravo,ottoboni2018risk,ottoboni2019bernoulli,stark2020sets,waudby2021rilacs,blom2021assertion,zagorski2021minerva,harrison2022adaptive,jones2024scan,fuller2024decisive}: If the reported outcome of an election is incorrect, an RLA establishes a known, pre-specified maximum chance---called the \emph{risk limit}---that the audit will fail to detect that error.

A critical, defining property of RLAs is \emph{software independence}: their guarantees 
do not rely on the correctness of 
ballot tabulators or other electronic equipment, but on an external, physical record of the election.
This can be achieved by assuming that auditors have access to the paper ballots that reflect the intent of the voters. RLAs are typically designed to further guarantee \emph{transparency}; that is, a third party observing the audit proceedings can also verify the outcome of the audit.

RLAs call for auditors to hand count a collection of randomly selected ballots, with results that are incorporated into an appropriate statistical test. Thus, an important metric for evaluating the efficiency of an RLA is \emph{sample complexity}: the total number of ballots that must be counted by hand in order to guarantee the desired risk limit. This topic has been given extensive attention in the literature and depends on both the audit method and the \emph{margin} $\mu$ of the contest. (For a conventional two-candidate, first-past-the-post election, the margin is the difference between the vote totals for the candidates as a fraction of the total number of ballots cast.)

The prevailing methods have a further requirement: a \emph{ballot manifest}. In the typical setting where ballots are organized into physical batches, a ballot manifest specifies the number of ballots in each batch. Sample-efficient RLA methods are extremely sensitive to ballot manifest inaccuracies: in particular, multiplicative inaccuracies in the declared batch sizes of order $1 + \inAccSmall$, where $\inAccSmall \approx \mu$, are sufficient to invalidate the audit's risk guarantee (see Lindeman et al.~\cite{lindeman2012bravo} and Appendix~\ref{sec:inaccurate no work}). 
In order to maintain strict software independence, the conventional approach is to establish the manifest by manually counting the ballots in each batch, a process which scales linearly with the total number of ballots cast. Consequently, when the underlying statistical tests have moderate sample complexity, manifest creation can dominate the overall cost of conducting the audit.

An additional challenge arises for \emph{ballot comparison audits}, which provide the best sample efficiency among known audit methods. These audits require a detailed ``ballot-by-ballot'' tabulation of the election, consisting of a comprehensive table of ``cast vote record'' (CVR) entries; an individual CVR declares the votes appearing on a particular ballot. Ballots are sampled by uniformly picking a CVR, identifying the individual cast ballot associated with it, and then retrieving and interpreting the ballot. While modern tabulators can imprint identifiers onto ballots as they are cast for CVR association, voter-facing tabulators---that is, tabulators with which 
voters interact directly 
during ballot casting---intentionally randomize ballot order to preserve voter privacy~\cite{crimmins2024dvsorder}. As a result, 
locating 
a specific ballot in an unordered batch may require handling the entire batch, an independent requirement that can necessitate effort scaling linearly with the number of ballots cast. (A costly alternative method that does not 
require ballot identifiers is to hand interpret the whole batch, called batch comparison; see below.)

We propose two new techniques for risk-limiting audits that address these challenges. They improve RLA efficiency across a wide spectrum of election architectures and auditing methods.
\begin{itemize}
\itemsep0em
\item\textbf{Statistical Accuracy Testing for Ballot Manifests.}
Coarse ballot manifests---providing batch-size estimates accurate to
within, say, 10\%---can be obtained with low effort through bulk
methods such as weighing or measuring ballots. 
On the other hand, electronic tabulation provides an untrusted but 
typically accurate declaration of the batch sizes. We 
define a statistical test
that uses the trusted coarse manifest to assess the tabulator's
declaration. If the test 
succeeds, then the declared batch sizes have small total $\ell_1$ error,
except with a prescribed failure probability.
This guarantee simultaneously bounds the error
in the total number of ballots and the distortion of the
size-proportional sampling distribution. We incorporate both the
failure probability and the residual sampling distortion into the
risk limit of the resulting audit.
  \item\textbf{Comparison RLAs via Direct Ballot Selection.} In election settings with high-accuracy manifests  
  and ballot identifiers in random order, we propose a new RLA method. Procedurally, these audits---which we call \emph{direct ballot selection} audits---reverse the conventional ballot comparison procedure by randomly selecting a physical ballot within a sampled batch and comparing it with its declared CVR. The method incorporates a new statistical test for \emph{identifier duplication}.
    For this reason, direct ballot selection exhibits a rather different scaling than that of existing audits. In certain natural parameter regimes, these new audits provide striking efficiency improvements over state-of-the-art polling and batch comparison audits, which are the immediate alternatives when identifiers are 
    arbitrarily ordered.
\end{itemize}

\subsection{The Risk-Limiting Auditing Landscape}
Traditional RLAs can be roughly organized into three methods:
\begin{itemize}
\itemsep0em
\item \textbf{Ballot polling audits} randomly sample ballots and test whether their empirical statistics support the reported winner. Polling requires $O\left({\log(1/\alpha)}/{\mu^2}\right)$ sampled ballots, where $\alpha$ is the risk limit and $\mu$ is the tabulated margin.
\item \textbf{Batch comparison audits} assume that ballots are organized into batches, each of which has been separately tabulated. The audit selects and hand counts random batches, comparing these subtotals against the corresponding tabulator 
subtotals. This requires $O\left({B\log(1/\alpha)}/{\mu}\right)$ sampled ballots in expectation, where $B$ is the size-weighted average batch size.
\item \textbf{Ballot comparison audits} assume that the tabulation includes a CVR
---a declaration of how a ballot was interpreted by the voting tabulator---for each cast ballot.
The audit randomly selects CVRs and compares them to human interpretations of the corresponding ballots. Ballot comparison requires $O\left({\log(1/\alpha)}/{\mu}\right)$ sampled ballots and requires a method to locate a specific ballot within a batch.
\end{itemize}

\myparab{Ballot manifests.} An essential component of all these procedures is an authoritative \emph{ballot manifest}---a determination of the total number of ballots in the election and, 
in settings in which ballots are organized into batches, the number of ballots in each batch. As mentioned above, accurate, software-independent manifests are critical for correctness; as a result, 
creating a ballot manifest according to best practices requires handling every ballot in the election. There are certain electoral settings where such exhaustive handling of ballots is unavoidable, such as unpacking mailed envelopes and updating the corresponding voter's record in the pollbook; however, in settings where voters cast ballots directly into a voter-facing tabulator, manifest creation is an additional step that scales linearly with the overall size of the election. In these cases, manifest creation is a significant factor in the overall effort. 

Lindeman et al.~\cite{lindeman2012bravo} point out the necessity of an accurate ballot manifest and, in particular, discuss the threshold beneath which errors can be tolerated in conventional RLAs.

\subsection{Statistically Approximate Manifests}

As mentioned above, conventional RLA methods require sufficiently accurate
manifests. Our analysis needs a weaker guarantee: most of the reported
ballot mass must be in batches whose true and reported sizes are close, while
the remaining batches may have larger errors. This guarantee controls both
the total number of ballots and how much the ballot-sampling distribution can
change when reported batch sizes are used in place of the true sizes. The
latter change can be measured by total variation distance, or equivalently by
one half of the $\ell_1$ distance between the two distributions. Direct
selection requires one additional property: most ballots must retain a
lower bound on the probability of being selected. 
Typical election procedures inherently provide partial information about batch sizes, ranging from simple upper bounds based on the capacity of storage containers to 
a ceiling on the total number of ballots cast, provided by a hand-verified count of the voter check-ins.
This information can be supplemented by efficient bulk measurements, such as weighing ballots or measuring ballot stacks. We assume that these procedures provide a software-independent, certified multiplicative bound on each batch size---for example, accuracy to within 10\%. Electronic tabulation, by contrast, provides untrusted but typically 
more precise declarations of batch sizes. 
Although these coarse bounds alone may be insufficient when the contest margin is small, they can still rule out gross errors in the tabulator-supplied declarations.

Our test first compares the tabulator-supplied manifest against the trusted coarse estimates. It then samples batches with probability proportional to their declared sizes, 
requests an accurate hand count of the selected batches, and compares 
the batches' true and declared sizes. If every sampled batch passes this check, then, except with a prescribed failure probability, only a small fraction of the reported ballot mass lies in substantially inaccurate batches. This certificate controls both the error in the total ballot count and the distortion of the resulting ballot-sampling distribution. We use these guarantees to establish the soundness of polling, comparison, and direct ballot selection audits conducted with the statistically certified manifest.

\subsection{Direct Ballot Selection}
The asymptotic scalings of current RLA methods appear to establish ballot comparison audits as the dominant approach, especially since audits must typically be planned and provisioned under a worst-case assumption of small $\mu$ (e.g., $\approx 0.5\%$). However, ballot comparison audits make the strong assumption that auditors can efficiently locate the ballot associated with a specific CVR.

Locating a particular ballot is a noted challenge in settings with voter-facing tabulators, which allow voters to cast ballots directly into the tabulation system.
Voter-facing tabulation improves both the chain of custody of ballots and the forensic record of the election, as it requires no storage of untabulated ballots. However, such tabulators intentionally randomize ballot order (by dropping them into a large collection bin), which is crucial to protect voter privacy (see, e.g., Crimmins et al.~\cite{crimmins2024dvsorder}). (For this very same reason, rows in the CVR table generated by 
a voter-facing tabulator are intentionally exported in random order rather than the order in which ballots are cast.) Thus, there is no natural association between the CVRs in the table and the resulting physical order of ballots. Many modern tabulators can imprint identifiers on ballots as they are cast, which provides a ready method for associating a particular ballot with a particular CVR entry. These same voter privacy concerns dictate that such identifiers must be suitably randomized. 
These circumstances lead to two natural approaches to locate ballots associated with particular CVR entries:
\begin{itemize}
\item \textbf{Linear search.} In settings with imprinted identifiers, a linear search can locate a particular ballot. Of course, on average, locating a single ballot requires reading the identifiers of half of the batch. 
\item \textbf{Transitive tabulation.} Ballots can be retabulated by a process that preserves order. Physical position itself can then act as a method for locating ballots. If the (re-)tabulator can imprint (or read previously imprinted identifiers), the physical ordering is now consistent with a known ordering of identifiers and can be used to quickly locate ballots according to a requested identifier. This approach handles every ballot.
\end{itemize}

In these settings, sample complexity is not a comprehensive measure of efficiency. 
A full accounting must reflect additional factors, such as time to locate ballots, audit preparation, and the availability of appropriate equipment.
 As a result, jurisdictions with imprinting often use batch comparison and polling audits, which require significantly more ballot interpretations but without the need to locate a specific identifier. For example, 
conducting a batch comparison audit with a 1\% margin and 5\% risk limit in Connecticut requires over 500,000 interpretations according to \href{https://github.com/aeharrison815/Adaptive-RLA-Tools}{Harrison's auditing simulation tool}—orders of magnitude more than ballot comparison, which requires around 1,200 ballots according to McBurnett's \href{https://bcn.boulder.co.us/~neal/electionaudits/rlacalc.html}{\texttt{RLACalc}}.

We propose a new approach to ballot comparison audits in settings with voter-facing tabulators and imprinted ballots: \emph{direct ballot selection}. 
It avoids linear searches through batches or retabulation, instead opting for the random ballot selection used in ballot polling. The high-level principle reverses the indexing of a standard ballot comparison audit:
\begin{enumerate}
\itemsep0em
\item Randomly select a physical ballot (e.g., using \emph{k}-cut~\cite{sridhar2020k}),
\item Read its imprinted identifier,
\item Compare against the corresponding CVR.
\end{enumerate}
As the basic procedure calls simply for drawing a ballot uniformly, this avoids the operational challenges associated with locating a \emph{particular, distinguished} ballot in a large collection. An immediate concern about this high-level approach is that \emph{identifier duplication} can undermine risk guarantees. Fuller, Harrison, and Russell showed that duplicated identifiers do not impact risk if one samples CVR rows uniformly~\cite{harrison2022adaptive}. If several physical ballots share one identifier, however, they can all correspond to a single CVR row, leaving other CVR rows free to carry identifiers that appear on no physical ballot. Because these phantom identifiers can never be sampled via direct selection, an adversary may use them to change the reported outcome without detection when duplicates are sufficiently frequent relative to the margin.

Thus, direct ballot selection must introduce additional statistical testing in order to deliver a rigorous risk limit. We establish that if duplicate (or blank) identifiers are sufficiently rare, this ``reverse'' indexing yields an RLA; motivated by this, we develop a simple statistical test for duplicate identifier detection.
Ultimately, the audit calls for selecting $k$ ballots uniformly at random and invokes two statistical tests: 
\begin{itemize}
\item\textbf{Comparison.} Ballots are compared against the corresponding CVR to establish consistency;
\item\textbf{Uniqueness testing.} Sampled ballot identifiers are checked for blank and repeated values.
\end{itemize}
We provide a detailed analysis establishing explicit risk limits for the procedure (as functions of margin and sample complexity). The total sample complexity of the procedure is asymptotically
\[
  O\left( \max \left(\frac{\log(1/\alpha)}{\mu},\sqrt{\frac{N\log(1/\alpha)}{\mu}}\right)\right)\,,
\]
where $N$ is the size of the total population of ballots. The two terms here correspond to the two statistical tests (which can be run in parallel with each other, using the same ballot samples). The first term is exactly the asymptotic sample complexity of a conventional comparison audit. The second term---reflecting the sample complexity of ensuring sufficiently few duplicate identifiers---involves the total population size, a factor that does not appear in any of the conventional audit approaches. The favorable dependence on $\sqrt{N}$ (for uniqueness testing) arises from the classical ``birthday paradox,'' which asserts that the probability of observing no collisions among $k$ objects sampled with replacement is $\exp(-\Theta(k^2)/N)$, scaling quadratically in $k$. Compared against polling, these asymptotics indicate that direct ballot selection audits perform advantageously for total ballot populations $N=O(\log(1/\alpha)/\mu^3)$.

While such asymptotic statements are useful for intuition and understanding large-scale behavior, an evaluation suitable for practical RLAs demands explicit numerics for concrete parameters. 
Since polling and direct ballot selection require the same setup and conventions for sampling random ballots from prescribed batches, one can decide which method to use after computing the margin $\mu$ for any races to be audited.

Finally, we remark that by externalizing duplicate (collision) detection, the analysis of direct ballot selection can be framed with the familiar notion of ``discrepancy.'' In particular, the extension of such audits to more complex electoral contexts---e.g., contests involving multiple races, multiple candidates, or alternative voting rules---follows directly from existing techniques (see, e.g., Stark~\cite{stark2010super}) that suitably (re-)define discrepancy for these settings.

\myparab{Summary of Results}
We focus on six ballot-population sizes motivated by the US electoral landscape: 354,000 ballots, representing the average number of ballots cast in a Congressional district in 2024; $1$M ballots; and four state-level populations corresponding to ballots cast  
in 2024: $1.7$M (Connecticut), $5$M (Georgia), $11$M (Florida), and $16$M (California).
For the congressional-sized districts, direct ballot selection improves over the first round of Minerva~\cite{zagorski2021minerva}/Providence~\cite{broadrick2023providence} polling audits for margins from $0.5\%$ to $2.5\%$, requiring over $77\%$ fewer sampled ballots for margins from $0.5\%$ to $1\%$. 
Sub-1\% margins are where the $1/\mu^2$ asymptotic behavior of polling audits drastically increases the sample size. This expensive scaling, despite the relative infrequency of such margins, is what pushes many election officials to prefer 
comparison audits over polling.

We also show that statistically certified manifests with 
small aggregate error relative to the 
margin suffice for polling, ballot comparison, and direct ballot selection audits. 
These manifests can dramatically reduce total audit time by decreasing the total number of ballots handled in audit preparation. Under our effort model, manifest creation accounts for 95\%--99\% of the total audit time 
in large states at moderate margins. For example, 
in California at a 3\% margin, adopting a statistically certified manifest 
reduces the modeled 
time to conduct comparison and polling audits 
by factors of approximately $438$ and $27$, respectively. Gains are more modest for direct ballot selection, because the dependence on $N$ is reduced rather than eliminated; the resulting speedup is a factor of 10.

Our techniques can be immediately incorporated into stratified RLAs~\cite{higgins2011sharper,ottoboni2019bernoulli}.
For example, at small margins, one can 
conduct a direct ballot selection audit for ballots with out-of-order identifiers and ballot comparison for ballots with 
sequential identifiers (e.g., cast by mail). For large margins, one can use ballot polling.

\subsection{Related Work} 
\label{ssec:rel work}
Risk-limiting audits were first
articulated by Stark in 2008~\cite{starkconservative}, 
followed by seminal work by Stark and Lindeman~\cite{lindeman2012gentle}. Following this, a body of
work laid down the foundations, including key assumptions and
guarantees~\cite{bernhard2021risk,verifiedvotingprinciples,Hall2009,lindeman2012gentle}. There has been a long line of work on optimizing practical efficiency~\cite{lindeman2012gentle,stark2010super,lindeman2012bravo,starkconservative,checkoway,starkcast}. A significant body of literature has also developed around various generalizations and refinements, including
\begin{enumerate*}\item choice functions beyond plurality~\cite{stark2020sets,blom2021assertion}, \item
combining $p$-values from multiple 
jurisdictions~\cite{stark2009auditing,stark2010super,higgins2011sharper,ottoboni2018risk,stark2009efficient}, \item reducing sample complexity~\cite{higgins2011sharper,lindeman2012bravo,ottoboni2019bernoulli,waudby2021rilacs,Stark:Conservative,starkconservative,checkoway,banuelos2012limiting}
and \item addressing practical implementation challenges~\cite{verifiedvotingprinciples,Hall2009,bernhard2021risk,harrison2022adaptive,fuller2024decisive,jones2024scan}.\end{enumerate*} 

\myparab{Organization} The remainder of this article is organized as follows: Section~\ref{sec:preliminaries} lays out the modeling and main definitions; Section~\ref{sec:auditor} introduces our auditor and proves it is risk limiting; and Section~\ref{sec:efficiency} presents our simulation results for ballots examined, 
including the time required to conduct common RLA methods.

\section{Preliminaries}
\label{sec:preliminaries}
 For simplicity, we consider an audit of a single plurality race with two candidates, denoted $\winner$ and $\loser$.  By our naming convention, the candidate $\winner$ is \textbf{reported} to have received more votes in the race.
Throughout, we use boldface to refer to ``physical'' objects, such as
individual ballots (typically denoted $\mathbf{b}$) or the set of all ballots (typically $\mathbf{B}$). Variables determined by these physical objects are typically denoted with a super- or subscript ($X^{\mathbf{b}}$), indicating
their dependence on the physical object.

\ifnum\conference=0
 \begin{table}[t]
 \centering
 \small
 \begin{tabular}{l | l | r} 
 & Notation & Description\\\hline
 \multirow{12}{*}{concepts}& $\size$ & size\\
 & $\winner$ &  tabulated winner\\
 & $\loser$ &  tabulated loser \\
 & $\mathbf{b}, \mathbf{B}$ & ballot, set of ballots\\
 & $\disc$ & discrepancy\\
 & $\mu$ & diluted margin\\
 & $\cvr$ & Cast vote record table\\
 & $\manifest$ & ballot manifest\\
 & $\alpha$ & risk limit\\
 & $k$ & samples\\
 & $\batches$ & number of batches\\
 & $\inAcc$& large inaccuracy bound\\
 & $\inAccSmall$& small inaccuracy bound\\
 \hline
 \multirow{6}{*}{modifiers} & $\act$& on ballots\\
 & $\tab$ & in tabulation results\\
 & $\cvr$ & in CVR\\
 & $\size$ & size bounding\\
 & $\sample$ & audit\\
 & $\dup$ & duplicate detection\\
 & $\batchnum$ & batch index
 \end{tabular}
 \caption{Summary of Notation.}  
 \label{tab:notation}
 \end{table}
\fi

We define $\N = \{0, 1, \ldots\}$ to be the natural numbers (including
zero). For a natural number $k$, we define $[k] = \{1, \ldots, k\}$
(and $[0] = \emptyset$). We let $\Sigma = \{-2, -1, 0, 1, 2\}$, a set
that will play a special role in our setting.
In general, for a
finite set $X$, we define $X^*$ to be the set of all finite-length
sequences over $X$; that is,
$X^* = \{ (x_1, \ldots, x_k) \mid k \geq 0, x_i \in X \}$. Note that
this includes a sequence of length $0$ which we denote
$\bot$. Finally, we define $X^\N$ to be the set of all sequences
$\{ (x_0, x_1, \ldots) \mid x_i \in X\}$.

\subsection{Election Definitions}

We now set down the definitions of elections, manifests,
and CVRs. We build on the notation and formal modeling of~\cite{harrison2022adaptive}, though certain aspects are adapted to account for the details of our setting.

\begin{definition}[Ballot family; ballot conventions]
  \label{def:ballot-family}
  A \emph{ballot family} is a collection of physical \emph{ballots}
  partitioned into disjoint sets denoted
  $\Ballots_1, \ldots, \Ballots_\batches$. As a matter of notation, the
  ballot family is denoted
  $\Ballots = (\Ballots_1, \ldots, \Ballots_\batches)$ and the sets are
  referred to as ``batches.'' For the sake of brevity, we use
  $\ballot \in \Ballots$ as shorthand for
  $\ballot \in \bigcup \Ballots_\batchnum$ and use $|\Ballots|$
  as shorthand for $\sum |\mathbf{B}_\batchnum|$.  Throughout, we
  reserve the variable $\batches$ to refer to the number of batches.
  Physical ballots, typically denoted $\ballot$, have the following properties:
  \begin{enumerate}
  \item There is an immutable interpretation of the votes contained on
    the ballot:  Each $\ballot \in \Ballots$
    determines a pair $(\winner_{\mathbf{b}}, \loser_{\mathbf{b}})$,
    where each $\winner_{\mathbf{b}}, \loser_{\mathbf{b}}\in \{0,1\}$.
  \item Each ballot $\ballot \in \Ballots$ is associated with
    an indelible \emph{identifier} $\id_\ballot \in
    \{0,1\}^*$. If $\id_\ballot= \bot$ (the empty string in $\{0,1\}^*$), we say that $\ballot$ is \emph{unlabeled}; otherwise $\mathbf{b}$ is \emph{labeled}. 
    $\Labels_{\mathbf{B}} \subset \{0,1\}^*$ denotes the set of all identifiers appearing on ballots in $\Ballots $.
      \item Each ballot $\ballot$ is associated with an indelible batch $\batchnum$ for $1\le \batchnum \le \batches$. We let $\batch(\ballot)$ denote the batch number, so that $\ballot \in \Ballots_{\batch(\ballot)}$.
      \end{enumerate}
It is convenient to discuss physical ballots with such immutable votes ($\winner_{\ballot},\loser_{\ballot}$) before they are partitioned into batches or assigned labels; we refer to this ungrouped collection as ``raw ballots.''
\end{definition}

\begin{definition}[Cast-Vote Record Table (CVR)] A \emph{Cast-Vote Record Table (CVR)} is a sequence of triples
$    \cvr = ((\iota_1, \winner_1, \loser_1),$ $\ldots,$ $(\iota_s,
    \winner_s, \loser_s))
  $ where $\iota_r$ are distinct bitstrings in $\{0,1\}^* \setminus \{\bot\}$ and each
  $\winner_r, \loser_r$ is an element of $\{0,1\}$. We use the following language:
  \begin{enumerate}
 \itemsep0em
  \item The elements
    $\iota_r$ are \emph{identifiers}.
  \item The number $s$ is 
  the \emph{size} of the CVR.
  \item The \emph{$r$th row}  is the triple
  $\cvr_r = (\iota_r, \winner_r, \loser_r)$.
\item  We use $r_\iota$ to refer to the (unique) row with identifier $\iota$.\footnote{An auditor can use the identifiers $\bot_i$ to transform a CVR to one with unique labels; such labels would not appear on CVRs generated by tabulators. See discussion in Fuller, Harrison, and Russell on CVR transforms~\cite{harrison2022adaptive}.} (Recall that CVR identifiers are distinct.)
  \end{enumerate}
  \label{def:cvr format}
  \end{definition}

\paragraph{Remark} In the Introduction, we used CVR as an abbreviation for cast vote record. We now switch to using CVR as an abbreviation for the table. The term CVR is used to refer to both in RLA literature. 

  \begin{definition}[Tabulation; Election]
    We define a $\batches$-\emph{tabulation} to be a tuple $T = (\cvr^{(1)}, \ldots, \cvr^{(\batches)})$ of CVRs with globally unique identifiers. Thus, a tabulation declares the results of a ballot family organized into $\batches$ batches, with $\cvr^{(\beta)}$ providing a tabulation of those ballots in batch $\beta$. We simply use the word \emph{tabulation} when $\batches$ can be inferred from context. The CVR for batch $\batchnum$ is $\cvr^{(\beta)} = ((\iota_1^{\beta}, \winner_1^{\beta}, \loser_1^{\beta}), \ldots, (\iota_{s_\beta}^{\beta}, \winner_{s_\beta}^{\beta}, \loser_{s_\beta}^{\beta}))$. Such a tabulation implicitly determines, for each $1 \leq \beta \leq \batches$ the batchwise quantities
  \[
    \size^\tab_\beta = s_\beta\,,\quad \winner^{\tab}_\beta = \sum_j \winner^{\beta}_j\,,\quad \loser^\tab_\beta = \sum_j \loser^{\beta}_j
  \]
  and the aggregate quantities
  \[
    \size^\tab = \sum_\beta \size^{\tab}_\beta, \quad \winner^\tab = \sum_\beta \winner^{\tab}_\beta, \quad 
    \loser^\tab = \sum_{\beta} \loser^{\tab}_\beta\,.
  \]
  The \emph{tabulated totals} refer to the pair $\winner^{\tab}$ and $\loser^{\tab}$; the \emph{tabulated size} refers to $\size^{\tab}$. An \emph{election} $E$ is a pair $E = (\mathbf{B}, T)$ where $\mathbf{B}$ is a ballot family organized into $\batches$ batches and $T$ is a $\batches$-tabulation.

  We assume throughout the convention that $\size^{\tab} \geq \winner^\tab > \loser^\tab \geq 0$, so that the names $\winner$ and $\loser$ are reserved for the tabulated winner and loser of the election, respectively.

  We remark that while we insist on unique identifiers across the CVRs in a tabulation, this is merely a convenience; the only requirement is that identifiers be unique within a batch. By replacing every CVR identifier $\iota$ in batch $\beta$, and every corresponding nonempty ballot identifier, with an injective encoding of $(\beta,\iota)$, a tabulation can be transformed to one with globally distinct identifiers. We leave $\bot$ unchanged and adopt this convention throughout; consequently, looking up a ballot's identifier in $T$ is equivalent to looking it up in the CVR for that ballot's batch.
\end{definition}

\begin{definition}[Actual vote totals; ballot manifests]
  Let $E = (\mathbf{B}, T)$ be an election.
  Let
  $
  (\size^\act_1; \winner^\act_1, \loser^\act_1), \ldots, (\size^\act_\batches; \winner^\act_\batches, \loser^\act_\batches)
  $
  denote the actual size and vote totals where, for each $\beta$, $\size^\act_\beta = |\mathbf{B}_\beta|$ is the
  actual size of the set $\mathbf{B}_\beta$ and
  \[
    \winner^\act_\beta = \sum_{\mathbf{b} \in \mathbf{B}_\beta} \winner_{\mathbf{b}} \qquad\text{and}\qquad \loser^\act_\beta = \sum_{\mathbf{b} \in \mathbf{B}_\beta} \loser_{\mathbf{b}}
  \]
  are the total number of actual votes received by the candidate $\winner$ and candidate $\loser$ over the ballots in batch $\beta$. As above, we define the aggregate quantities
  \[
    \size^\act = \sum_\beta \size^{\act}_\beta, \quad \winner^\act = \sum_\beta \winner^{\act}_\beta, \quad \text{and} \quad
    \loser^\act = \sum_{\beta} \loser^{\act}_\beta\,.
  \]
  We assume throughout that $\size^\act>0$ and omit from the batch index any batch for which $\size_\beta^\act=\size_\beta^\tab=0$.
  The values $\manifest^\act = (\size_1^\act, \ldots, \size_\batches^\act)$ are referred to as the \emph{ballot manifest} of $E$. We occasionally refer to $\manifest^\tab = (\size_1^\tab, \ldots, \size_\batches^\tab)$ as the declared sizes of the batches. Finally, we define $|T| = \size^{\tab}$.
\end{definition}

\begin{definition}[Diluted margin; valid and invalid elections] The 
  \emph{tabulated diluted margin} of an election $E$ is
    \[\mu^\tab= (\winner^\tab - \loser^\tab)/\size^\tab\,.\]
  An election $E$ is \emph{invalid} if the tabulated winner is
  incorrect: $\loser^{\act} \geq \winner^{\act}$; otherwise, we say
  that $E$ is \emph{valid}.
  The \emph{actual diluted margin} is defined to be
$
    \mu^\act = | \winner^\act-\loser^\act |/\size^\act\,,
  $ which is determined only by the physical ballots.
\end{definition}

\begin{definition}[Discrepancy] Let\label{def:overall disc}
  $E = (\mathbf{B}, T)$ be an election. The \emph{discrepancy} of $E$ is 
  \[\disc = (\winner^\tab - \loser^{\tab}) - \left(\winner^\act - \loser^{\act}\right)\,.\]
\end{definition}
\noindent
For invalid elections $\loser^\act\ge \winner^\act$ and thus $\mu^\act = -(\winner^\act- \loser^\act)/\size^\act$. In this case
\ifnum\conference=1
$\disc = (\winner^\tab - \loser^{\tab}) - (\winner^\act - \loser^{\act}) = \mu^{\tab} \cdot \size^\tab+ \mu^\act \cdot \size^\act.$
\else
\[\disc = (\winner^\tab - \loser^{\tab}) - (\winner^\act - \loser^{\act}) = \mu^{\tab} \cdot \size^\tab+ \mu^\act \cdot \size^\act.\]
\fi

\begin{definition}
  [Ballot Discrepancy]
  \label{def:discrepancy}
  Let $(\mathbf{B},T)$ be an election, let $\ballot \in \mathbf{B}_\beta$ be a ballot, and let $\cvr$ be the CVR in $T$ for batch $\beta$. Then the 
  \emph{discrepancy} $\disc_T(\ballot)$ of the ballot $\ballot$ (with respect to $T$) is defined to be the value
  \begingroup
  \ifnum\conference=0
    \small
  \fi
  \begin{equation}
    \label{eq:row-disc}
    \disc_T(\ballot) = \begin{cases}
      \begin{aligned}
        &(\winner_{r} - \loser_{r}) -
          (\winner_\ballot - \loser_{\ballot} ),\\[-0.25ex]
        &\quad \text{if row $r$ of $\cvr$ has $\iota_r = \id_\ballot$};
      \end{aligned}\\[1ex]
      \begin{aligned}
        &1 - (\winner_\ballot - \loser_{\ballot} ),\\[-0.25ex]
        &\quad \text{otherwise}.
      \end{aligned}
    \end{cases}
  \end{equation}
  \endgroup
\end{definition}

A positive value for discrepancy $d$ is called a
\emph{$d$-vote overstatement}; likewise, a negative value of $-d$ is a \emph{$d$-vote understatement}. A $d$-vote overstatement for a ballot $\ballot$, for example, indicates that
the reported difference, $\winner_r - \loser_r$, is $d$ votes
 larger than the ground truth value $\winner_{\ballot} - \loser_{\ballot}$. Observe that 
Equation~\eqref{eq:row-disc} assigns a notion of discrepancy to a
particular ballot, which always takes a value in the set
$\Sigma = \{ -2, -1, 0, 1, 2\}$.  Note that there may be multiple ballots with the same 
identifier but yielding different values of discrepancy.

\myparab{Statistical Tests}
\label{ssec:statistical tests} A standard approach for designing RLAs
is to consider the discrepancy $\disc_T(\ballot)$ of a uniformly selected \emph{row} $r$ from the CVRs in a tabulation $T$
in comparison with some ballot $\ballot$ with the drawn identifier $r_\iota$. In this conventional setting, if the election is invalid, one has that
\[
\Exp_{\substack{\textrm{$\ballot$ matching}\\\textrm{random row}}}[\disc_T(\ballot)] \ge \frac{\disc}{\size^\tab} \ge \mu^\tab +\mu^\act\frac{\size^\act}{\size^\tab} \ge \mu^\tab\,.
\]
(One needs to handle labels $\iota$ that don't appear on any ballot~\cite{harrison2022adaptive}; this is straightforward by treating the missing ballot as the worst interpretation for the audit.)
Independently carrying out such observations results in
a sequence of discrepancy observations $\disc_1^\obs, \disc_2^\obs,\ldots$ taking values in
$ \{-2, \ldots,2\}$.  An RLA can then be given by a one-sided statistical test
for the null hypothesis that \emph{$E$ is invalid and hence} \[\Exp[\disc_i^\obs] \geq \frac{\disc}{\size^\tab} \geq \mu^\tab\] (which is guaranteed when the election is invalid); the alternative hypothesis is that the election is indeed correct. In this language, the risk of the audit is the probability that the null hypothesis is rejected when it is in fact true.

Observe that, in contrast, the quantity $\Exp_{\ballot \leftarrow \Ballots} [\disc_T(\ballot)]$ arising from the natural ``discrepancy of a randomly selected ballot'' experiment is not necessarily greater than, say, $\disc/\size^{\tab}$. One goal of the full analysis below is to establish a sufficient connection between these quantities (and hence to $\mu^\tab$). Ultimately, as mentioned above, we will be able to rely on standard one-sided statistical tests, which we define and discuss below.

\begin{definition}[$\sigma$-dominating distributions and random variables]
  A sequence of bounded (real-valued) random variables $X_1, \ldots$
  is said to be \emph{$\sigma$-dominating} if, for each $t \geq 1$,
  $\Exp[X_t \mid X_1, \ldots, X_{t-1}] \geq \sigma$ almost surely.
  We also use this terminology to apply to the distribution $\mathcal{D}$ corresponding to the random variables, writing $\sigma \unlhd \mathcal{D}$.
\end{definition}

\begin{definition}[Stopping rules and stopping times] Let $\Sigma = \{-2, -1, 0, 1,
  2\}$. A \emph{stopping rule} is a function
  $\Stop: \Sigma^* \rightarrow \{0,1\}$.    
  For a sequence of random variables $X_1, \ldots$ taking values in
  $\Sigma$, let
  \[
    \FirstStop(X_1,\ldots)
    :=\inf\{t\geq 1:\Stop(X_1,\ldots,X_t)=1\},
  \]
  with $\inf\emptyset=\infty$. This is the stopping time induced by
  $\Stop$. If $\FirstStop=\infty$, the test is deemed not to reject.
  \label{def:stop time}
\end{definition}

\begin{definition}[Adaptive Audit Test~\cite{harrison2022adaptive}]\label{def:audit-test}
  An \emph{adaptive audit test}, denoted $\Test = (\Stop, \Criterion)$, is
  described by two families of functions, $\Stopdeltaalp$ and $\Criteriondeltaalp$. For
  each $0 < \alpha \le 1$, $0 < \sigma \leq 2$,
  \begin{enumerate}
  \item $\Stopdeltaalp$ is a stopping rule, as in Definition~\ref{def:stop time}, and
  \item $\Criteriondeltaalp:\Sigma^* \rightarrow \{0,1\}$ is
    the \emph{rejection criterion}. 
  \end{enumerate}
  Let $\mathcal{D}$ be a probability distribution on $\Sigma^\N$; for
  such a distribution, define
  $\alpha_{\sigma,\alpha}[\mathcal{D}] =
  \Pr[\tau<\infty\text{ and }\Criteriondeltaalp(X_1,\ldots,X_\tau)=1]$
  where $X_1, \ldots$ are random variables
  distributed according to $\mathcal{D}$ and $\tau$ is determined by
  $\Stopdeltaalp$.
The test $\Test$ is $\alpha$-risk-limiting if
  \begin{equation}
  \label{def:audit test}
  \sup_{\substack{0 < \sigma \le 2\\ \sigma \unlhd \mathcal{D}}} \alpha_{\sigma,\alpha}[\mathcal{D}] \leq \alpha\,,
  \end{equation}
  where this supremum is taken over all $\sigma \in (0,2]$ and over all probability distributions
  $\mathcal{D}$ for which $\sigma \unlhd \mathcal{D}$.
\end{definition}

The Kaplan-Markov test from the ``super simple'' ballot comparison audit~\cite{stark2009efficient,stark2009auditing,stark2010super,Stark:Conservative} is an adaptive audit test~\cite[Claim 3]{harrison2022adaptive} which is used in all of our ballot comparison experiments.

Finally, we define two notions of closeness between probability distributions that will be useful in analyzing the auditor.

\begin{definition}\label{def:mult-close}
  Let $A$ be a probability distribution on a finite set $X$ and $\delta \ge 0$. We say that $A$ is \emph{$\delta$-close to uniform} if for all $x \in X$,
  \[
    \frac{1}{|X|(1 + \delta)} \leq A(x) \leq \frac{1 + \delta}{|X|}\,.
  \]
\end{definition}

\begin{definition}[Contaminated reweighting]\label{def:contaminated-reweighting}
Let $A$ and $B$ be probability distributions on a finite set $X$, and let
$0 \le \delta \le \Delta$ and $p \in [0,1]$.
We say that $B$ is a \emph{$(p,\delta,\Delta)$-contaminated reweighting} of $A$
if there exist a (weight) function $w:X\to(0,\infty)$ and a set $\Bad\subseteq X$ such that
\begin{align*}
  B(x)&=\frac{A(x)w(x)}{\sum_{y\in X}A(y)w(y)}
        &&\text{for all }x\in X\,,\\
  \frac{1}{1+\delta}&\le w(x)\le 1+\delta
&&\text{for all }x\in X\setminus \Bad\,,\\
  \frac{1}{1+\Delta}&\le w(x)\le 1+\Delta
&&\text{for all }x\in X\,,
\end{align*}
and $A(\Bad)\le p$.
\end{definition}

When \(A\) and \(B\) are ballot- or batch-sampling distributions, we also say that \(B\) is a \((p,\delta,\Delta)\)-manifest of \(A\).
We call an assertion that \(B\) is a \((p,\delta,\Delta)\)-manifest of \(A\) an \emph{statistical manifest certificate}; the certificate's failure probability is specified separately.

\subsection{The Formal Auditor--Adversary Model}
\label{sec:adversary model}

We build on the formal auditing model of~\cite{harrison2022adaptive}, adapted to our setting to account for random ballot selection and statistical manifests. The model introduces an abstract party---the adversary $\Adversary$---responsible for labeling ballots, grouping them into batches, and producing tabulations (and CVRs). The model does not require the adversary to assign unique identifiers to ballots; moreover, it permits the adversary to form batches after observing the ballots. As a result, the adversary may induce highly irregular allocations of votes across batches. Once the adversary has chosen identifiers and formed batches, these assignments are permanent throughout the audit.

Then, a \emph{coarse manifest}
is provided to the Auditor. Concretely, this coarse manifest gives an approximate size
\[
  \manifest^\coarse=(\size^\coarse_1,\ldots,\size^\coarse_\batches)
\]
for the batches: it must satisfy, for every batch \(\beta\),
\begin{equation}\label{eq:coarse}
  \size^\coarse_\beta \in
  \left[
    \frac{\size^\act_\beta}{1+\inAcc_0},
    (1+\inAcc_0)\size^\act_\beta
  \right].
\end{equation}
This reflects the practical point emphasized in the introduction: rough batch-size information can often be efficiently obtained through bulk methods such as weighing or stack measurement. Within the constraints given by~\eqref{eq:coarse}, the coarse manifest may be chosen adversarially. The quantity \(\inAcc_0\) is determined from a parameter of the model $(\inAcc)$ by the rule $(1+\inAcc_0)^2 = 1+\inAcc$; see below.
In particular, the experimental setting \(\inAcc=0.1\) requires \(\inAcc_0=\sqrt{1.1}-1\approx0.0488\).

The model does place one important structural constraint on the tabulation
(CVRs) produced by the adversary: each identifier only occurs once in the CVRs. This reflects the practical fact that CVRs are available to
the auditor and can therefore be checked for duplicate identifiers. (Such duplicates can alternatively be handled with ``CVR transformations'' that do not increase risk~\cite{harrison2022adaptive}.)

Beyond this, however, the tabulation produced by the adversary is not assumed
to satisfy any accuracy conditions. Instead, the auditor may compare the batch
sizes induced by the tabulation,
$
  \manifest^\tab=(\size^\tab_1,\ldots,\size^\tab_\batches),
$
against the coarse manifest. In particular, if the auditor verifies that, for
every batch \(\beta\), $\size^\tab_\beta \in
[
    {\size^\coarse_\beta}/{(1+\inAcc_0)},
    (1+\inAcc_0)\size^\coarse_\beta]$,
then, considering that
$\size^\coarse_\beta \in
[{\size^\act_\beta}/{(1+\inAcc_0)},$ $
    (1+\inAcc_0)\size^\act_\beta]$, it follows that
\begin{align*}
  \size^\tab_\beta &\in
  \left[
    \frac{\size^\act_\beta}{(1+\inAcc_0)^2},
    (1+\inAcc_0)^2 \size^\act_\beta
  \right]\\
  &=
  \left[
    \frac{\size^\act_\beta}{1+\inAcc},
    (1+\inAcc)\size^\act_\beta
  \right].
\end{align*}
Thus, a successful comparison against the coarse manifest certifies that the
batch sizes arising from \(T\) constitute a \(\inAcc\)-accurate manifest.

Our final auditing procedure provides guarantees against \emph{any adversary} satisfying these conditions. Because the adversary represents exactly those aspects of the election that are not assumed trustworthy or fixed in advance, the resulting guarantee applies in any practical setting adequately reflected by the model.

\myparab{The Auditor--Adversary Game}
The
\emph{Auditor}, denoted by $\Auditor$, and the \emph{Adversary}, denoted
by $\Adversary$, together interact with a family of ballots reflecting the results of the election to be audited. Figure~\ref{fig:auditing game} formalizes the game.
An important feature of the game is the mechanism by which the auditor obtains random ballot samples. As discussed above, statistical audits require sampling accuracy commensurate with the margin of the contest under audit. While the model allows the auditor to sample uniformly within any chosen batch, it does not directly allow sampling batches with probability proportional to their true sizes. Instead, the auditor samples batches according to their tabulated sizes, which may induce a substantially nonuniform distribution over physical ballots; the coarse manifest is used to check those sizes. The auditor is also given a mechanism for determining the true size of any batch on request.

A straightforward solution would be to determine the exact sizes of all batches, and thereby recover the correct size-proportional distribution; ballots could then be sampled accordingly. The practical obstacle is that obtaining the exact size of a batch has cost linear in that batch's size. 
Therefore, a key challenge is to limit the total size of the batches whose true counts are requested while still ensuring that the resulting sampling distribution is sufficiently close to uniform to support the desired risk bound.

\begin{figure}[th!]
  \begin{framed}
  Auditor ($\Auditor$)--Adversary ($\Adversary$) game for the raw ballots $\mathbf{B}$.
    \begin{enumerate}[noitemsep]
    \item \textbf{Formation of a Full Ballot Collection.} The adversary $\Adversary$ may inspect the raw physical ballots $\Ballots$ and determines the following:
      \begin{enumerate}[noitemsep]
      \item\textbf{Batches.}  A partition of the ballots into $\batches$ batches, for a parameter $m$ determined by $\Adversary$.
      \item\textbf{Labels.} A label assignment for each ballot. (The adversary may choose to leave some ballots
      unlabeled, that is, assign the label $\bot$.)
    \end{enumerate}
    \item \textbf{Coarse Manifest}. $\Adversary$ generates a coarse manifest $\manifest^\coarse = (\size^\coarse_1,\ldots,\size^\coarse_\batches)$ for which
    \[
       \size^\coarse_\batchnum \in \left[{\size^\act_\batchnum}/{(1+\inAcc_0)}, (1+\inAcc_0)\size^\act_{\batchnum}\right]
     \]
     (where $\inAcc_0$ satisfies $(1 + \inAcc_0)^2 = 1 + \inAcc$).\\
     $\manifest^\coarse$ is provided to $\Auditor$.
  \item \textbf{Tabulation}. $\Adversary$ generates a tabulation $T = (\cvr^{(1)}, \ldots, \cvr^{(\batches)})$.
    $T$ is provided to $\Auditor$.
  \item\textbf{The Audit.}
      \begin{enumerate}[noitemsep]
        \item \textbf{Count Batch}. $\Auditor$ repeatedly requests a batch $\beta$ ($1\le \beta\le \batches$) and is provided with $\size^\act_\beta = |\Ballots_\beta|$.
        \item \textbf{Sample Ballot}. $\Auditor$ repeatedly requests ballots to be drawn uniformly from selected batches: specifically, $\Auditor$ issues a batch number $\beta$ and is provided with a ballot $\ballot$ sampled uniformly from ${\mathbf{B}}_\beta$ (with replacement).
        \end{enumerate}
  \item \textbf{Conclusion}. $\Auditor$ returns one of the two values:
$\consistent$ or $\inconclusive$.
\end{enumerate}
\end{framed}
\caption{The $\RLA^{\inAcc}_{\Auditor,\Adversary}(\mathbf{B})$ auditing game with parameter $\inAcc$.}
\label{fig:auditing game}
\end{figure}

\begin{definition}[Risk; soundness]
For an auditor $\Auditor$, an adversary $\Adversary$, and a collection
of raw ballots $\mathbf{B}$, let
$\RLA^{\inAcc}_{\Auditor,\Adversary}(\mathbf{B})$
denote the conclusion of the audit described in
Figure~\ref{fig:auditing game}. The adversary's choices determine an
election $E_{\Adversary}(\mathbf{B})$ and a coarse manifest before the
audit begins.

An auditor $\Auditor$ has \emph{$\risk$-risk} (or
\emph{$\risk$-soundness}) if, for every collection of raw ballots
$\mathbf{B}$ and every adversary $\Adversary$ satisfying the conditions of the game for which
$E_{\Adversary}(\mathbf{B})$ is invalid,
\label{def:risk}
\[
  \Pr\!\left[
    \RLA^{\inAcc}_{\Auditor,\Adversary}(\mathbf{B})=\consistent
  \right]
  \leq \risk.
\]
The probability is over the random choices made during the audit; all
choices determining the election and coarse manifest are fixed before
the audit begins.
\end{definition}

\section{The Direct Ballot Selection Auditor}
\label{sec:auditor}
\noindent
The direct ballot selection RLA is designed to operate for any
tabulated diluted margin
$\mu \;:=\; (\winner^\tab-\loser^\tab)/\size^\tab,
$
which is computed from the tabulation provided to the auditor.
The auditor is parameterized by the tuple $(\inAccSmall,\rho_{\mathrm{tv}},\rho_{\dup},\alpha,\alpha_{\mathrm{tv}},\alpha_{\dup})$,
where:
\begin{enumerate}
\itemsep0em
\item \(\inAccSmall\in[0,\inAcc]\) is the sharper multiplicative manifest-accuracy
  threshold that the auditor seeks to certify for all but a small exceptional set of
  batches;
\item \(\rho_{\mathrm{tv}}\in[0,1]\) is the fraction of \(\mu\) reserved to
  absorb the loss arising from nonuniform ballot sampling caused by residual
  manifest inaccuracy;
\item \(\rho_{\dup}\in[0,1]\) is the fraction of \(\mu\) reserved to absorb the
  loss arising from duplicate or unlabeled identifiers;
\item \(\alpha\) is the overall risk limit;
\item \(\alpha_{\mathrm{tv}}\) is the portion of the risk budget allocated to
  manifest certification; and
\item \(\alpha_{\dup}\) is the portion of the risk budget allocated to duplicate
  detection.
\end{enumerate}
We assume throughout that
$0\le \inAccSmall \le \inAcc$, $\rho_{\mathrm{tv}}+\rho_{\dup}<1$, and $\alpha_{\dup}+\alpha_{\mathrm{tv}}<\alpha$. 
We further assume that \(\rho_{\dup}>0\), \(\alpha_{\dup}>0\), and \(\alpha_{\mathrm{tv}}>0\).

\myparab{Structure of Auditor}
The auditor first uses the coarse manifest to rule out gross inconsistencies between the tabulated and true batch sizes. Once that check succeeds, the manifest-certification step samples batches according to 
\(r^\tab_\beta := {\size^\tab_\beta}/{\size^\tab}\)
and counts each sampled batch to certify a manifest relation between the tabulation-induced sampling law and the uniform ballot law. Concretely, if the accepted certificate yields \((p_{\mathrm{tv}},\inAccSmall,\inAcc)\), then all but a declared mass \(p_{\mathrm{tv}}\) of batches satisfy the sharper \(\inAccSmall\)-criterion, while the full manifest remains within the coarser \(\inAcc\)-bound. The duplicate-detection stage then uses the same certificate to control the distortion between \(\sample\) and uniform ballot sampling.

The auditor next performs duplicate detection.  Because the applied ballot-sampling rule
is not exactly uniform, the duplicate-detection routine is calibrated
using the full triple \((p_{\mathrm{tv}},\inAccSmall,\inAcc)\) describing the distributional distortion. It then guarantees that except with small failure probability $\alpha_{\dup}$, the excess identifier multiplicity rate $\kappa$ is no more than $\kappa_\dup$. (The additional convention that any sampled unlabeled ballot causes the ``Collision'' outcome is conservative and affects completeness only.)

Although presented as separate stages for clarity, duplicate detection and ballot comparison may use the same ballot samples, so only \(\max\{k_{\dup},k_{\sample}\}\) ballots need be drawn; the union-bound analysis does not require the two tests to be independent.

If either preliminary stage fails, the audit returns \(\inconclusive\).
Otherwise, the remaining risk budget \(
\alpha_{\sample}:=\alpha-\alpha_{\dup}-\alpha_{\mathrm{tv}}
\)
is allocated to the adaptive audit test, which is run at effective margin
\(
\mu_{\sample}:=(1-\rho_{\mathrm{tv}}-\rho_{\dup})\mu\).
The conclusion of the procedure is either \(\consistent\) or
\(\inconclusive\). The formal description appears in
Figures~\ref{fig:adaptive-auditor} and~\ref{fig:adaptive-auditor-subroutines}.

\begin{figure}[t!]
\begin{framed}
\underline{Auditor \(\Auditor[(\Stop,\Criterion)](\inAccSmall,\rho_{\mathrm{tv}},\rho_{\dup},\alpha,\alpha_{\mathrm{tv}},\alpha_{\dup})\) for  \(E\)}
\begin{enumerate}[noitemsep]
\item Receive the coarse manifest
  \(\manifest^\coarse=(\size^\coarse_1,\ldots,\size^\coarse_\batches)\) and
  tabulation \(T=(\cvr^{(1)},\ldots,\cvr^{(\batches)})\).

\item Let
  $\manifest^\tab := (\size^\tab_1,\ldots,\size^\tab_\batches)$
  be the batch sizes induced by \(T\), and let
  \(\winner^\tab,\loser^\tab,\size^\tab\) denote the corresponding tabulated
  totals.  Define
$\mu := {(\winner^\tab-\loser^\tab)}/{\size^\tab}$.

\item Run
  \(\CheckManifest(\manifest^\coarse,\manifest^\tab)\).
  If it returns \(\mathtt{Error}\), return \(\inconclusive\).

\item Define $\alpha_{\sample} := \alpha-\alpha_{\dup}-\alpha_{\mathrm{tv}}$,     
    $\mu_{\sample} := (1-\rho_{\mathrm{tv}}-\rho_{\dup})\mu$, and
  $
    \kappa_{\dup} := {\rho_{\dup}\mu}/{2}$.
  
\item Compute
$    (p_{\mathrm{tv}},k_{\mathrm{tv}})
\leftarrow
\Psi(\mu,\rho_{\mathrm{tv}},\alpha_{\mathrm{tv}},\inAcc,\inAccSmall).
$
  If \(p_{\mathrm{tv}}\le 0\), return \(\inconclusive\). Define $\epsilon = (p_{\mathrm{tv}}\inAcc + (1-p_{\mathrm{tv}})\inAccSmall)$.

  \item Run
  \(\BoundSize(\manifest^\tab,k_{\mathrm{tv}},\inAccSmall)\).
  If it returns \(\mathtt{Error}\), return \(\inconclusive\).

\item Let \(\sample\) be the ballot-sampling rule obtained by:
  \begin{enumerate}[noitemsep]
  \item sampling batch \(j\) with probability \(\size^\tab_j/\size^\tab\), and
  \item sampling a ballot uniformly from batch \(j\).
  \end{enumerate}

\item Run
$    \DetectDuplicates(\sample,p_{\mathrm{tv}},\inAccSmall,\inAcc,\kappa_{\dup},$ $\alpha_{\dup},(1+\epsilon)\size^\tab).
  $ If it returns \(\mathtt{Collision}\), return \(\inconclusive\).

\item Initialize \(\iter=0\). Repeat:
  \begin{enumerate}[leftmargin=.5cm,nosep]
  \item Increment \(\iter := \iter+1\).
  \item \(\disc_{\iter}:=\BasicExperiment(\sample,T)\).
  \end{enumerate}
  until $\Stop_{\mu_{\sample},\alpha_{\sample}}(\disc_1,\ldots,\disc_{\iter})=1$.

\item If
  $\Criterion_{\mu_{\sample},\alpha_{\sample}}(\disc_1,\ldots,\disc_{\iter})=1$,
  
  return \(\consistent\); otherwise return \(\inconclusive\).
\end{enumerate}
\end{framed}
\caption{The auditor \(\Auditor[(\Stop,\Criterion)]\), parameterized by the
sharp manifest threshold \(\inAccSmall\), margin fractions
\(\rho_{\mathrm{tv}},\rho_{\dup}\), and risk allocations
\(\alpha_{\mathrm{tv}},\alpha_{\dup}\).}
\label{fig:adaptive-auditor}
\end{figure}

\begin{figure}[t!]
\begin{framed}
\underline{Auditor \(\Auditor[(\Stop,\Criterion)](\inAccSmall,\rho_{\mathrm{tv}},\rho_{\dup},\alpha,\alpha_{\mathrm{tv}},\alpha_{\dup})\) routines}

\underline{\(\CheckManifest(\manifest^\coarse,\manifest^\tab)\):}
\begin{enumerate}[noitemsep]
\item If there is a batch \(\beta\in[\batches]\) for which
    \[
      \size^\tab_\beta \not\in
      \left[
        \frac{\size^\coarse_\beta}{1+\inAcc_0},
        (1+\inAcc_0)\size^\coarse_\beta
      \right]\,,
    \]
then output \(\mathtt{Error}\); otherwise output \(\mathtt{NoError}\).
\end{enumerate}

\underline{\(\BoundSize(\manifest^\tab,k_{\mathrm{tv}},\inAccSmall)\):}
\begin{enumerate}[noitemsep]
\item For \(i=1\) to \(k_{\mathrm{tv}}\):
  \begin{enumerate}[noitemsep]
  \item Sample a batch \(j\) according to the tabulated batch law
      $\Pr[j=\beta]=\size^\tab_\beta/\size^\tab.$
  \item Invoke \textbf{Count Batch} on batch \(j\) and obtain its true size
    \(\size^\act_j\).
  \item If
    \(
      \size^\act_j \not\in
      [
        {\size^\tab_j}/{(1+\inAccSmall)},
        (1+\inAccSmall)\size^\tab_j]
    \),
    output \(\mathtt{Error}\).
  \end{enumerate}
\item Output \(\mathtt{NoError}\).
\end{enumerate}

\underline{\(\DetectDuplicates(\sample,p_{\mathrm{tv}},\inAccSmall,\inAcc,\kappa_{\dup},\alpha_{\dup},N)\):}
\begin{enumerate}[noitemsep]
\item Compute
$    \eta_{\dup}\leftarrow \eta_{\dup}(p_{\mathrm{tv}},\inAccSmall,\inAcc),
    \qquad
    k_{\dup}\leftarrow
    \Phi(\eta_{\dup},\kappa_{\dup},\alpha_{\dup},N).
$
\item Draw \(k_{\dup}\) ballots using the ballot-sampling rule \(\sample\),
  forming a collection \(\mathbf{C}\). (We treat these as drawn without
  replacement by resampling if necessary.)
\item Note the identifiers appearing on the ballots in \(\mathbf{C}\).
\item If any two ballots in \(\mathbf{C}\) have the same identifier, or if any
  of the ballots are unlabeled, return \(\mathtt{Collision}\); otherwise return
  \(\mathtt{NoCollision}\).
\end{enumerate}

\underline{\(\BasicExperiment(\sample,T)\):}
\begin{enumerate}[noitemsep]
\item Use \(\sample\) to select a batch \(\beta\), and invoke
  \textbf{Sample Ballot} on \(\beta\) to obtain a ballot
  \(\ballot\in \Ballots_\beta\) uniformly at random.
\item Let \(\iota\) be the identifier appearing on \(\ballot\).
\item If \(\iota\notin T\) or \(\iota=\bot\), set
$    \winner^\cvr_\iota := 1, \loser^\cvr_\iota := 0.
  $
  Otherwise let \(r\) be the unique row of \(T\) with identifier \(\iota\), and
  set
$
    \winner^\cvr_\iota := \winner_r,\ \loser^\cvr_\iota := \loser_r.
  $
\item Let \(\winner^\act,\loser^\act\in\{0,1\}\) denote the votes on
  \(\ballot\) for the reported winner and loser, respectively.
\item Return
    $(\winner^\cvr_\iota-\loser^\cvr_\iota)-(\winner^\act-\loser^\act).$
\end{enumerate}
\end{framed}
\caption{Auditor \(\Auditor[(\Stop,\Criterion)]\) subroutines.}
\label{fig:adaptive-auditor-subroutines}
\end{figure}

To see the necessity of the \(\mathtt{DetectDuplicates}\) step, consider an
extreme example in which each vote pattern that appears on a ballot is assigned
a separate identifier, along with a single CVR row. All ballots with the same
vote pattern are assigned the same identifier. Then, the rest of the CVR is
arbitrarily populated by the adversary with identifiers and vote patterns of
their choosing. Obviously, the CVR can be made consistent with nearly any
tabulation.

The final risk-limit guarantees arise from controlling three distinct losses:
\begin{enumerate}
  \itemsep0em
  \item \textbf{Residual manifest distortion.} After
  \(\CheckManifest\) and \(\BoundSize\), the auditor has certified that the
  uniform ballot law \(U\) and the tabulation-induced ballot distribution \(\sample\)
  are linked by a \((p_{\mathrm{tv}},\inAccSmall,\inAcc)\)-manifest certificate.
  \ifnum\conference=0
  Theorem~\ref{thm:disc-lb-mostly-good-manifest} of \appref{app:auditor-analysis} 
  \else
  \cite[Theorem 6]{fuller2026sublinear}
  \fi
  converts this
  certificate into the quantitative sampling-error bound used in the discrepancy
  analysis. 
\item \textbf{Duplicate identifiers.}
  \(\DetectDuplicates\) is parameterized by the same certificate
  \((p_{\mathrm{tv}},\inAccSmall,\inAcc)\). Except with probability at most
  \(\alpha_{\dup}\), the loss caused by repeated or missing identifiers is at
  most the budgeted quantity \(\rho_{\dup}\mu\). (Unlabeled sampled ballots are treated conservatively by immediate failure of \(\DetectDuplicates\); this can only reduce the chance of acceptance and therefore does not contribute to risk.)
\item \textbf{Comparison risk.} The final adaptive audit test is then run at the
  reduced margin \(\mu_{\sample}=(1-\rho_{\mathrm{tv}}-\rho_{\dup})\mu\) with
  the remaining risk budget \(\alpha_{\sample}\).
\end{enumerate}

\subsection{Overview of Auditor Analysis}

We now analyze a single execution of \(\BasicExperiment(\sample,T)\) from
Figure~\ref{fig:adaptive-auditor-subroutines}. Fix an election \(E=(\Ballots,T)\), where
\(T\) is the tabulation provided to the auditor, and let
  $\mu := (\winner^\tab-\loser^\tab)/\size^\tab$
denote its tabulated diluted margin.

\smallskip
\noindent
Recall from Definition~\ref{def:discrepancy} that for a ballot \(\ballot\in\Ballots\), its \emph{signed ballot discrepancy}
with respect to \(T\) is defined by
\[
  \disc_T(\ballot)
  :=
  \begin{cases}
    \begin{aligned}
      &(\winner_r-\loser_r)-(\winner_{\ballot}-\loser_{\ballot}),\\[-0.25ex]
      &\quad \text{if \(\id_{\ballot}\) appears in \(T\) at \(r\)};
    \end{aligned}\\[1ex]
    \begin{aligned}
      &1-(\winner_{\ballot}-\loser_{\ballot}),\\[-0.25ex]
      &\quad \text{otherwise}.
    \end{aligned}
  \end{cases}
\]
\(\BasicExperiment(\sample,T)\) is determined by drawing a ballot \(\ballot \sim \sample\) according to the sampling rule produced by
\(\BoundSize\); the value returned by \(\BasicExperiment\) is \(\disc_T(\ballot)\), which is in \(\Sigma=\{-2,-1,0,1,2\}\).

To articulate the guarantees of the auditor, we set down some notation. Let
\[
  \kappa := \frac{(|\Ballots|-|\Labels_{\Ballots}|)}{|\Ballots|} = \frac{1}{|\Ballots|} \sum_{\iota \in \Labels_{\Ballots}} \left(|\{ \ballot \in \Ballots \mid \id_\ballot = \iota\}|-1\right)
\]
denote the \emph{excess identifier multiplicity rate}: after retaining at most one ballot for each identifier value, a $\kappa$-fraction of ballots remain unpaired.
Let
\[
  \Bad
  :=
  \left\{
    i\in[\batches]:
    \size^\tab_i \notin
    \left[
      {\size^\act_i}/{(1+\inAccSmall)},
      (1+\inAccSmall)\size^\act_i
    \right]
  \right\}
\]
denote the set of \(\inAccSmall\)-inaccurate batches, and let
$
  p:=\sum_{i\in\Bad}{\size^\tab_i}/{\size^\tab}
$
denote their total \emph{tabulated} ballot mass.
Finally, define
\begin{align}
  \varepsilon(p) &:=(1-p)\inAccSmall+p\inAcc, \label{eq:eps}\\
  \tau(p) &:=
  \left(
    \frac{1-p}{1+\inAccSmall}
    +
    \frac{p}{1+\inAcc}
  \right)^{-1}
  -1,\label{eq:tau} \\
\intertext{and}
  \Gamma_{\mathrm{tv}}(p)
  &:=
  \frac{\inAccSmall}{2+\inAccSmall}
    +
    \frac{(1+\inAccSmall)p\inAcc}{1+\inAccSmall p}\,.\label{eq:Gamma}
\end{align}

The informal theorem below summarizes the manifest-distortion bound proved in \appref{app:auditor-analysis}.
\begin{theorem}[Informal Main Theorem]
Suppose that $\size_i^\tab/\size_i^\act\in[1/(1+\inAcc),1+\inAcc]$ for every batch $i$, and let $p$ be the total tabulated mass of the $\inAccSmall$-inaccurate batches, as defined above. Then, for any invalid election,
\begingroup
\ifnum\conference=0
\small
\fi
\begin{align}
  \Exp_{\sample}[\disc_T(\ballot)]
  &\ge
  \min\!\left\{
    \frac{\mu}{1+\varepsilon(p)},
    \,
    \mu(1+\tau(p))-\tau(p)
  \right\}\nonumber
  \\&-2\kappa
  -4\Gamma_{\mathrm{tv}}(p).\label{eq:disc-lower-main-text}
\end{align}
\endgroup
\end{theorem}

The two preliminary routines in the auditor are calibrated precisely so that the
right-hand side of Equation~\eqref{eq:disc-lower-main-text} is at least the effective
margin used by the final adaptive audit test.

\begin{enumerate}
\itemsep0em
\item The duplicate-detection routine is parameterized by
\(
  \kappa_{\dup}:= {\rho_{\dup}\mu}/{2}
\).
Except with probability at most \(\alpha_{\dup}\), either
the test \(\DetectDuplicates\) returns \(\mathtt{Collision}\) or else
\(
  \kappa < \kappa_{\dup}
\).
Thus, on the event that duplicate detection is sound and the audit continues,
\(2\kappa \le 2\kappa_{\dup} = \rho_{\dup}\mu\). 

\item The manifest-certification routine is parameterized by $\mu$, $\inAcc$, $\inAccSmall$, $\rho_{\mathrm{tv}}$ and $\alpha_{\mathrm{tv}}$, and ensures that 
except with probability at most
\(\alpha_{\mathrm{tv}}\), either \(\BoundSize\) returns \(\mathtt{Error}\) or else
the accepted manifest has parameters satisfying
\begin{align}
\label{eq:manifest-budget-main-text}
\min\!\left\{
  \frac{\mu}{1+\varepsilon(p_{\mathrm{tv}})},\,
  \mu(1+\tau(p_{\mathrm{tv}}))-\tau(p_{\mathrm{tv}})
\right\}
-4\Gamma_{\mathrm{tv}}(p_{\mathrm{tv}})
\ifnum\conference=0
\nonumber
\\
\fi
\ge
(1-\rho_{\mathrm{tv}})\mu.
\end{align}
\end{enumerate}

Combining Equation~\eqref{eq:disc-lower-main-text} with
Equation~\eqref{eq:manifest-budget-main-text}, we conclude that whenever both
preliminary checks are sound and the audit reaches the comparison stage,
\[
  \Exp_{\sample}[\disc_T(\ballot)]
  \ge
  (1-\rho_{\mathrm{tv}})\mu - \rho_{\dup}\mu
  =
  (1-\rho_{\mathrm{tv}}-\rho_{\dup})\mu.
\]
Recalling the notation from Figure~\ref{fig:adaptive-auditor},
\(
  \mu_{\sample}:=(1-\rho_{\mathrm{tv}}-\rho_{\dup})\mu
\),
so the observations returned by \(\BasicExperiment\) are
\(\mu_{\sample}\)-dominating.
We summarize this as follows.

\begin{theorem}
\label{thm:one-step-domination-direct}
Fix an invalid election \(E=(\Ballots,T)\). Suppose that the manifest-certification
step and duplicate-detection step are both sound and that the auditor proceeds to the
comparison stage. Then each execution of \(\BasicExperiment(\sample,T)\) returns a
random variable \(\obsdisc\in\Sigma\) satisfying
\[
  \Exp[\obsdisc \mid \text{previous executions}] \ge \mu_{\sample},
\]
where
\(
  \mu_{\sample}=(1-\rho_{\mathrm{tv}}-\rho_{\dup})\mu
\).
Consequently, the resulting sequence of observations is
\(\mu_{\sample}\)-dominating.
\end{theorem}

\begin{corollary}
\label{cor:direct-auditor-sound}
Let
\(
  \alpha_{\sample}:=\alpha-\alpha_{\dup}-\alpha_{\mathrm{tv}}.
\)
If \((\Stop,\Criterion)\) is an \(\alpha_{\sample}\)-risk-limiting adaptive audit test,
then the auditor
\[
  \Auditor[(\Stop,\Criterion)]
  (\inAccSmall,\rho_{\mathrm{tv}},\rho_{\dup},\alpha,\alpha_{\mathrm{tv}},\alpha_{\dup})
\]
has risk at most \(\alpha\).
\end{corollary}

\begin{proof}
By Theorem~\ref{thm:one-step-domination-direct}, conditioned on the event that the two
preliminary checks are sound and the auditor reaches the comparison stage, the sequence
of outputs from \(\BasicExperiment\) is \(\mu_{\sample}\)-dominating. Therefore, the
final adaptive audit test returns \(\consistent\) with probability at most
\(\alpha_{\sample}\).

The manifest-certification step can fail to detect a bad manifest with probability at
most \(\alpha_{\mathrm{tv}}\), and the duplicate-detection step can fail to detect an
excess of duplicate identifiers with probability at most \(\alpha_{\dup}\). A union
bound therefore gives total risk at most
$\alpha_{\sample}+\alpha_{\dup}+\alpha_{\mathrm{tv}}=\alpha. \qedhere$
\end{proof}

\subsection{Duplicate Detection -- \(\Phi\)}
\label{sec:narrative-dup}

The duplicate-detection step is parameterized by the risk allocation
\(\alpha_{\dup}\) and by a duplicate budget
$  \kappa_{\dup} := \rho_{\dup}\mu/2,$
where \(\mu\) is the diluted margin of the tabulation. The role of \(\kappa_{\dup}\) is
discussed above 
\ifnum\conference=0(and treated formally in the proof of Theorem~\ref{thm:disc-lb-mostly-good-manifest}): 
\else
(and treated formally in the proof of \cite[Theorem 6]{fuller2026sublinear}): 
\fi if the excess identifier multiplicity rate
is at most \(\kappa_{\dup}\), then the loss in the discrepancy
lower bound arising from duplicate identifiers is at most
\(2\kappa_{\dup}=\rho_{\dup}\mu\).

The manifest-certification step supplies a certificate
\(
  (p_{\mathrm{tv}},\inAccSmall,\inAcc)
\)
with the property that the uniform ballot law \(U\) is a
\((p_{\mathrm{tv}},\inAccSmall,\inAcc)\)-manifest of the
tabulation-induced sampling law \(\sample\). For simplicity, we present a duplicate detector that uses this certificate through a conservative reduction that just depends on closeness in the sense of Def.~\ref{def:mult-close}.

Let
\(
  \varepsilon(p):=(1-p)\inAccSmall+p\inAcc
\) (as in Equation~\eqref{eq:eps}).
\ifnum\conference=0
By Theorem~\ref{thm:batchwise-smooth} from \appref{app:auditor-analysis},
\else
By \cite[Theorem 5]{fuller2026sublinear}
\fi if \(U\) is a
\((p,\inAccSmall,\inAcc)\)-manifest of \(\sample\), then
\(\sample\) is \(\eta_{\dup}\)-close to uniform, where
\begin{equation}
\label{eq:eta-dup-main}
\eta_{\dup} = \eta_{\dup}(p,\inAccSmall,\inAcc)
  :=
  (1+\inAcc)(1+\varepsilon(p))-1.
\end{equation}
Applying this with \(p=p_{\mathrm{tv}}\), we obtain that \(\sample\) is
\(\eta_{\dup}(p_{\mathrm{tv}}, \inAccSmall, \inAcc)\)-close to uniform.
This calibration is fully derived in
\appref{app:asymmetric-manifests}.
Accordingly, the procedure
\[  \DetectDuplicates(\sample,p_{\mathrm{tv}},\inAccSmall,\inAcc,\kappa_{\dup},\alpha_{\dup},N)
\]
takes as input the ballot-sampling rule \(\sample\), the manifest certificate
\((p_{\mathrm{tv}},\inAccSmall,\inAcc)\), the duplicate budget
\(\kappa_{\dup}\), a risk allocation \(\alpha_{\dup}\), and an upper bound
\(N\ge |\Ballots|\) on the total number of ballots (which will arise from the statistical manifest; see~\eqref{eq:size-bound} below). It first computes
\(\eta_{\dup}\) from Equation~\eqref{eq:eta-dup-main}, and then draws
\[
  k_{\dup}:=\Phi(\eta_{\dup},\kappa_{\dup},\alpha_{\dup},N)
\]
ballots without replacement according to \(\sample\), records their identifiers,
and returns \(\mathtt{Collision}\) if either
\begin{enumerate}
\itemsep0em
\item some sampled ballot is unlabeled, i.e., has identifier \(\bot\), or
\item some pair of sampled ballots have the same identifier.
\end{enumerate}
Otherwise it returns \(\mathtt{NoCollision}\).

Thus if the excess identifier multiplicity rate satisfies $\kappa \geq \kappa_{\dup}$, the
procedure returns \(\mathtt{Collision}\) with probability at least
\(1-\alpha_{\dup}\). The additional rule that unlabeled ballots also cause
immediate failure is conservative: it can only decrease the probability that the audit returns \consistent, so it affects completeness but not soundness. (Ballots
that aren't run through a tabulator are organized in hand-counted batches with  a manual $\cvr$.)
We calibrate \(\Phi\) using the following bound.

\begin{lemma}[Duplicate detection under sequential sampling]
\label{lem:dup-advertise}
Let \(A\) be a finite set with \(|A|=n\), let \(f:A\to B\) be a function with
\(|B|=n-\ell\), and let \(\Pi\) be a probability distribution on \(A\) such that
\[
\Pi(x)\ge \frac{1}{n(1+c)}
\qquad\text{for all }x\in A,
\]
for some \(c\ge 0\).

Let \(X_1,\dots,X_k\) be sampled without replacement according to the sequential law
\(
X_1\sim \Pi,
\)
and for \(t\ge 1\),
\[
\Pr(X_{t+1}=x\mid X_1,\dots,X_t)
=
\frac{\Pi(x)}
     {1-\sum_{j=1}^t \Pi(X_j)}
\]
for every $x\notin\{X_1,\dots,X_t\}$.

Define the no-collision event
\[
E_k
:=
\Bigl\{
\forall i\neq j\le k,\ f(X_i)\neq f(X_j)
\Bigr\}.
\]
Thus \(E_k\) is the event that no duplicate value $f(X)$ is
detected among the \(k\) draws. Then, for $k \leq n$,
\begin{align}
  \label{eq:dup-advertise-strong}
  \Pr[E_k] &\le 2\exp\!\left(-\ell \left[ 1-\exp\left(\frac{-k}{n(1+c)}\right)\right]^2\right)\\
\label{eq:dup-advertise}
           &\le 2\exp\!\left(-\frac{k^2\ell}{4n^2(1+c)^2}\right).
\end{align}
\end{lemma}

The proof of Lemma~\ref{lem:dup-advertise} is deferred to
\appref{sec:dup detection}. In the setting given by the auditor, we take \(A=\Ballots\)
and define
$  f(\ballot):=\id_{\ballot}$,
restricted to the event that all sampled ballots are labeled. If some sampled
ballot is unlabeled, then \(\DetectDuplicates\) immediately returns
\(\mathtt{Collision}\), so no further analysis is needed. Otherwise,
\(
  \ell = |\Ballots|-|\Labels_{\Ballots}|
\),
and the event
\(
  \forall i\neq j,\ f(X_i)\neq f(X_j)
\)
is exactly the event that no duplicate identifier is observed in the sample.
Accordingly,
\[
  \Phi(\eta_{\dup},\kappa_{\dup},\alpha_{\dup},N)
\]
may be defined as the smallest integer value \(k\) for which the right-hand side of Equation~\eqref{eq:dup-advertise} is at most \(\alpha_{\dup}\) for
\(
  n = N\) and \(\ell = \kappa_{\dup} N\). Invoking the weaker bound of~\eqref{eq:dup-advertise}, $k$ can be taken to be
\[
  2 N (1+\eta_{\dup})\sqrt{\frac{\ln(2/\alpha_{\dup})}{\ell}}=
  2 (1+\eta_{\dup})\sqrt{\frac{N \ln(2/\alpha_{\dup})}{\kappa_{\dup}}}
\]
(capped at $|\Ballots|$, since examining all ballots detects any existing collision with certainty). When generating the concrete bounds, we use the slightly stronger bound given by~\eqref{eq:dup-advertise-strong}.

\paragraph{Refinements.} We emphasize that duplicate detection can, in principle, exploit the full certificate \((p_{\mathrm{tv}},\inAccSmall,\inAcc)\) which mandates that most of the distribution has tighter ($\inAccSmall$) multiplicative guarantees. For simplicity, the development here depends only on \(\eta_{\dup}\).

\subsection{Manifest Accuracy -- \(\Psi\)}

The purpose of \(\BoundSize\) is to certify that the ballot-sampling rule induced by the tabulated batch sizes is sufficiently accurate for the final ballot-discrepancy test. (As noted above, the accuracy also affects the number of samples drawn for duplicate detection.) The certificate we want is stated in terms of the two batch probability laws that matter downstream:
\[
  r^\act_\beta := \frac{\size^\act_\beta}{\size^\act},
  \qquad
  r^\tab_\beta := \frac{\size^\tab_\beta}{\size^\tab}\,.
\]
Here \(r^\act\) is the true size-proportional batch probability law and \(r^\tab\) is the
tabulation-induced batch probability law. (The coarse manifest is used only in
\(\CheckManifest\) to certify the global \(\inAcc\)-accuracy condition.)

Define the set of \emph{bad batches}
\[
  \Bad
  :=
  \left\{
    \beta\in[\batches]:
    \size^\tab_\beta \notin
    \left[
      \frac{\size^\act_\beta}{1+\inAccSmall},
      (1+\inAccSmall)\size^\act_\beta
    \right]
  \right\},
\]
and define their total \emph{tabulated ballot mass} by
\[
  p:=\sum_{\beta\in\Bad} r^\tab_\beta
  =
  \sum_{\beta\in\Bad}\frac{\size^\tab_\beta}{\size^\tab}.
\]
Because \(\BoundSize\) samples batches directly from the tabulated batch probability law
\(r^\tab\), a sampled batch lands in the bad set \(\Bad\) with probability
exactly $p$.
Therefore, for a particular target threshold $p^\star$, the probability that
\(\BoundSize\) accepts after \(k\) independent samples when $p \geq p^\star$ is at most
\[
  (1-p^\star)^{k}
  \le
  \exp(-k p^\star).
\]
When $p \leq p^\star$, it is immediate that the true batch probability law \(r^\act\) is a \((p^\star,\inAccSmall,\inAcc)\)-manifest of
the tabulated batch probability law \(r^\tab\). The remaining question is what $p^\star$ should be selected in order to establish the distributional guarantees necessary for the auditor's later invocations of $\sample$.

It is shown in 
\ifnum\conference=0
Theorem~\ref{thm:batchwise-smooth} of \appref{app:auditor-analysis}
\else
\cite[Theorem 5]{fuller2026sublinear}
\fi
 that if the bad set has tabulated mass at most \(p\) then
\begin{equation}
  \label{eq:size-bound}
\frac{1}{1+\tau(p)}
\le
\frac{\size^{\act}}{\size^{\tab}}
\le
1+\varepsilon(p)
\end{equation}
and
\[
d_{\mathrm{TV}}(r^{\act},r^{\tab}) \le \Gamma_{\mathrm{tv}}(p)\,,
\]
where $d_{\mathrm{TV}}(\cdot,\cdot)$ denotes the distance in 
\ifnum\conference=0
total variation (see Def.~\ref{def:tv} of \appref{app:reweighting}).
\else
total variation.
\fi
Consequently, the retained margin in the final
ballot-discrepancy lower bound is controlled by the functions
\(\varepsilon(p)\), \(\tau(p)\), and \(\Gamma_{\mathrm{tv}}(p)\) defined above in equations~\eqref{eq:eps}, \eqref{eq:tau}, and~\eqref{eq:Gamma}.

For a fixed tabulated diluted margin
\(\mu\)
define the \emph{retained margin under manifest error \(p\)} by
\[
  L_\mu(p)
:=
\min\!\left\{
  \frac{\mu}{1+\varepsilon(p)},
  \,
  \mu(1+\tau(p))-\tau(p)
\right\}
-
4\Gamma_{\mathrm{tv}}(p).
\]
As each of \(\varepsilon(p)\), \(\tau(p)\), and \(\Gamma_{\mathrm{tv}}(p)\) is
nondecreasing in \(p\), the function \(L_\mu(p)\) is nonincreasing.

Thus, the manifest-certification routine should ensure that 
the total tabulated mass of the bad set does not exceed the largest
amount consistent with the retained margin remaining above the budget 
allocated for the manifest step. Define
\[
  p^\star_{\mathrm{tv}}
  :=
  \sup\left\{
    p\in[0,1]:
    L_\mu(p)\ge (1-\rho_{\mathrm{tv}})\mu
  \right\}.
\]
If this set is empty, define \(p^\star_{\mathrm{tv}}=0\).
If \(p^\star_{\mathrm{tv}}=0\), then the manifest step cannot certify enough
accuracy at the chosen budget split and the auditor should return
\(\inconclusive\).

It follows that if
$k_{\mathrm{tv}}
  :=
  \left\lceil
\ln(1/\alpha_{\mathrm{tv}})
         /p^\star_{\mathrm{tv}}
  \right\rceil,$
then, except with probability at most \(\alpha_{\mathrm{tv}}\), acceptance of
\(\BoundSize\) implies $p\le p^\star_{\mathrm{tv}}$. Accordingly, we define
\[  \Psi(\mu,\rho_{\mathrm{tv}},\alpha_{\mathrm{tv}},\inAcc,\inAccSmall)
  :=
  \left(
    p^\star_{\mathrm{tv}},
    \left\lceil
      \frac{\ln(1/\alpha_{\mathrm{tv}})}
           {p^\star_{\mathrm{tv}}}
    \right\rceil
  \right).
\]
When \(p^\star_{\mathrm{tv}}=0\), we define \(\Psi=(0,0)\).

In summary, except with probability at most
\(\alpha_{\mathrm{tv}}\), acceptance of \(\BoundSize\) implies that the bad set
has tabulated mass at most
$p_{\mathrm{tv}}=p^\star_{\mathrm{tv}}$ and that the true batch law \(r^\act\) is a
\((p^\star_{\mathrm{tv}},\inAccSmall,\inAcc)\)-manifest of the
tabulated batch law \(r^\tab\), as desired.

\myparab{Completeness}
The above test returns $\inconclusive$ if it encounters a single batch with size distortion beyond $\inAccSmall$. We focused on the strict test as the number of ballots examined for the manifest is $k_{\mathrm{tv}}$ multiplied by the average size of the sampled batch.  If a small batch of size $100$ is sampled, our setting of $\inAccSmall$ demands that the actual size is the same as the tabulated size. Other statistical methods (e.g., Wald's SPRT) can soften the strict demand that the tabulated manifest is $\inAccSmall$-accurate on each batch. 

\Appref{app:asymmetric-manifests} discusses how to analyze polling and comparison audits using the concept of a statistical manifest.
\section{Efficiency Improvements Over Existing Methods}
\label{sec:efficiency}
Our efficiency numbers focus on evaluating two primary questions: \begin{enumerate} 
\item Does direct ballot selection avoid the high cost of ballot polling
(Minerva/Providence) for small margins?
\item Does using a statistically accurate manifest meaningfully reduce the time required for direct ballot selection, ballot comparison, and ballot polling audits?\footnote{We do not take into account batch comparison audits in our analysis, because they require too many ballots. For example, for the 2024 General Election, Georgia conducted a (MACRO) batch comparison audit of the presidential race,
which had a 2.2\% margin, and audited a total of 753,651 ballots~\cite{ga2024audit}. 
At our presumed 25 seconds per ballot, this would take a total of 5,234 hours before we factor in 
manifest creation. We regard batch comparison as most effective for small ballot populations where a full hand recount is likely, because the audit has already entailed most of the work required for a recount.}\end{enumerate} 

We build on the election simulation of Fuller, Harrison, and Russell~\cite{harrison2022adaptive}, which supports ballot polling and comparison. The tool takes the specified margins ($\mu$), risk limits ($\alpha$), discrepancy rates ($o1, u1, o2, u2$), and ballot populations ($\size$) as inputs. For all methods, we assume the goal is a $90\%$ probability of completion in a single round. Across 1,000 simulations, we record the $90$th-percentile sample size for the Kaplan--Markov super-simple ballot comparison audit. To calculate Minerva sample sizes, we use the normal approximation implemented in Arlo~\cite{arlo}:
\begin{align*}
\left\lceil \left( 
\frac{
z_a \cdot \sqrt{p(1 - p)} - \tfrac{1}{2}\,z_b
}{
p - \tfrac{1}{2}
}
\right)^2 \right\rceil,
\quad 
\end{align*}
where $p=(1+\mu)/2$, $q$ is the desired probability of stopping in the next round, and $\alpha$ is the risk limit. Here $z_a=\Phi_{\mathrm N}^{-1}(1-q)$ and $z_b=\Phi_{\mathrm N}^{-1}(1-\alpha q)$, where $\Phi_{\mathrm N}$ is the standard normal CDF. Equivalently, Arlo computes these as the standard normal inverse survival function at $q$ and $\alpha q$, respectively. In our experiments, $q=0.90$ and $\alpha=0.05$.

The simulated one-vote understatement and overstatement rates are $0.1\%$, and the corresponding two-vote rates are $0.01\%$; these values are standard in the literature.
Our ballot populations are $354$K, $1$M, and the number of ballots cast in the 2024 General Election for Connecticut (1.7M), Georgia (5.0M), Florida (11.0M), and California (16.1M). 

Generating results for direct ballot selection requires splitting $\mu$ and $\alpha$ between the duplicate detection and Kaplan-Markov test. Additionally, when using a statistically accurate manifest, we split $\mu$ and $\alpha$ among three components: manifest-size certification ($\rho_{\mathrm{tv}}, \alpha_{\mathrm{tv}}$), duplicate detection ($\rho_{\dup}, \alpha_{\dup}$), and the final adaptive audit test ($\mu_{\sample}, \alpha_{\sample}$).
We compute Kaplan--Markov sample sizes only for a discrete set of $(\alpha_\sample, \mu_\sample)$ pairs, so in some cases we overestimate the number of ballots required.

For each $\size$, $\mu$, and $\alpha$ combination, we consider three possible optimizations when dividing risk and margin:  
\begin{enumerate}
\itemsep0em
\item The maximum of $\{k_\dup, k_\sample\}$,
\item The time to conduct a one-race audit of a plurality race, and
\item The time to conduct a ten-race audit of plurality races.
\end{enumerate}
We report on the best solution according to the above metrics with different allocations of $\mu$ and $\alpha$. We assign different constant fractions to $\rho_{\mathrm{tv}},\alpha_\mathrm{tv}/\alpha, \rho_\dup, \alpha_\dup/\alpha$ and keep the best solution in terms of the above objective. 
Let $k_\size:= k_\mathrm{tv} B_\sample$. 
We set limits of $k_\mathrm{\size} \le 3\size/4$ and $k_\dup \le \min\{10^6,$ $\size\}$. 
Times adopted are as follows: 
\begin{enumerate}
\itemsep0em
\item $35$ seconds to pull a ballot (median time using the scale method~\cite{RI-Pilot}), 
\item $10$ seconds to read and record the identifier (for duplicate checking), and
\item $25$ seconds to interpret one race.
\end{enumerate}
We make a few assumptions when it comes to the practical time to create an accurate ballot manifest.
We imagine a setting where auditors need to set up stacks of ballots for efficient manifest creation. Ballots are often transported in ballot bags; while many of these ballots remain stacked
neatly and can be pulled out in one large stack, some get scattered throughout the bag and will need to be picked up individually.
We assume 90\% of the ballots remain stacked and 10\% of the ballots become scattered within the bag. We assume an auditor can remove a stack of 50 ballots within 60 seconds, and it takes 6 seconds to retrieve a scattered ballot and 
add it to the stack. Saputra et al.~\cite{sapturahandcount} report 1.43 seconds to manually count a sheet of A4 paper. Combining these assumptions, setup and
counting take $0.9(60/50) + 0.1(6) + 1.43 = 3.11$ seconds per ballot, or 1,158 ballots/hr for setup and accurate manifest creation time.
This is within the range of Connecticut's pilot time of 895-1,469 ballots/hr~\cite{ctworkinggroup} and significantly slower than Rhode Island's pilot
time of 4,800 ballots/hour~\cite{RI-Pilot}, which was not software independent. We assume that the average size $B_\sample$ of a selected batch is $900$, the average precinct size in the 2020 presidential election.\footnote{There were 161M cast ballots and 176K precincts according to the \href{https://www.eac.gov/sites/default/files/document_library/files/2020_EAVS_Report_Final_508c.pdf}{EAC 2020 report}.}

\subsection{Results}
\myparab{Direct Ballot Selection and Polling}
We present three items that compare direct ballot selection to ballot polling with an accurate manifest. 
Figure~\ref{fig:main results} shows the number of sampled ballots for Minerva and direct ballot selection for different population sizes.
\ifnum\conference=1
The full version~\cite{fuller2026sublinear} reports the ballots sampled per batch across different methods and the total sample sizes.
\else
In the Appendix, Table~\ref{tab:main results} shows the ballots sampled per batch across different methods and
Table~\ref{tab:main results raw ballots} contains the total sample sizes.
\fi
All three results optimize the tradeoff between $k_\dup$ and $k_\sample$ by minimizing the total number of sampled ballots.

\begin{figure}[t]
\centering
\ifnum\conference=1
\includegraphics[width=.95\columnwidth]{figures/direct-sample-sizes}
\else
\begin{tikzpicture}
\begin{axis}[
    width=.95\columnwidth,
    height=.70\columnwidth,
    xlabel={$\mu$ (\%)},
    ylabel={Sample Size $k$},
    legend pos=north east,
    grid=both,
    mark size=2.5pt, 
    ymax=200,
    ymin=0 
]
\addplot[color=blue, mark=*] coordinates {
(0.50,180) (0.60,159) (0.70,147) (0.75,140) (0.80,134) (0.90,126) (1.00,120) (1.10,114) (1.20,107) (1.30,102) (1.40,98) (1.50,95) (1.60,92) (1.70,89) (1.80,86) (1.90,84) (2.00,82)
};
\addlegendentry{CA}

\addplot[color=red, mark=square*] coordinates {
(0.50,149) (0.60,132) (0.70,122) (0.75,117) (0.80,112) (0.90,104) (1.00,99) (1.10,94) (1.20,88) (1.30,84) (1.40,81) (1.50,78) (1.60,76) (1.70,74) (1.80,72) (1.90,70) (2.00,68)
};
\addlegendentry{FL}

\addplot[color=green!60!black, mark=triangle*] coordinates {
(0.50,101) (0.60,90) (0.70,82) (0.75,80) (0.80,77) (0.90,71) (1.00,68) (1.10,64) (1.20,62) (1.30,58) (1.40,56) (1.50,54) (1.60,52) (1.70,50) (1.80,48) (1.90,48) (2.00,46)
};
\addlegendentry{GA}

\addplot[color=black, mark=diamond*] coordinates {
(0.50,61) (0.60,54) (0.70,49) (0.75,47) (0.80,46) (0.90,42) (1.00,40) (1.10,38) (1.20,36) (1.30,34) (1.40,33) (1.50,32) (1.60,30) (1.70,30) (1.80,28) (1.90,28) (2.00,27)
};
\addlegendentry{CT}

\addplot[color=orange, mark=square*] coordinates {
(0.50,48) (0.60,42) (0.70,39) (0.75,37) (0.80,36) (0.90,33) (1.00,31) (1.10,30) (1.20,28) (1.30,26) (1.40,26) (1.50,24) (1.60,24) (1.70,23) (1.80,22) (1.90,22) (2.00,21)
};
\addlegendentry{1M}

\addplot[color=purple, mark=*] coordinates {
(0.50,30) (0.60,26) (0.70,24) (0.75,24) (0.80,22) (0.90,20) (1.00,20) (1.10,18) (1.20,17) (1.30,16) (1.40,16) (1.50,15) (1.60,14) (1.70,14) (1.80,14) (1.90,14) (2.00,13)
};
\addlegendentry{Cong.}

\addplot[dashed, ultra thick, color=gray!70!black, mark=none, smooth] coordinates {
(0.50,354) (0.60,246) (0.70,181) (0.75,158) (0.80,138) (0.90,109) (1.00,89) (1.10,73) (1.20,62) (1.30,52) (1.40,45) (1.50,39) (1.60,35) (1.70,31) (1.80,27) (1.90,25) (2.00,22)
};
\addlegendentry{Minerva}

\end{axis}
\end{tikzpicture}
\fi
\caption{Direct ballot selection sample sizes, $\max\{k_{\dup},k_{\sample}\}$, compared with Minerva sample sizes.}
\label{fig:main results}
\end{figure}

We stress that polling is particularly problematic for small elections and small margins.
According to the results in the
\ifnum\conference=1
full version~\cite{fuller2026sublinear},
\else
Appendix, Table~\ref{tab:main results},
\fi
at a 0.5\% margin, the sample size of a Minerva polling audit is the same as the size of a Congressional election. Direct ballot 
selection reduces the number of ballots per batch to 76; roughly $30$K or 8\% of the ballots. At larger sizes, the 
improvements are not as drastic but still significant. At a 0.5\% margin, direct ballot selection 
reduces the ballots sampled by a factor of $2.4$ in Florida and a factor of 2 in California. 
At a 1\% margin in Connecticut (ranked 29th by population), direct ballot selection reduces the number of ballots pulled by $55\%$. For ballot populations of practical interest, direct ballot selection substantially improves upon polling's performance at small margins. Of course, audit software can still choose to revert to polling audits for large margins 
without the need for additional setup or equipment.

\myparab{Making manifests accurate without a full hand count.}
We now discuss how one can build an accurate enough manifest for direct ballot selection and ballot comparison without a full hand count of the ballots.
Creating an accurate manifest dominates the time of every audit method, especially for larger states.  
For California, the manifest requires nearly 14,000 hours to prepare. Even for smaller states such as Connecticut, preparing the manifest requires roughly 1,450 hours of work. \textbf{Meanwhile, the time difference for ballot sampling and analysis between a $3\%$ and $0.5\%$ margin is only $38$ hours.}

For ballot comparison, a sufficiently accurate statistical manifest can be certified by examining relatively few ballots. Figure~\ref{fig:comp inacc savings} shows the time savings from using a statistical manifest rather than a fully accurate manifest.
\ifnum\conference=1
The full version~\cite{fuller2026sublinear} reports the corresponding numbers of ballots required to construct the manifests.
\else
Table~\ref{tab:comp inacc savings} in the Appendix reports the corresponding numbers of ballots required to construct the manifests.
\fi
At a 3\% margin, California can certify a statistical manifest by counting $k_\size \approx 21$K ballots rather than counting all $16$M ballots.
As a result, the number of ballots to audit increases from $263$ to $836$, making the overall audit take  $32$ hours. For Connecticut at a $1\%$ margin, $k_\size \approx 73$K ballots (instead of $1.7$M), which results in a sample that increases from $1020$ to $3496$ and an audit that takes approximately $121$ hours.
\ifnum\conference=1
Polling, also discussed in the full version~\cite{fuller2026sublinear}, behaves similarly.
\else
Polling behaves similarly but with less graceful degradation with margin. 
\fi

\begin{figure}[th!]
\centering
\begin{subfigure}{0.9\columnwidth}
\centering
\ifnum\conference=1
\includegraphics[width=\columnwidth]{figures/manifest-savings-comparison}
\else
\begin{tikzpicture}
\begin{axis}[
    width=\columnwidth,
    height=0.75\columnwidth,
    xlabel={$\mu$ (\%)},
    ylabel={Ratio},
    ymode=log,
    xtick={2, 4, 6, 8, 10},
    legend style={at={(0.98,0.02)}, anchor=south east},
    grid=both,
]

\pgfplotstableread[col sep=comma]{time_results_by_sizecomparison.csv}\datatable

\addplot+[mark=*, thick] table[x=Margin, y=CA] {\datatable};
\addlegendentry{CA}

\addplot+[mark=square*, thick] table[x=Margin, y=FL] {\datatable};
\addlegendentry{FL}

\addplot+[mark=triangle*, thick] table[x=Margin, y=GA] {\datatable};
\addlegendentry{GA}

\addplot+[mark=diamond*, thick] table[x=Margin, y=CT] {\datatable};
\addlegendentry{CT}
\end{axis}
\end{tikzpicture}
\fi
\caption{Ballot comparison (log scale). 
\ifnum\conference=0 Details in Table~\ref{tab:comp inacc savings}.\fi}
\label{fig:comp inacc savings}
\end{subfigure}

\vspace{1em}

\begin{subfigure}{0.9\columnwidth}
\centering
\ifnum\conference=1
\includegraphics[width=\linewidth]{figures/manifest-savings-direct}
\else
\begin{tikzpicture}
\begin{axis}[
    width=\linewidth,
    height=0.75\textwidth,
    xlabel={$\mu$ (\%)},
    ylabel={Ratio},
    xtick={2, 4, 6, 8, 10},
    legend style={at={(0.02,0.98)}, anchor=north west},
    grid=both,
]

\pgfplotstableread[col sep=comma]{time_results_by_sizedirect.csv}\datatable

\addplot+[mark=*, thick] table[x=Margin, y=CA] {\datatable};
\addlegendentry{CA}

\addplot+[mark=square*, thick] table[x=Margin, y=FL] {\datatable};
\addlegendentry{FL}

\addplot+[mark=triangle*, thick] table[x=Margin, y=GA] {\datatable};
\addlegendentry{GA}

\addplot+[mark=diamond*, thick] table[x=Margin, y=CT] {\datatable};
\addlegendentry{CT}

\end{axis}
\end{tikzpicture}
\fi
\caption{Direct ballot selection. \ifnum\conference=0 Details in Table~\ref{tab:inacc direct savings}.\fi}
\label{fig:inacc direct savings}
\end{subfigure}

\vspace{1em}
\begin{subfigure}{0.9\columnwidth}
\centering
\ifnum\conference=1
\includegraphics[width=\linewidth]{figures/manifest-savings-polling}
\else
\begin{tikzpicture}
\begin{axis}[
    width=\linewidth,
    height=0.75\textwidth,
    xlabel={$\mu$ (\%)},
    ylabel={Ratio},
    ymode=log,
    xtick={2, 4, 6, 8, 10},
    legend style={at={(0.98,0.02)}, anchor=south east},
    grid=both,
]

\pgfplotstableread[col sep=comma]{time_results_by_sizepolling.csv}\datatable

\addplot+[mark=*, thick] table[x=Margin, y=CA] {\datatable};
\addlegendentry{CA}

\addplot+[mark=square*, thick] table[x=Margin, y=FL] {\datatable};
\addlegendentry{FL}

\addplot+[mark=triangle*, thick] table[x=Margin, y=GA] {\datatable};
\addlegendentry{GA}

\addplot+[mark=diamond*, thick] table[x=Margin, y=CT] {\datatable};
\addlegendentry{CT}

\end{axis}
\end{tikzpicture}
\fi
\caption{Minerva polling (log scale). \ifnum\conference=0 Details in Table~\ref{tab:polling inacc savings}.\fi}
\label{fig:inacc polling savings}
\end{subfigure}

\caption{Ratio of time to conduct audit with full manifest versus statistical manifest across audit methods.}
\label{fig:combined_inacc_savings}
\end{figure}

For direct ballot selection, the improvements from using a statistical manifest rather than a fully accurate manifest are more modest. 
\ifnum\conference=0
These can be seen in Figure~\ref{fig:inacc direct savings} and Table~\ref{tab:inacc direct savings}.
\else
This is shown in  Figure~\ref{fig:inacc direct savings}.
\fi
For large states, the savings from a statistical manifest grow with the margin: audit time falls by a factor of four or more once the margin reaches $3\%$, reaching a factor of $11$ in Florida and $14$ in California at a $5\%$ margin. This allows every component of the audit to scale with the margin instead of allowing manifest creation to dominate the total audit time. 
For example, at a 2\% margin, Georgia can reduce its audit time from 5,000 hours to 1,300 hours by hand-counting $k_\size = 355$K ballots to construct a statistical manifest rather than counting all $5$M ballots.
This causes $k_\dup$ to increase from roughly $50$K to $75$K and $k_\sample$ to increase from roughly $5$K to $11$K as part of the direct ballot selection audit.

For direct ballot selection, the method reverts to a full manifest for margins below $0.5\%$ at almost all population sizes. For ballot polling with populations below $1$M, a full manifest is used for margins at most $1\%$.
Ballot comparison, which is already efficient in the number of ballots examined, still achieves a 41-fold reduction in time for California at a $0.5\%$ margin. Statistically accurate manifests accomplish the goal that all aspects of the audit scale with $\mu$.

\paragraph{Comparison to Lindeman et al.~\cite{lindeman2012bravo}}
We also evaluate Lindeman et al.'s technique for supporting statistically accurate manifests. The technique treats each CVR as having $\size(1+\inAcc)$ entries and records a discrepancy of $2$ whenever one of the $\size\inAcc$-added rows is selected. Across our experiments, this technique converges when $\inAcc <\mu/5$. One can use our $\BoundSize$ technique to produce tighter estimates for the statistical size bounds. Lindeman et al.'s technique does not naturally extend to polling or direct ballot selection. For comparison audits, we obtain better results by removing the added discrepancy-$1$ rows and adjusting the margin accordingly; see \appref{app:asymmetric-manifests}.

\subsection{Manifest-Accuracy Models}
\label{subsec:manifest-comparison}
In the main text we focus on a symmetric multiplicative model where the guarantee is that $\size^\cvr$ is within a multiplicative $(1+\Delta)$ factor of $\size^\act$. We present two other models as well in
\ifnum\conference=0 Appendix~\ref{app:auditor-analysis}. Its general directional certificate treats the two size distortions separately, and Appendix~\ref{app:upper-bound-manifests} gives the upper-bound-only specialization, so the model assumes that
\else
the full version of this work~\cite{fuller2026sublinear}. The directional model generalizes the model in the main text by assuming that
\fi \[\frac{1}{1+\Delta_{\mathrm{lower}}}\le \frac{\size^\act_i}{\size^\cvr_i}\le 1+\Delta_{\mathrm{upper}}\] for each batch $i$. In addition, we also show a total variation bound when $\Delta_{\mathrm{upper}}=0$ and no bound is known for $\Delta_{\mathrm{lower}}$. Setting both parameters to $0$ yields an accurate manifest. The upper-bound-only method requires sampling the most ballots across our results for direct selection and polling. However, for comparison audits, it is enough: the comparison analysis depends only on an upper bound on the total number of ballots, so once that upper bound is established, no additional manifest-certification work is needed. In particular, comparison does not require the batchwise lower-threshold certification used by the polling and direct ballot selection analyses.

The regime requiring the fewest sampled ballots sets $\Delta_{\mathrm{upper}}=0$ and $\Delta_{\mathrm{lower}}=\Delta$:
\[
\frac{S^{\cvr}_\beta}{1+\Delta} \le S^{\act}_\beta \le S^{\cvr}_\beta.
\]
Its quantitative advantages over the symmetric multiplicative distortion model are modest when $\Delta$ is small, but for larger values the two bounds can reduce the manifest workload noticeably in the polling experiments.

\ifnum\conference=1
\begin{acks}
The work of all authors was supported through grants by the Connecticut
Secretary of State's Office and the Department of Homeland Security. In
addition, the work of B.F. was supported by NSF Grants \#2141033 and
\#2232813.
\end{acks}
\else
\section*{Acknowledgments}
The work of all authors was supported through grants by the Connecticut Secretary of State's Office and the Department of Homeland Security. In addition, the work of B.F. was supported by NSF Grants \#2141033 and \#2232813.
\fi

\ifnum\conference=1
\bibliographystyle{ACM-Reference-Format}
\else
\bibliographystyle{unsrt}
\fi
\bibliography{RLA}

\appendix
\section{Open Science}
Code to reproduce our efficiency experiments is available online at: 
\url{https://github.com/VoterCenter/direct-rla}.
This code allows for the complete reproduction of the results in this paper.
\section{Ethical Considerations}
This work proposes new risk-limiting audits for elections that are usable in various municipalities throughout the world. We use data from publicly available sources, including websites maintained by secretaries of state throughout the United States. In addition, we remark on the necessity of accurate ballot manifests in currently used risk-limiting audits. Vendors (and previous literature) already stress the need for an accurate manifest, so we do not believe this exposes a vulnerability in current systems.
\section{Counterexamples with Inaccurate Manifests}
\label{sec:inaccurate no work}
We provide simple examples showing that polling, ballot comparison, and direct ballot selection audits are not risk-limiting when each batch size is distorted by an amount $\inAcc$ comparable to the margin $\mu$.  This discussion applies to methods that use reported batch sizes to sample ballots uniformly.  This description covers most methods in each class but does not apply to taint-tracking methods~\cite{higgins2011sharper}.  Throughout the discussion below, batch sizes are manipulated by less than the $\inAcc = 0.1$ used in our experimental results.

\paragraph{Ballot polling.} Consider an election with $\winner^\act$ votes for the reported winner $\winner$ and $\loser^\act$ votes for the reported loser $\loser$.  Furthermore, the tabulated margin is $\mu^\tab$.  In simplified terms, a polling audit uses uniformly sampled ballots to estimate the signed vote margin
\[
\frac{\winner^\act - \loser^\act}{\size}
\]
and tests whether the samples support the reported winner.  Recall that since   $\winner^\act = \sum_{\mathbf{b} \in \mathbf{B}} \winner_{\mathbf{b}}$ and $\loser^\act = \sum_{\mathbf{b} \in \mathbf{B}} \loser_{\mathbf{b}}$ one can estimate these quantities by repeatedly sampling ballots $\mathbf{b} \in \mathbf{B}$ uniformly and computing $ \winner_{\mathbf{b}}-  \loser_{\mathbf{b}}$.  With an accurate manifest each $\mathbf{b}$ is selected with probability $1/\size$ as $\size = |\mathbf{B}|$.  Define $\mu^\obs_{\mathbf{b}}:=\winner_{\mathbf{b}}-\loser_{\mathbf{b}}$. The goal of the statistical test is to lower-bound the expectation $
\Exp_{\mathbf{b}\leftarrow \mathbf{B}}\mu^\obs_{\mathbf{b}}\,.$ 
This expectation is highly sensitive to the sampling of $\mathbf{b}\leftarrow \mathbf{B}$ which can be adjusted by manipulating the size of batches.  We use $U$ to denote the proper uniform sampling of ballots.  Consider the following simplified example where there are two batches of size $49$ and $51$ denoted as $\winner$ and $\loser$ respectively, where all ballots in $\winner$ have a vote for candidate $\winner$ and all ballots in $\loser$ have a vote for $\loser$. The correct $\Exp_{\mathbf{b}\leftarrow \mathbf{B}_U}\mu^\obs_{\mathbf{b}}=
-.02$  for this election.  One can adjust the diluted margin from $.02$ for the $\loser$ to $.02$ for the winner  by adjusting the sizes to $51$ and $49$ respectively. This creates a new distribution of sampled ballots $D_\Adversary$ where $\Exp_{\mathbf{b}\leftarrow \mathbf{B}_{D_\Adversary}}\mu^\obs_{\mathbf{b}}=
.02$ .  

The above example is very simplified; the quality of the manipulation is proportional to the number of batches that have strong margins for one candidate or another. However, the polarization displayed does mirror current electoral trends in the United States. Often, large cities have overwhelming Democratic preferences while rural areas have overwhelming Republican preferences.  Taking the 2024 Presidential Election in Connecticut as an example, 26\% of batches had a margin that was 30\% different than the overall statewide margin. 

\paragraph{Ballot comparison.} A traditional ballot comparison audit tests the null hypothesis that the reported outcome is incorrect, equivalently,
\[
\frac{\disc}{\size^\tab} \ge \mu^\tab.
\]
Operationally, Fuller, Harrison, and Russell~\cite{harrison2022adaptive} define sampled discrepancy with respect to CVR rows rather than physical ballots (again assuming that identifiers are unique within the CVR).
If the manifest underestimates $|\mathbf{B}|$, the CVR can omit rows corresponding to physical ballots, thereby biasing the audit sample.  We continue with the same example of two batches from the ballot polling above.  By reporting the $\loser$ batch as having $47$ ballots instead of $51$, the tabulated diluted margin becomes $2.1\%$ for $\winner$. The adversary can then construct a $\cvr$ containing accurate rows for all $49$ $\winner$ ballots but only $47$ of the $51$ $\loser$ ballots. Because the omitted ballots have no corresponding CVR rows, the audit can never observe their discrepancies. Unlike in ballot polling, an adversary cannot beneficially inflate batch sizes in a ballot comparison audit, because CVR identifiers with no corresponding physical ballots increase discrepancy~\cite{harrison2022adaptive}. This example does not depend on how ballots are distributed among batches: the adversary can omit CVR entries for ballots favoring the actual winner in any batch.

\paragraph{Direct ballot selection.} Direct ballot selection is similarly vulnerable. Instead of reducing the size of the $\cvr$, the adversary increases it. Continuing the same example, the adversary reports that the $\winner$ batch contains $53$ ballots instead of $49$ by adding four $\cvr$ rows that record votes for $\winner$ and have identifiers that do not correspond to any physical ballot. Sampling batches in proportion to their reported sizes then selects the $\winner$ batch with probability $53/104$ and the $\loser$ batch with probability $51/104$, producing an expected sampled margin of $(53-51)/104 \approx 0.0192$ for $\winner$. Because the four phantom rows can never be sampled through physical-ballot selection, the audit does not observe their discrepancies. This illustrates the dual risks: ballot-based sampling can miss phantom CVR rows, whereas CVR-row sampling can miss physical ballots omitted from the CVR.

\section{Additional Concrete Results}
\ifnum\conference=0
Table~\ref{tab:main results} and 
Table~\ref{tab:main results raw ballots} show the main results of our election simulation: the 
average number of ballots sampled per batch and the total number of ballots needed to conduct a
direct ballot selection audit, respectively. These results are also discussed in detail in Section~\ref{sec:efficiency}. 
\else
Table~\ref{tab:main results} shows the 
average number of ballots sampled per batch to conduct a
direct ballot selection audit. These results are also discussed in detail in Section~\ref{sec:efficiency}. 
\fi

\ifnum\conference=0
\begin{table}[t!]
\else
\begin{table}[th!]
\fi
  \centering
  \small
  \ifnum\conference=0
\begin{tabular}{V{3}l V{3}r|r|r|r|r|rV{3}} \hlineB{3}
\else
\begin{tabular}{|l|r|r|r|r|r|r|} \hline\fi
& & \multicolumn{4}{c|}{Direct Ballot Selection } &\\
& $\mu$ & \multicolumn{2}{c|}{$\dup$ } & \multicolumn{2}{c|}{$\sample$ } & Minerva \\
State& (\%) & $k$ & $\alpha$ & $k$ & $\alpha$ & ~\cite{zagorski2021minerva} \\
 \ifnum\conference=0
\hlineB{3}
\else
\hline
\fi
\multirow{7}{*}{\makecell[l]{Cong.}}
 & 0.50 & 76 & 90 & 71 & 10 & 900 \\ \cline{2-7}
 & 0.75 & 60 & 70 & 58 & 30 & 401 \\ \cline{2-7}
 & 1.00 & 50 & 70 & 42 & 30 & 225 \\ \cline{2-7}
 & 1.50 & 38 & 83 & 37 & 17 & 100 \\ \cline{2-7}
 & 2.00 & 33 & 80 & 23 & 20 & 56 \\ \cline{2-7}
 & 2.50 & 29 & 78 & 28 & 22 & 36 \\ \cline{2-7}
 & 3.00 & 27 & 70 & 22 & 30 & 25 \\ \cline{2-7}
 \ifnum\conference=0
\hlineB{3}
\else
\hline
\fi

\multirow{7}{*}{\makecell[l]{1M}}
 & 0.50 & 43 & 88 & 36 & 12 & 319 \\ \cline{2-7}
 & 0.75 & 33 & 78 & 33 & 22 & 142 \\ \cline{2-7}
 & 1.00 & 28 & 98 & 23 & 2 & 80 \\ \cline{2-7}
 & 1.50 & 22 & 98 & 15 & 2 & 35 \\ \cline{2-7}
 & 2.00 & 19 & 83 & 16 & 17 & 20 \\ \cline{2-7}
 & 2.50 & 17 & 93 & 14 & 7 & 13 \\ \cline{2-7}
 & 3.00 & 15 & 88 & 8 & 12 & 9 \\ \cline{2-7}
 \ifnum\conference=0
\hlineB{3}
\else
\hline
\fi

\multirow{7}{*}{\makecell[l]{CT}}
 & 0.50 & 33 & 75 & 32 & 25 & 190 \\ \cline{2-7}
 & 0.75 & 25 & 88 & 21 & 12 & 84 \\ \cline{2-7}
 & 1.00 & 21 & 85 & 21 & 15 & 47 \\ \cline{2-7}
 & 1.50 & 17 & 90 & 8 & 10 & 21 \\ \cline{2-7}
 & 2.00 & 14 & 88 & 10 & 12 & 12 \\ \cline{2-7}
 & 2.50 & 13 & 90 & 6 & 10 & 8 \\ \cline{2-7}
 & 3.00 & 12 & 88 & 5 & 12 & 5 \\ \cline{2-7}
 \ifnum\conference=0
\hlineB{3}
\else
\hline
\fi

\multirow{7}{*}{\makecell[l]{GA}}
 & 0.50 & 18 & 98 & 18 & 2 & 64 \\ \cline{2-7}
 & 0.75 & 14 & 98 & 14 & 2 & 28 \\ \cline{2-7}
 & 1.00 & 12 & 90 & 7 & 10 & 16 \\ \cline{2-7}
 & 1.50 & 10 & 90 & 7 & 10 & 7 \\ \cline{2-7}
 & 2.00 & 8 & 93 & 4 & 7 & 4 \\ \cline{2-7}
 & 2.50 & 7 & 95 & 3 & 5 & 3 \\ \cline{2-7}
 & 3.00 & 7 & 78 & 7 & 22 & 2 \\ \cline{2-7}
 \ifnum\conference=0
\hlineB{3}
\else
\hline
\fi

\multirow{7}{*}{\makecell[l]{FL}}
 & 0.50 & 12 & 98 & 8 & 2 & 29 \\ \cline{2-7}
 & 0.75 & 10 & 83 & 10 & 17 & 13 \\ \cline{2-7}
 & 1.00 & 8 & 98 & 6 & 2 & 7 \\ \cline{2-7}
 & 1.50 & 6 & 98 & 6 & 2 & 3 \\ \cline{2-7}
 & 2.00 & 6 & 98 & 2 & 2 & 2 \\ \cline{2-7}
 & 2.50 & 5 & 98 & 1 & 2 & 1 \\ \cline{2-7}
 & 3.00 & 4 & 88 & 3 & 12 & 1 \\ \cline{2-7}
 \ifnum\conference=0
\hlineB{3}
\else
\hline
\fi

\multirow{7}{*}{\makecell[l]{CA}}
 & 0.50 & 10 & 98 & 6 & 2 & 20 \\ \cline{2-7}
 & 0.75 & 8 & 90 & 7 & 10 & 9 \\ \cline{2-7}
 & 1.00 & 7 & 98 & 4 & 2 & 5 \\ \cline{2-7}
 & 1.50 & 5 & 98 & 4 & 2 & 2 \\ \cline{2-7}
 & 2.00 & 5 & 98 & 1 & 2 & 1 \\ \cline{2-7}
 & 2.50 & 4 & 83 & 4 & 17 & 1 \\ \cline{2-7}
 & 3.00 & 4 & 93 & 3 & 7 & 1 \\ \cline{2-7}
 \ifnum\conference=0
\hlineB{3}
\else
\hline
\fi
\end{tabular}
  \caption{Simulated sample sizes per batch for direct ballot selection and
  Minerva~\cite{zagorski2021minerva}. 
 $\alpha_\dup, \alpha_\sample$ are normalized percentages of overall risk $\alpha$. Direct ballot selection optimization minimizes the max of $k_\dup, k_\sample$. 
 \ifnum\conference=0 Total numbers of ballots drawn are given in Table~\ref{tab:main results raw ballots}.\fi
  }
  \label{tab:main results}
\end{table}

\ifnum\conference=0

\begin{table}[th!]
  \centering
  \small
\begin{tabular}{V{3}l V{3}r|r|r|r|r|rV{3}} \hlineB{3}
& & \multicolumn{4}{c|}{Direct Ballot Selection } & \\
& $\mu$ & \multicolumn{2}{c|}{$\dup$ } & \multicolumn{2}{c|}{$\sample$ } & Minerva \\
State& (\%) & $k$ & $\alpha$ & $k$ & $\alpha$ & ~\cite{zagorski2021minerva} \\\hlineB{3}
\multirow{7}{*}{\makecell[l]{Cong.}}
 & 0.50 & 30 & 90 & 28 & 10 & 354 \\ \cline{2-7}
 & 0.75 & 24 & 70 & 23 & 30 & 158 \\ \cline{2-7}
 & 1.00 & 20 & 70 & 17 & 30 & 89 \\ \cline{2-7}
 & 1.50 & 15 & 83 & 15 & 17 & 39 \\ \cline{2-7}
 & 2.00 & 13 & 80 & 9 & 20 & 22 \\ \cline{2-7}
 & 2.50 & 12 & 78 & 11 & 22 & 14 \\ \cline{2-7}
 & 3.00 & 10 & 70 & 9 & 30 & 10 \\ \cline{2-7}
\hlineB{3}

\multirow{7}{*}{\makecell[l]{1M}}
 & 0.50 & 48 & 88 & 40 & 12 & 354 \\ \cline{2-7}
 & 0.75 & 37 & 78 & 37 & 22 & 158 \\ \cline{2-7}
 & 1.00 & 31 & 98 & 25 & 2 & 89 \\ \cline{2-7}
 & 1.50 & 24 & 98 & 17 & 2 & 39 \\ \cline{2-7}
 & 2.00 & 21 & 83 & 18 & 17 & 22 \\ \cline{2-7}
 & 2.50 & 18 & 93 & 16 & 7 & 14 \\ \cline{2-7}
 & 3.00 & 17 & 88 & 9 & 12 & 10 \\ \cline{2-7}
\hlineB{3}

\multirow{7}{*}{\makecell[l]{CT}}
 & 0.50 & 61 & 75 & 60 & 25 & 354 \\ \cline{2-7}
 & 0.75 & 47 & 88 & 40 & 12 & 158 \\ \cline{2-7}
 & 1.00 & 40 & 85 & 38 & 15 & 89 \\ \cline{2-7}
 & 1.50 & 32 & 90 & 15 & 10 & 39 \\ \cline{2-7}
 & 2.00 & 27 & 88 & 18 & 12 & 22 \\ \cline{2-7}
 & 2.50 & 24 & 90 & 12 & 10 & 14 \\ \cline{2-7}
 & 3.00 & 22 & 88 & 9 & 12 & 10 \\ \cline{2-7}
\hlineB{3}

\multirow{7}{*}{\makecell[l]{GA}}
 & 0.50 & 101 & 98 & 100 & 2 & 354 \\ \cline{2-7}
 & 0.75 & 80 & 98 & 76 & 2 & 158 \\ \cline{2-7}
 & 1.00 & 68 & 90 & 40 & 10 & 89 \\ \cline{2-7}
 & 1.50 & 54 & 90 & 40 & 10 & 39 \\ \cline{2-7}
 & 2.00 & 46 & 93 & 20 & 7 & 22 \\ \cline{2-7}
 & 2.50 & 41 & 95 & 14 & 5 & 14 \\ \cline{2-7}
 & 3.00 & 38 & 78 & 37 & 22 & 10 \\ \cline{2-7}
\hlineB{3}

\multirow{7}{*}{\makecell[l]{FL}}
 & 0.50 & 149 & 98 & 100 & 2 & 354 \\ \cline{2-7}
 & 0.75 & 116 & 83 & 117 & 17 & 158 \\ \cline{2-7}
 & 1.00 & 99 & 98 & 76 & 2 & 89 \\ \cline{2-7}
 & 1.50 & 78 & 98 & 76 & 2 & 39 \\ \cline{2-7}
 & 2.00 & 68 & 98 & 25 & 2 & 22 \\ \cline{2-7}
 & 2.50 & 60 & 98 & 14 & 2 & 14 \\ \cline{2-7}
 & 3.00 & 54 & 88 & 40 & 12 & 10 \\ \cline{2-7}
\hlineB{3}

\multirow{7}{*}{\makecell[l]{CA}}
 & 0.50 & 180 & 98 & 100 & 2 & 354 \\ \cline{2-7}
 & 0.75 & 140 & 90 & 131 & 10 & 158 \\ \cline{2-7}
 & 1.00 & 120 & 98 & 76 & 2 & 89 \\ \cline{2-7}
 & 1.50 & 95 & 98 & 76 & 2 & 39 \\ \cline{2-7}
 & 2.00 & 82 & 98 & 25 & 2 & 22 \\ \cline{2-7}
 & 2.50 & 73 & 83 & 72 & 17 & 14 \\ \cline{2-7}
 & 3.00 & 66 & 93 & 50 & 7 & 10 \\ \cline{2-7}
\hlineB{3}
\end{tabular}
  \caption{Simulated sample sizes in thousands for direct ballot selection and
  Minerva~\cite{zagorski2021minerva}. 
 $\alpha_\dup, \alpha_\sample$ are normalized percentages of overall risk $\alpha$. Direct ballot selection optimization minimizes the max of $k_\dup, k_\sample$.
  }
  \label{tab:main results raw ballots}
  
\end{table}

\ifnum\conference=0
\begin{table*}[t!]
\centering
\resizebox{\textwidth}{!}{
\begin{tabular}{V{3}l V{3}r V{3}r|r|rV{3}r|r|rV{3}r|r|rV{3}r|r|rV{3}r|r|rV{3}r|r|rV{3}} \hlineB{3}
&  & \multicolumn{3}{cV{3}}{Accurate} & \multicolumn{3}{cV{3}}{$\Delta_{Upp}=0$} & \multicolumn{3}{cV{3}}{$\Delta=0.5$} & \multicolumn{3}{cV{3}}{$\Delta_{Low}=0.5\wedge \Delta_{Upp}=0$} & \multicolumn{3}{cV{3}}{$\Delta=0.1$} & \multicolumn{3}{cV{3}}{$\Delta_{Low}=0.1 \wedge \Delta_{Upp}=0$} \\
State & $\mu\ (\%)$ & $k_{\mathrm{dup}}$ & $k_{\mathrm{comp}}$ & $k_{\mathrm{size}}$ & $k_{\mathrm{dup}}$ & $k_{\mathrm{comp}}$ & $k_{\mathrm{size}}$ & $k_{\mathrm{dup}}$ & $k_{\mathrm{comp}}$ & $k_{\mathrm{size}}$ & $k_{\mathrm{dup}}$ & $k_{\mathrm{comp}}$ & $k_{\mathrm{size}}$ & $k_{\mathrm{dup}}$ & $k_{\mathrm{comp}}$ & $k_{\mathrm{size}}$ & $k_{\mathrm{dup}}$ & $k_{\mathrm{comp}}$ & $k_{\mathrm{size}}$ \\\hlineB{3}
\multirow{7}{*}{Cong.} & 0.5 & 35 & 11 & 354 & 48 & 43 & 354 & 62 & 63 & 354 & 48 & 48 & 354 & 72 & 65 & 354 & 48 & 46 & 354 \\\cline{2-20}
 & 1.0 & 22 & 6 & 354 & 23 & 20 & 354 & 21 & 18 & 354 & 23 & 20 & 354 & 24 & 23 & 354 & 23 & 23 & 354 \\\cline{2-20}
 & 2.0 & 15 & 3 & 354 & 14 & 10 & 354 & 12 & 10 & 354 & 14 & 10 & 354 & 14 & 10 & 354 & 24 & 7 & 219 \\\cline{2-20}
 & 3.0 & 12 & 2 & 354 & 11 & 9 & 354 & 9 & 7 & 354 & 11 & 9 & 354 & 19 & 5 & 165 & 19 & 5 & 139 \\\cline{2-20}
 & 5.0 & 9 & 2 & 354 & 8 & 6 & 354 & 7 & 4 & 354 & 17 & 4 & 216 & 14 & 4 & 98 & 13 & 3 & 87 \\\cline{2-20}
 & 8.0 & 7 & 2 & 354 & 7 & 6 & 354 & 16 & 3 & 192 & 12 & 3 & 139 & 11 & 2 & 61 & 10 & 2 & 54 \\\cline{2-20}
 & 10.0 & 6 & 1 & 354 & 13 & 3 & 224 & 13 & 3 & 165 & 10 & 2 & 121 & 9 & 2 & 50 & 9 & 2 & 47 \\\cline{2-20}
\hlineB{3}
\multirow{7}{*}{1 mil} & 0.5 & 55 & 14 & 1000 & 74 & 63 & 1000 & 94 & 91 & 1000 & 75 & 70 & 1000 & 107 & 101 & 1000 & 76 & 55 & 1000 \\\cline{2-20}
 & 1.0 & 34 & 9 & 1000 & 37 & 31 & 1000 & 33 & 27 & 1000 & 37 & 30 & 1000 & 38 & 35 & 1000 & 65 & 16 & 492 \\\cline{2-20}
 & 2.0 & 23 & 4 & 1000 & 23 & 22 & 1000 & 20 & 12 & 1000 & 46 & 12 & 597 & 39 & 10 & 281 & 37 & 8 & 232 \\\cline{2-20}
 & 3.0 & 19 & 3 & 1000 & 18 & 15 & 1000 & 45 & 9 & 572 & 33 & 8 & 415 & 29 & 6 & 189 & 26 & 6 & 176 \\\cline{2-20}
 & 5.0 & 14 & 2 & 1000 & 31 & 6 & 476 & 31 & 5 & 355 & 24 & 4 & 247 & 22 & 3 & 104 & 20 & 3 & 104 \\\cline{2-20}
 & 8.0 & 11 & 1 & 1000 & 22 & 5 & 303 & 23 & 4 & 219 & 18 & 3 & 159 & 16 & 2 & 73 & 15 & 2 & 64 \\\cline{2-20}
 & 10.0 & 10 & 1 & 1000 & 20 & 3 & 242 & 20 & 3 & 180 & 16 & 3 & 121 & 14 & 2 & 61 & 14 & 2 & 50 \\\cline{2-20}
\hlineB{3}
\multirow{7}{*}{CT} & 0.5 & 67 & 19 & 1680 & 91 & 86 & 1680 & 116 & 101 & 1680 & 94 & 87 & 1680 & 134 & 131 & 1680 & 94 & 94 & 1680 \\\cline{2-20}
 & 1.0 & 44 & 9 & 1680 & 47 & 40 & 1680 & 42 & 35 & 1680 & 47 & 30 & 1680 & 92 & 20 & 659 & 76 & 16 & 565 \\\cline{2-20}
 & 2.0 & 29 & 5 & 1680 & 29 & 28 & 1680 & 73 & 15 & 938 & 55 & 12 & 641 & 50 & 10 & 281 & 44 & 9 & 263 \\\cline{2-20}
 & 3.0 & 24 & 3 & 1680 & 52 & 10 & 835 & 53 & 9 & 630 & 42 & 8 & 415 & 37 & 6 & 189 & 34 & 5 & 176 \\\cline{2-20}
 & 5.0 & 18 & 2 & 1680 & 37 & 7 & 495 & 37 & 6 & 373 & 30 & 5 & 256 & 27 & 4 & 111 & 25 & 4 & 104 \\\cline{2-20}
 & 8.0 & 14 & 2 & 1680 & 27 & 4 & 315 & 28 & 4 & 230 & 23 & 3 & 165 & 20 & 2 & 85 & 19 & 2 & 73 \\\cline{2-20}
 & 10.0 & 13 & 1 & 1680 & 24 & 3 & 257 & 24 & 3 & 199 & 20 & 2 & 135 & 17 & 2 & 66 & 16 & 2 & 61 \\\cline{2-20}
\hlineB{3}
\multirow{7}{*}{GA} & 0.5 & 114 & 19 & 5008 & 152 & 130 & 5008 & 186 & 189 & 5008 & 153 & 149 & 5008 & 216 & 211 & 5008 & 253 & 55 & 1746 \\\cline{2-20}
 & 1.0 & 73 & 11 & 5008 & 182 & 30 & 3278 & 198 & 45 & 2478 & 150 & 30 & 1481 & 135 & 23 & 789 & 118 & 17 & 659 \\\cline{2-20}
 & 2.0 & 50 & 5 & 5008 & 104 & 21 & 1372 & 111 & 16 & 1041 & 86 & 13 & 691 & 75 & 11 & 355 & 69 & 11 & 303 \\\cline{2-20}
 & 3.0 & 39 & 5 & 5008 & 79 & 11 & 909 & 85 & 10 & 662 & 63 & 8 & 502 & 57 & 8 & 237 & 53 & 8 & 203 \\\cline{2-20}
 & 5.0 & 30 & 3 & 5008 & 57 & 8 & 539 & 59 & 6 & 413 & 46 & 5 & 298 & 42 & 4 & 151 & 38 & 5 & 138 \\\cline{2-20}
 & 8.0 & 23 & 3 & 5008 & 41 & 5 & 360 & 43 & 4 & 284 & 35 & 4 & 192 & 32 & 3 & 100 & 30 & 3 & 92 \\\cline{2-20}
 & 10.0 & 21 & 2 & 5008 & 36 & 4 & 297 & 40 & 3 & 210 & 31 & 3 & 158 & 28 & 2 & 85 & 25 & 3 & 78 \\\cline{2-20}
\hlineB{3}
\multirow{7}{*}{FL} & 0.5 & 163 & 24 & 11005 & 220 & 211 & 11005 & 266 & 230 & 11005 & 402 & 96 & 5139 & 497 & 122 & 4023 & 330 & 79 & 1973 \\\cline{2-20}
 & 1.0 & 107 & 12 & 11005 & 246 & 50 & 3320 & 265 & 55 & 2648 & 199 & 27 & 1691 & 185 & 30 & 873 & 162 & 22 & 719 \\\cline{2-20}
 & 2.0 & 71 & 9 & 11005 & 146 & 19 & 1452 & 148 & 18 & 1167 & 115 & 16 & 781 & 108 & 11 & 387 & 98 & 9 & 355 \\\cline{2-20}
 & 3.0 & 57 & 6 & 11005 & 107 & 11 & 1008 & 111 & 11 & 782 & 89 & 10 & 523 & 82 & 7 & 259 & 76 & 6 & 237 \\\cline{2-20}
 & 5.0 & 44 & 3 & 11005 & 75 & 8 & 619 & 81 & 8 & 461 & 65 & 8 & 311 & 59 & 5 & 165 & 56 & 5 & 138 \\\cline{2-20}
 & 8.0 & 34 & 3 & 11005 & 58 & 5 & 381 & 60 & 4 & 320 & 48 & 4 & 219 & 44 & 4 & 110 & 41 & 4 & 110 \\\cline{2-20}
 & 10.0 & 31 & 2 & 11005 & 48 & 5 & 331 & 53 & 5 & 250 & 41 & 4 & 189 & 40 & 3 & 85 & 37 & 3 & 85 \\\cline{2-20}
 \hlineB{3}
\multirow{7}{*}{CA} & 0.5 & 197 & 24 & 16141 & 595 & 162 & 11068 & 320 & 279 & 16141 & 464 & 122 & 5139 & 600 & 122 & 4023 & 381 & 65 & 2266 \\\cline{2-20}
 & 1.0 & 129 & 12 & 16141 & 272 & 56 & 3546 & 310 & 46 & 2841 & 239 & 30 & 1691 & 223 & 30 & 873 & 190 & 30 & 719 \\\cline{2-20}
 & 2.0 & 86 & 9 & 16141 & 160 & 21 & 1598 & 174 & 16 & 1240 & 139 & 16 & 781 & 127 & 16 & 387 & 113 & 12 & 387 \\\cline{2-20}
 & 3.0 & 69 & 6 & 16141 & 125 & 15 & 1008 & 131 & 10 & 830 & 103 & 9 & 572 & 94 & 11 & 284 & 89 & 10 & 237 \\\cline{2-20}
 & 5.0 & 54 & 3 & 16141 & 88 & 8 & 649 & 95 & 8 & 489 & 74 & 7 & 355 & 70 & 5 & 182 & 66 & 5 & 165 \\\cline{2-20}
 & 8.0 & 41 & 3 & 16141 & 65 & 4 & 443 & 73 & 4 & 320 & 56 & 5 & 230 & 52 & 5 & 123 & 49 & 5 & 110 \\\cline{2-20}
 & 10.0 & 37 & 2 & 16141 & 57 & 5 & 349 & 62 & 5 & 266 & 50 & 4 & 180 & 47 & 3 & 94 & 44 & 3 & 94 \\\cline{2-20}
\hlineB{3}
\end{tabular}
}
\caption{Direct selection concrete parameters of accurate versus inaccurate manifest. $\Delta=.1$ is used to compute time savings in  Figure~\ref{fig:inacc direct savings}. All sizes are in thousands.}
\label{tab:inacc direct savings}
\end{table*}

\begin{table*}[t!]
\centering
\begin{tabular}{V{3}rV{3}r|rV{3}r|rV{3}r|rV{3}r|rV{3}}
\hlineB{3}
$\mu$ & \multicolumn{2}{cV{3}}{Accurate} & \multicolumn{2}{cV{3}}{$\Delta_{Upp}=0$} & \multicolumn{2}{cV{3}}{$\Delta=0.5$} & \multicolumn{2}{cV{3}}{$\Delta=0.1$} \\
$(\%)$ & $k_{\mathrm{comp}}$ & $k_{\mathrm{size}}$ & $k_{\mathrm{comp}}$ & $k_{\mathrm{size}}$ & $k_{\mathrm{comp}}$ & $k_{\mathrm{size}}$ & $k_{\mathrm{comp}}$ & $k_{\mathrm{size}}$ \\
\hlineB{3}
0.5 & 2537 & 16141 & 2537 & 0 & 17139 & 724 & 9221 & 216 \\\hline
1.0 & 1020 & 16141 & 1020 & 0 & 6049 & 269 & 3496 & 73 \\\hline
2.0 & 395 & 16141 & 395 & 0 & 2874 & 114 & 1352 & 34 \\\hline
3.0 & 263 & 16141 & 263 & 0 & 1469 & 76 & 836 & 21 \\\hline
5.0 & 157 & 16141 & 157 & 0 & 841 & 44 & 529 & 11 \\\hline
8.0 & 81 & 16141 & 81 & 0 & 563 & 27 & 563 & 6 \\\hline
10.0 & 65 & 16141 & 65 & 0 & 563 & 19 & 563 & 5 \\\hline
\hlineB{3}
\end{tabular}
\caption{Ballot comparison concrete parameters of accurate versus inaccurate manifest for the CA size. $k_\mathrm{size}$ is in thousands. The only size dependence on these methods is when to resort to a full count of size or a full hand recount. $\Delta=.1$ is used to compute time savings in  Figure~\ref{fig:comp inacc savings}.}
\label{tab:comp inacc savings}
\end{table*}

\begin{table*}[t!]
\centering
\resizebox{.9\textwidth}{!}{
\begin{tabular}{V{3}rV{3}r|rV{3}r|rV{3}r|rV{3} L{1.5cm} | L{1.5cm} V{3}r|rV{3} L{1.5cm} | L{1.5cm} V{3}}
\hlineB{3}
$\mu$ & \multicolumn{2}{cV{3}}{Accurate} & \multicolumn{2}{cV{3}}{$\Delta_{Upp}=0$} & \multicolumn{2}{cV{3}}{$\Delta=0.5$} & \multicolumn{2}{cV{3}}{$\Delta_{Low}=0.5, \Delta_{Upp}=0$} & \multicolumn{2}{cV{3}}{$\Delta=0.1$} & \multicolumn{2}{cV{3}}{$\Delta_{Low}=0.1, \Delta_{Upp}=0$} \\
$(\%)$ & $k_{\mathrm{comp}}$ & $k_{\mathrm{size}}$ & $k_{\mathrm{comp}}$ & $k_{\mathrm{size}}$ & $k_{\mathrm{comp}}$ & $k_{\mathrm{size}}$ & $k_{\mathrm{comp}}$ & $k_{\mathrm{size}}$ & $k_{\mathrm{comp}}$ & $k_{\mathrm{size}}$ & $k_{\mathrm{comp}}$ & $k_{\mathrm{size}}$ \\
\hlineB{3}
0.5 & 355 & 16141 & 783 & 16141 & 1105 & 16141 & 1112 & 7132 & 1302 & 4023 & 958 & 3203 \\\hline
1.0 & 89 & 16141 & 297 & 4028 & 261 & 2841 & 225 & 2071 & 199 & 1107 & 176 & 977 \\\hline
2.0 & 23 & 16141 & 81 & 1435 & 67 & 938 & 57 & 748 & 46 & 355 & 43 & 327 \\\hline
3.0 & 10 & 16141 & 40 & 826 & 31 & 547 & 27 & 430 & 22 & 189 & 18 & 218 \\\hline
5.0 & 4 & 16141 & 17 & 428 & 13 & 276 & 11 & 216 & 8 & 98 & 8 & 87 \\\hline
8.0 & 2 & 16141 & 8 & 236 & 6 & 148 & 6 & 113 & 4 & 51 & 3 & 54 \\\hline
10.0 & 1 & 16141 & 6 & 180 & 5 & 110 & 4 & 85 & 3 & 36 & 3 & 35 \\\hline
\hlineB{3}
\end{tabular}
}
\caption{Polling concrete parameters of accurate versus inaccurate manifest for the CA size. Sizes are in thousands. The only size dependence on these methods is when to resort to a full count of size or a full hand recount. $\Delta=.1$ is used to compute time savings in Figure~\ref{fig:inacc polling savings}.}
\label{tab:polling inacc savings}
\end{table*}
\fi
\ifnum\conference=0
\section{Duplicate detection}
    \label{sec:dup detection}

    We return to the proof of Lemma~\ref{lem:dup-advertise}. We first establish a bound for independent sampling with replacement, then transfer it to sequential sampling without replacement by a coupling argument.
    
\begin{theorem}[Duplicate-detection bound under i.i.d.\ sampling]
\label{thm:dup-iid-simple}
Let \(A\) be a finite set with \(|A|=n\), let \(f:A\to B\) be a function with
\(|B|=n-\ell\), and let \(\Pi\) be a probability distribution on \(A\) such that
\[
\Pi(x)\ge \frac{1}{n(1+c)}
\qquad\text{for all }x\in A,
\]
for some \(c\ge 0\).

Let \(Y_1,\dots,Y_k\) be i.i.d.\ from \(\Pi\), and define the no-collision event
\[
E_k^{\mathrm{iid}}
:=
\Bigl\{
\forall i\neq j\le k,\ f(Y_i)=f(Y_j)\implies Y_i=Y_j
\Bigr\}.
\]
Thus \(E_k^{\mathrm{iid}}\) is the event that no duplicate identifier is
detected among the \(k\) draws, where repeated draws of the \emph{same}
element are not counted as collisions. If \(k\le n(1+c)\) then
\begin{align}
\nonumber
\Pr(E_k^{\mathrm{iid}}) &\le 2\exp\!\left(-\ell \left[1 - \exp\left(\frac{-k}{n(1+c)}\right)\right]^2\right)\\
\label{eq:dup-iid}
&\leq 2\exp\!\left(-\frac{k^2\ell}{4n^2(1+c)^2}\right).
\end{align}
\end{theorem}

\begin{proof}
For each \(b\in B\), let
$F_b:=f^{-1}(b)$ and $
s_b:=|F_b|.$ 
Let
$B_+:=\{b\in B:s_b>0\}$ 
be the set of nonempty fibers. Then $\sum_{b\in B_+} s_b=n,$
\begin{align*}
\sum_{b\in B_+}(s_b-1)
=
n-|B_+|
\ifnum\conference=0
\\
\fi
=
n-|\operatorname{im}(f)|
\ge
n-|B|
=
\ell.
\end{align*}
Define
\[
\tilde\ell:=\sum_{b\in B_+}(s_b-1),
\]
so that \(\tilde\ell\ge \ell\). Further define
\[
  L := \frac{1}{n(1+c)}\,.
\]
The inequality can be strict when \(f\) is not surjective onto \(B\).
\medskip\noindent
\textbf{Step 1: Poissonization.}
Let \(N\sim \mathrm{Pois}(k)\), independent of everything else, and let
\(Z_1,\dots,Z_N\) be i.i.d.\ from \(\Pi\). Let \(E^{\mathrm P}\) denote the
event that no duplicate identifier is detected among these \(N\) draws.
For each \(x\in A\), let
$M_x:=\#\{i:Z_i=x\}.$ 
Then the random variables \((M_x)_{x\in A}\) are independent and satisfy
\[
M_x\sim \mathrm{Pois}(k\Pi(x)).
\]
Now define
\[
I_x:=\mathbf 1_{\{M_x>0\}}.
\]
Then the \(I_x\) are independent Bernoulli random variables with
\[
\Pr[I_x=1] =1-e^{-k\Pi(x)}\ge 1-e^{-kL}=:r.
\]

Fix \(b\in B_+\). The event that the fiber \(F_b\) causes no detected collision
is
\[
A_b:=\left\{\sum_{x\in F_b} I_x\le 1\right\}.
\]
This is a coordinatewise decreasing event in the family \((I_x)_{x\in F_b}\).
Since each \(I_x\) stochastically dominates a Bernoulli\((r)\), monotonicity of
product measures implies
\[
\Pr[A_b]\le \Pr\!\bigl[\mathrm{Bin}(s_b,r)\le 1\bigr].
\]
Now
\begin{align*}
\Pr\!\bigl[\mathrm{Bin}(s_b,r)\le 1\bigr]
&=(1-r)^{s_b}+s_b r(1-r)^{s_b-1}\\
&=(1-r)^{s_b-1}\bigl(1+(s_b-1)r\bigr).
\end{align*}
Using \(1+mr\le (1+r)^m\) for every integer \(m\ge 0\), we obtain
\[
\Pr[A_b]
\le
(1-r)^{s_b-1}(1+r)^{s_b-1}
=
(1-r^2)^{s_b-1}.
\]

Because different fibers involve disjoint collections of independent variables,
the events \((A_b)_{b\in B_+}\) are independent. Therefore
\[
\Pr\left[E^{\mathrm P}\right]
=
\prod_{b\in B_+}\Pr[A_b]
\le
\prod_{b\in B_+}(1-r^2)^{s_b-1}
=
(1-r^2)^{\tilde\ell}.
\]
Using \(1-u\le e^{-u}\) and \(\tilde\ell\ge \ell\),
\[
\Pr[E^{\mathrm P}]
\le
e^{-\tilde\ell r^2}
\le
e^{-\ell r^2}
=
\exp\!\bigl(-\ell(1-e^{-kL})^2\bigr).
\]

\medskip\noindent
\textbf{Step 2: De-Poissonization.}
For each \(m\ge 0\), let
\[
p_m:=\Pr\left[E_m^{\mathrm{iid}}\right],
\]
where \(E_m^{\mathrm{iid}}\) denotes the no-collision event for exactly \(m\)
i.i.d.\ draws from \(\Pi\). The sequence \((p_m)\) is nonincreasing in \(m\),
since adding more draws can only create more opportunities for collision.
Because \(N\sim \mathrm{Pois}(k)\),
\[
\Pr[E^{\mathrm P}]=\sum_{m\ge 0} p_m\,\Pr[N=m].
\]
Since \(p_m\ge p_k\) for all \(m\le k\), it follows that
\[
\Pr[E^{\mathrm P}]\ge p_k\,\Pr[N\le k].
\]
Hence
\[
p_k\le \frac{\Pr[E^{\mathrm P}]}{\Pr[N\le k]}.
\]
Combining this with the bound from Step 1 yields
\[
\Pr[E_k^{\mathrm{iid}}]
=
p_k
\le
\frac{\exp\!\bigl(-\ell(1-e^{-kL})^2\bigr)}
     {\Pr[\mathrm{Pois}(k)\le k]}\,.
\]

Recalling that \(\Pr[\mathrm{Pois}(k)\le k]\ge 1/2\), we conclude that
\[
\Pr[E_k^{\mathrm{iid}}]
=
p_k
\le
2 \exp\!\bigl(-\ell(1-e^{-kL})^2\bigr)\,.
\]
Finally, if \(kL\le 1\), then
\[
1-e^{-kL}\ge \frac{kL}{2}.
\]
Thus
\[
\Pr[E_k^{\mathrm{iid}}]
\le
2\exp\!\left(-\ell\left(\frac{kL}{2}\right)^2\right)
=
2\exp\!\left(-\frac{k^2\ell}{4n^2(1+c)^2}\right),
\]
which is \eqref{eq:dup-iid}.
\end{proof}

We now prove Lemma~\ref{lem:dup-advertise} as a corollary of the previous theorem.

\begin{corollary}[Reduction to sequential weighted sampling without replacement]
\label{cor:dup-seq}
Assume the hypotheses of Theorem~\ref{thm:dup-iid-simple}. Let
\(X_1,\dots,X_k\) be sampled without replacement according to the sequential law
\(
X_1\sim \Pi,
\)
and for \(t\ge 1\),
\[
\Pr[X_{t+1}=x\mid X_1,\dots,X_t]
=
\frac{\Pi(x)}
     {1-\sum_{j=1}^t \Pi[X_j]}
\]
for every $x\notin\{X_1,\dots,X_t\}$.
Define
\[
E_k^{\mathrm{seq}}
:=
\{\forall i\neq j,\ f(X_i)\neq f(X_j)\}.
\]
Then, for $k \leq n$,
\begin{align}
\label{eq:dup-seq-strong}
\Pr[E_k^{\mathrm{seq}}] &\le
2\exp\!\left(-\ell \left[1 - \exp\left(\frac{-k}{n(1+c)}\right)\right]^2\right)\\
\label{eq:dup-seq}
&\le 2\exp\!\left(-\frac{k^2\ell}{4n^2(1+c)^2}\right).
\end{align}
\end{corollary}

\begin{proof}
Let \(Y_1,Y_2,\dots\) be an infinite i.i.d.\ sequence with law \(\Pi\). Define
stopping times $\tau_1:=1,$
\[
\tau_{m+1}:=\min\{t>\tau_m:\ Y_t\notin\{Y_{\tau_1},\dots,Y_{\tau_m}\}\}.
\]
Since \(\Pi(x)>0\) for every \(x\in A\), these stopping times are almost surely
finite for \(m\le n\). Set
\[
\widetilde X_m:=Y_{\tau_m},
\qquad m=1,\dots,k.
\]
Then \((\widetilde X_1,\dots,\widetilde X_k)\) has exactly the sequential
without-replacement law above: conditional on
\(\widetilde X_1=x_1,\dots,\widetilde X_t=x_t\), the next accepted value is the
first future i.i.d.\ draw that lands outside \(\{x_1,\dots,x_t\}\), so for
\(x\notin\{x_1,\dots,x_t\}\),
\begin{align*}
&\Pr\left[\widetilde X_{t+1}=x\mid
\widetilde X_1=x_1,\dots,\widetilde X_t=x_t\right]\\
&=
\sum_{m\ge 1}\left(\sum_{j=1}^t\Pi(x_j)\right)^{m-1}\Pi(x)\\
&=
\frac{\Pi(x)}{1-\sum_{j=1}^t \Pi(x_j)}.
\end{align*}
Now, if \(E_k^{\mathrm{iid}}\) fails for the first \(k\) i.i.d.\ draws
\(Y_1,\dots,Y_k\), then there exist \(i\neq j\le k\) such that \(Y_i\neq Y_j\)
but \(f(Y_i)=f(Y_j)\). Those two distinct elements have both appeared by time
\(k\), so they are among the first \(k\) distinct values
\(\widetilde X_1,\dots,\widetilde X_k\). Hence \(E_k^{\mathrm{seq}}\) also
fails for \(\widetilde X_1,\dots,\widetilde X_k\). Therefore
\[
E_k^{\mathrm{seq}}(\widetilde X_1,\dots,\widetilde X_k)
\subseteq
E_k^{\mathrm{iid}}(Y_1,\dots,Y_k),
\]
which gives
\[
\Pr[E_k^{\mathrm{seq}}]\le \Pr[E_k^{\mathrm{iid}}].
\]
The stated bounds now follow from Theorem~\ref{thm:dup-iid-simple}.
\end{proof}
\noindent
Lemma~\ref{lem:dup-advertise} follows immediately.
\section{General Sampling Bounds from Size Ratio and Total Variation}
\label{app:reweighting}

This appendix isolates the abstract analytic statement needed by the auditor.
The main theorem only uses two manifest-derived quantities: a size ratio
\(\lambda:=M/N\), where \(M=|T|\) and \(N=|\Ballots|\), and a total variation
bound \(\eta\) between the induced ballot-sampling law \(\pi\) and the
uniform ballot law \(U\). The results below assume that the auditor has a
lower bound on the size ratio and a total variation bound between the induced
sampling law and \(U\); the directional manifest calculations that supply
these quantities appear in Appendix~\ref{app:auditor-analysis}.

\subsection{Total variation and normalized reweighting}

\begin{definition}[Total variation distance]\label{def:tv}
Let \(P\) and \(Q\) be probability distributions on a finite set \(X\). The
total variation distance between \(P\) and \(Q\) is
\[
  d_{\mathrm{TV}}(P,Q)
  :=
  \frac12\sum_{x\in X}|P(x)-Q(x)|.
\]
\end{definition}

\begin{lemma}[Bounded functions are stable under total variation]\label{lem:tv-bounded}
Let \(f:X\to\mathbb{R}\) satisfy \(|f(x)|\le M_f\) for all \(x\in X\). Then for any
probability distributions \(P,Q\) on \(X\),
\[
  \bigl|\mathbb{E}[f(X_P)]-\mathbb{E}[f(X_Q)]\bigr|
  \le 2M_f\,d_{\mathrm{TV}}(P,Q),
\]
where \(X_P\sim P\) and \(X_Q\sim Q\).
\end{lemma}

\begin{proof}
By the triangle inequality and the bound $|f(x)|\le M_f$:
\begin{align*}
  \bigl|\mathbb{E}[f(X_P)]-\mathbb{E}[f(X_Q)]\bigr|
  &=\Bigl|\sum_{x\in X}(P(x)-Q(x))f(x)\Bigr|
  \\&\le \sum_{x\in X}|P(x)-Q(x)|\,|f(x)|,
\end{align*}
which is at most \(2M_f d_{\mathrm{TV}}(P,Q)\).
\end{proof}

\begin{lemma}[General reweighting estimates]
\label{lem:general-reweight}
Let \(P\) be a probability distribution on a finite set \(X\), let
\(w:X\to(0,\infty)\), and define
\[
m:=\sum_{x\in X}P(x)w(x),
\qquad
Q(x):=\frac{P(x)w(x)}{m}.
\]
Then
\[
d_{\mathrm{TV}}(P,Q)
=
\frac{1}{2m}\sum_{x\in X}P(x)\,|w(x)-m|.
\]
Moreover,
\[
d_{\mathrm{TV}}(P,Q)
\le
\sum_{x\in X}P(x)\,|w(x)-1|.
\]
\end{lemma}

\begin{proof}
Since \(Q(x)=P(x)w(x)/m\),
\begin{align*}
d_{\mathrm{TV}}(P,Q)
&=
\frac12\sum_{x\in X}\left|P(x)-Q(x)\right|\\
&=
\frac12\sum_{x\in X}P(x)\left|1-\frac{w(x)}{m}\right|\\
&=
\frac{1}{2m}\sum_{x\in X}P(x)\,|w(x)-m|.
\end{align*}
This proves the identity. For the inequality, define the finite measure
\(R(x):=P(x)w(x)\), so that \(\sum_xR(x)=m\) and \(Q=R/m\). By the triangle
inequality,
\[
\|P-Q\|_1\le \|P-R\|_1+\|R-Q\|_1.
\]
Now
\[
\|P-R\|_1=\sum_{x\in X}P(x)|1-w(x)|
\]
and
\[
\|R-Q\|_1
=
\left|1-\frac1m\right|\sum_{x\in X}R(x)
=|m-1|.
\]
Also,
\[
|m-1|
\le
\sum_{x\in X}P(x)|w(x)-1|.
\]
Hence
\[
2d_{\mathrm{TV}}(P,Q)=\|P-Q\|_1
\le
2\sum_{x\in X}P(x)|w(x)-1|,
\]
which proves the result.
\end{proof}

\begin{corollary}[Total variation under bounded reweighting]
\label{cor:sharp-reweight}
Let \(P\) be a probability distribution on a finite set \(X\), let
\(w:X\to(0,\infty)\), and define
\[
m:=\sum_{x\in X}P(x)w(x),
\qquad
Q(x):=\frac{P(x)w(x)}{m}.
\]
If, for some \(0<\alpha\le \beta\),
\[
\alpha\le w(x)\le \beta
\qquad\text{for all }x\in X,
\]
then
\[
d_{\mathrm{TV}}(P,Q)
\le
\frac{\sqrt{\beta}-\sqrt{\alpha}}
     {\sqrt{\beta}+\sqrt{\alpha}}.
\]
In particular, if
\[
\frac{1}{1+\eta}\le w(x)\le 1+\eta
\qquad\text{for all }x\in X,
\]
then
\[
d_{\mathrm{TV}}(P,Q)\le \frac{\eta}{2+\eta}.
\]
\end{corollary}

\begin{proof}
If \(\alpha=\beta\), then \(w\) is constant and \(P=Q\), so the claim is
immediate. Hence assume \(\alpha<\beta\).
By Lemma~\ref{lem:general-reweight},
\[
d_{\mathrm{TV}}(P,Q)
=
\frac{1}{2m}\sum_{x\in X}P(x)\,|w(x)-m|.
\]
Also \(m\in[\alpha,\beta]\), since \(m\) is a \(P\)-average of values in
\([\alpha,\beta]\). Fix \(m\) in this interval. Convexity of
\(y\mapsto |y-m|\) gives
\[
|y-m|
\le
\frac{\beta-y}{\beta-\alpha}(m-\alpha)
+
\frac{y-\alpha}{\beta-\alpha}(\beta-m).
\]
Averaging with respect to \(P\) and using
\(\sum_xP(x)w(x)=m\), we obtain
\[
\sum_{x\in X}P(x)|w(x)-m|
\le
\frac{2(\beta-m)(m-\alpha)}{\beta-\alpha}.
\]
Consequently,
\[
d_{\mathrm{TV}}(P,Q)
\le
\frac{(\beta-m)(m-\alpha)}{m(\beta-\alpha)}.
\]
The right-hand side is maximized at \(m=\sqrt{\alpha\beta}\), where it equals
\[
\frac{\sqrt{\beta}-\sqrt{\alpha}}
     {\sqrt{\beta}+\sqrt{\alpha}}.
\]
For the special case \(\alpha=1/(1+\eta)\) and \(\beta=1+\eta\), this
simplifies to \(\eta/(2+\eta)\).
\end{proof}

\begin{lemma}[Reweighting confined to an exceptional set]
\label{lem:exceptional-reweight}
Let \(P\) be a probability distribution, let \(B\) be an event with
\(P(B)=s\), and let
\[
  Q(x)=\frac{P(x)h(x)}{\sum_yP(y)h(y)}.
\]
If \(h=1\) outside \(B\) and
\[
  1-\ell\le h\le 1+u
  \qquad\text{on }B,
\]
where \(0\le\ell\le1\) and \(u\ge0\), then
\[
  d_{\mathrm{TV}}(P,Q)\le s\max\{\ell,u\}.
\]
\end{lemma}

\begin{proof}
Write \(z:=\sum_xP(x)h(x)\) and \(D:=\max\{\ell,u\}\). If \(z\ge1\), then
\((h-z)_+\) vanishes outside \(B\), and on \(B\) it is at most \(u\). Hence
\[
  d_{\mathrm{TV}}(P,Q)
  =\frac{1}{z}\sum_xP(x)(h(x)-z)_+
  \le sD.
\]
If \(z<1\), then \((z-h)_+\) vanishes outside \(B\). Moreover,
\(h\ge1-\ell\ge z(1-\ell)\), so \((z-h)_+\le z\ell\) on \(B\). Thus
\[
  d_{\mathrm{TV}}(P,Q)
  =\frac{1}{z}\sum_xP(x)(z-h(x))_+
  \le sD.
\]
\end{proof}

\begin{lemma}[Lifting batch laws preserves total variation]
\label{lem:batch-ballot-tv}
Let \(\Ballots\) be partitioned into nonempty batches
\(\Ballots_1,\ldots,\Ballots_K\), with \(n_i:=|\Ballots_i|\). For a
probability law \(P\) on \([K]\), let \(\pi_P\) be the ballot law that first
draws batch \(i\) according to \(P\) and then draws uniformly from
\(\Ballots_i\). For any batch laws \(P,Q\),
\[
d_{\mathrm{TV}}(\pi_P,\pi_Q)=d_{\mathrm{TV}}(P,Q).
\]
In particular, if \(Q_i=n_i/N\), where \(N=\sum_i n_i\), then
\(\pi_Q=U\), the uniform law on \(\Ballots\).
\end{lemma}

\begin{proof}
For every \(b\in\Ballots_i\), we have
\(\pi_P(b)=P_i/n_i\) and \(\pi_Q(b)=Q_i/n_i\). Therefore
\begin{align*}
d_{\mathrm{TV}}(\pi_P,\pi_Q)
&=\frac12\sum_{i=1}^K\sum_{b\in\Ballots_i}
  \left|\frac{P_i}{n_i}-\frac{Q_i}{n_i}\right|\\
&=\frac12\sum_{i=1}^K|P_i-Q_i|
=d_{\mathrm{TV}}(P,Q).
\end{align*}
For \(Q_i=n_i/N\), every ballot has probability
\((n_i/N)/n_i=1/N\), proving the final assertion.
\end{proof}

\subsection{Discrepancy under a nonuniform ballot law}
Fix a tabulation \(T\) and let \(\Ballots\) be the underlying ballot set. As in the
main text, let \(N:=|\Ballots|\) and \(M:=|T|\), and write
\[
  \lambda:=\frac{M}{N}.
\]
Let \(\disc_T(\ballot)\) denote the signed ballot discrepancy. As before,
\(\disc_T(\ballot)\in[-2,2]\) for every ballot \(\ballot\).

Let \(\kappa\) denote the excess identifier multiplicity rate, so that \(\kappa N\)
is the number of ballots discarded after keeping at most one ballot per label.

\begin{lemma}[Uniform-ballot lower bound]\label{lem:uniform-ballot}
Under the uniform ballot law \(U\),
\[
  \mathbb{E}[\disc_T(\ballot_U)]
  \ge
  \frac{\disc}{N}
  -2\kappa
  -(\lambda-1)_+.
\]
\end{lemma}

%
%
%

\begin{proof}
Let $\mathcal R$ be the set of all CVR rows in $T$, so that
$|\mathcal R|=M$. Write $c_r:=\winner_r-\loser_r$ for the reported
vote difference in row $r$, and
$a_{\ballot}:=\winner_{\ballot}-\loser_{\ballot}$ for the actual
vote difference on ballot $\ballot$.

Let $H\subseteq\mathcal R$ be the set of rows whose identifiers
appear on physical ballots. For each $r\in H$, let $t_r\ge1$
be the number of ballots carrying that identifier. Define
\[
  h:=|H|,
  \qquad
  d:=\sum_{r\in H}(t_r-1),
  \qquad
  u:=N-h-d.
\]
Thus $u$ counts ballots whose identifiers are absent from the CVR,
including unlabeled ballots. Since $d$ counts excess multiplicity
over a subset of the physical identifiers, $d\le\kappa N$.

By the definition of ballot discrepancy, every unmatched physical
ballot is assigned reported vote difference $+1$. Therefore
\[
  \sum_{\ballot\in\Ballots}\disc_T(\ballot)
  =
  \sum_{r\in H}t_r c_r+u
  -\sum_{\ballot\in\Ballots}a_{\ballot}.
\]
Subtracting
\[
  \disc
  =
  \sum_{r\in\mathcal R}c_r
  -\sum_{\ballot\in\Ballots}a_{\ballot}
\]
gives
\begin{align*}
  \sum_{\ballot\in\Ballots}\disc_T(\ballot)-\disc
  &=
  \sum_{r\in H}(t_r-1)c_r
  +u-\sum_{r\in\mathcal R\setminus H}c_r\\
  &\ge -d+u-(M-h)\\
  &=N-M-2d,
\end{align*}
where the inequality uses $-1\le c_r\le1$.
Dividing by $N$ and using $d\le\kappa N$ yields
\begin{align*}
  \mathbb{E}[\disc_T(\ballot_U)]
  &\ge \frac{\disc}{N}+1-\lambda-2\kappa\\
  &\ge \frac{\disc}{N}-2\kappa-(\lambda-1)_+,
\end{align*}
as desired.
\end{proof}

\begin{theorem}[General discrepancy bound from size ratio and TV]\label{thm:general-ballot-law}
Let \(\pi\) be any probability distribution on \(\Ballots\). Then
\[
  \mathbb{E}[\disc_T(\ballot_\pi)]
  \ge
  \frac{\disc}{N}
  -2\kappa
  -(\lambda-1)_+
  -4\,d_{\mathrm{TV}}(\pi,U).
\]
If the reported outcome is wrong, then
\[
  \mathbb{E}[\disc_T(\ballot_\pi)]
  \ge
  \mu^{\tab}\lambda
  -2\kappa
  -(\lambda-1)_+
  -4\,d_{\mathrm{TV}}(\pi,U).
\]
Consequently, if
\[
  \frac{1}{1+\varepsilon_-}\le \lambda\le 1+\varepsilon_+
  \qquad\text{and}\qquad
  d_{\mathrm{TV}}(\pi,U)\le \gamma,
\]
then, for an invalid election,
\[
  \mathbb{E}[\disc_T(\ballot_\pi)]
  \ge
  \min\!\left\{
    \frac{\mu^{\tab}}{1+\varepsilon_-},
    \mu^{\tab}(1+\varepsilon_+) - \varepsilon_+
  \right\}
  -2\kappa-4\gamma.
\]
\end{theorem}

\begin{proof}
Apply Lemma~\ref{lem:tv-bounded} with \(f=\disc_T\) and \(M_f=2\) to compare the
expectation under \(\pi\) with the expectation under \(U\). The uniform lower
bound from Lemma~\ref{lem:uniform-ballot} then gives
\[
  \mathbb{E}[\disc_T(\ballot_\pi)]
  \ge
  \frac{\disc}{N}
  -2\kappa
  -(\lambda-1)_+
  -4\,d_{\mathrm{TV}}(\pi,U).
\]
If the reported outcome is wrong, replace \(\disc/N\) by \(\mu^{\tab}\lambda\).
To obtain the endpoint form, write
\[
  g(\lambda):=\mu^{\tab}\lambda-(\lambda-1)_+.
\]
This function is linear on \([1/(1+\varepsilon_-),1]\) and on
\([1,1+\varepsilon_+]\), so its minimum over
\([1/(1+\varepsilon_-),1+\varepsilon_+]\) is attained at one of the two
endpoints. Evaluating there gives
\begin{align*}
  g\!\left(\frac{1}{1+\varepsilon_-}\right)
  &= \frac{\mu^{\tab}}{1+\varepsilon_-}\\
  g(1+\varepsilon_+)
  &=
  \mu^{\tab}(1+\varepsilon_+)-\varepsilon_+,
\end{align*}
which yields the displayed bound.
\end{proof}

\section{Directional Manifest Analysis}
\label{app:auditor-analysis}

This section proves Theorem~\ref{thm:disc-lb-mostly-good-manifest} by combining the abstract reweighting bound from \appref{app:reweighting} with the directional batchwise theorem proved below. The duplicate-rate term is handled exactly as in the main text: once the duplicate detector certifies that \(\kappa\) is small enough, the remaining discrepancy loss is controlled by the manifest certificate.

\subsection{Directional manifest certificate}
\label{app:directional-certificate}

Let \(n_1,\dots,n_K>0\) be actual batch sizes and \(m_1,\dots,m_K>0\) be
reported batch sizes. Define
\[
N:=\sum_{i=1}^K n_i,
\quad
M:=\sum_{i=1}^K m_i,
\quad
r_i:=\frac{m_i}{M},
\quad
\nu_i:=\frac{n_i}{N}.
\]
Let \(\Bad\subseteq[K]\), and define its \emph{tabulated mass} by
\[
p:=\sum_{i\in\Bad}r_i.
\]
Assume directional parameters
\begin{align*}
0&\le \delta_{\mathrm{lower}}\le \Delta_{\mathrm{lower}}\le \infty,
\\
0&\le \delta_{\mathrm{upper}}\le \Delta_{\mathrm{upper}}<\infty.
\end{align*}
We say that the actual batch law \(\nu\) is a directional
\((p,\delta_{\mathrm{lower}},\delta_{\mathrm{upper}},\Delta_{\mathrm{lower}},\Delta_{\mathrm{upper}})\)-asymmetric manifest of the tabulated batch law \(r\) if, with
\[
w(i):=\frac{n_i}{m_i},
\]
the following bounds hold:
\[
\frac{1}{1+\delta_{\mathrm{lower}}}\le w(i)\le 1+\delta_{\mathrm{upper}}
\qquad (i\notin\Bad),
\]
and
\[
\frac{1}{1+\Delta_{\mathrm{lower}}}\le w(i)\le 1+\Delta_{\mathrm{upper}}
\qquad (1\le i\le K).
\]
Consequently,
\[
\nu_i=\frac{r_i\,w(i)}{\sum_{j=1}^K r_jw(j)}
\qquad (1\le i\le K),
\]
and the total tabulated mass of \(\Bad\) is \(p\).
For \(t\ge 0\), define
\[
  \Theta(t):=\frac{\sqrt{1+t}-1}{\sqrt{1+t}+1}.
\]

\begin{theorem}[Directional batchwise asymmetric manifest]
\label{thm:batchwise-smooth}
Under the directional asymmetric manifest certificate above, define
\begin{align*}
\varepsilon_{\mathrm{upper}}(p) &:= (1-p)\delta_{\mathrm{upper}}+p\Delta_{\mathrm{upper}},\\
\eta_{\mathrm{good}}
&:=(1+\delta_{\mathrm{lower}})(1+\delta_{\mathrm{upper}})-1,\\
\chi_{\mathrm{bad}}
&:=\max\left\{
  \Delta_{\mathrm{upper}},
  \frac{\Delta_{\mathrm{lower}}}{1+\Delta_{\mathrm{lower}}}
\right\},\\
\Gamma_{\mathrm{tv}}(p)
&:=\min\left\{
  1,
  \Theta(\eta_{\mathrm{good}})
  +\frac{(1+\delta_{\mathrm{lower}})p}{1+\delta_{\mathrm{lower}}p}
   \chi_{\mathrm{bad}}
\right\},\\
\tau_{\mathrm{lower}}(p)&:=
\left(
  \frac{1-p}{1+\delta_{\mathrm{lower}}}+\frac{p}{1+\Delta_{\mathrm{lower}}}
\right)^{-1}-1
\end{align*}
Then
\[
\frac{1}{1+\varepsilon_{\mathrm{upper}}(p)}
\le
\frac{M}{N}
\le
1+\tau_{\mathrm{lower}}(p),
\]
and
\[
d_{\mathrm{TV}}(r,\nu)
\le
\Gamma_{\mathrm{tv}}(p).
\]
When \(\Delta_{\mathrm{lower}}=\infty\), the ratio
\(\Delta_{\mathrm{lower}}/(1+\Delta_{\mathrm{lower}})\) is interpreted as \(1\).
\end{theorem}

\begin{proof}
By definition,
\[
\sum_i r_iw(i)
=
\sum_i \frac{m_i}{M}\cdot\frac{n_i}{m_i}
=
\frac{N}{M}.
\]
Hence the normalized law induced by the weights is exactly
\[
\nu_i=\frac{r_iw(i)}{\sum_{j=1}^K r_jw(j)}
=
\frac{r_iw(i)}{N/M}
=
\frac{n_i}{N}.
\]
For the size ratio, the upper-side bounds give
\begin{align*}
\frac{N}{M}
&=
\sum_i r_iw(i)
\le
(1-p)(1+\delta_{\mathrm{upper}})+p(1+\Delta_{\mathrm{upper}})\\
&=
1+\varepsilon_{\mathrm{upper}}(p),
\end{align*}
so
\[
\frac{M}{N}\ge \frac{1}{1+\varepsilon_{\mathrm{upper}}(p)}.
\]
The lower-side bounds give
\begin{align*}
\frac{N}{M}
&=
\sum_i r_iw(i)\\
&\ge
\frac{1-p}{1+\delta_{\mathrm{lower}}}
+
\frac{p}{1+\Delta_{\mathrm{lower}}}
=
\frac{1}{1+\tau_{\mathrm{lower}}(p)},
\end{align*}
so
\[
\frac{M}{N}\le 1+\tau_{\mathrm{lower}}(p).
\]
For total variation, set
\[
  a:=\frac{1}{1+\delta_{\mathrm{lower}}},
  \qquad
  b:=1+\delta_{\mathrm{upper}},
\]
and introduce reference weights
\[
  v(i):=
  \begin{cases}
    w(i),&i\notin\Bad,\\
    1,&i\in\Bad.
  \end{cases}
\]
Let \(A:=\sum_i r_iv(i)\) and \(q_i:=r_iv(i)/A\). Since \(a\le v(i)\le b\),
the sharp bounded-reweighting estimate of Corollary~\ref{cor:sharp-reweight}
gives
\begin{align*}
  d_{\mathrm{TV}}(r,q)
  &\le
  \frac{\sqrt b-\sqrt a}{\sqrt b+\sqrt a}\\
  &=\Theta\!\left((1+\delta_{\mathrm{lower}})
                   (1+\delta_{\mathrm{upper}})-1\right)
  =\Theta(\eta_{\mathrm{good}}).
\end{align*}
Also,
\[
  A\ge a(1-p)+p
  =\frac{1+\delta_{\mathrm{lower}}p}{1+\delta_{\mathrm{lower}}},
\]
and hence
\[
  q(\Bad)=\frac{p}{A}
  \le\frac{(1+\delta_{\mathrm{lower}})p}
           {1+\delta_{\mathrm{lower}}p}.
\]
The law \(\nu\) is obtained from \(q\) by reweighting with
\(h(i):=w(i)/v(i)\). This reweighting is \(1\) off \(\Bad\), while on
\(\Bad\),
\[
  1-\frac{\Delta_{\mathrm{lower}}}{1+\Delta_{\mathrm{lower}}}
  \le h(i)\le 1+\Delta_{\mathrm{upper}}.
\]
Lemma~\ref{lem:exceptional-reweight} therefore gives
\[
  d_{\mathrm{TV}}(q,\nu)
  \le
  \frac{(1+\delta_{\mathrm{lower}})p}
       {1+\delta_{\mathrm{lower}}p}\,\chi_{\mathrm{bad}}.
\]
The triangle inequality, together with the universal bound
\(d_{\mathrm{TV}}\le1\), proves the claim. This joint construction is needed
because lower and upper exceptional weights share the same normalizing constant.
\end{proof}

\paragraph{Symmetric specialization.}
When the body assumes the same upper and lower directional bounds,
\begin{align*}
\delta_{\mathrm{lower}}&=\delta_{\mathrm{upper}}=\delta,\\
\Delta_{\mathrm{lower}}&=\Delta_{\mathrm{upper}}=\Delta,
\end{align*}
the directional quantities collapse to the symmetric expressions
\begin{align*}
\varepsilon_{\mathrm{upper}}(p)&=\varepsilon(p),\\
\tau_{\mathrm{lower}}(p)&=\tau(p),\\
\eta_{\mathrm{good}}&=(1+\delta)^2-1,\\
\chi_{\mathrm{bad}}&=\Delta.
\end{align*}
Consequently,
\[
  \Gamma_{\mathrm{tv}}(p)
  =\min\left\{
    1,
    \frac{\delta}{2+\delta}
    +\frac{(1+\delta)p\Delta}{1+\delta p}
  \right\},
\]
whose uncapped form is the bound used in the body.

Applications of this symmetric parameterization to polling and comparison
audits are given in \appref{app:asymmetric-manifests}; the direct-selection
application follows next.

\subsection{Ballot-discrepancy bound for direct selection}
\label{app:direct-selection-discrepancy}

\begin{theorem}[Ballot-discrepancy lower bound under directional asymmetric manifest]
\label{thm:disc-lb-mostly-good-manifest}
Let \(\Ballots\) be partitioned into batches \(1,\dots,K\), with actual sizes
\(n_1,\dots,n_K>0\) and reported sizes \(m_1,\dots,m_K>0\). Let
\[
N:=\sum_{i=1}^K n_i=|\Ballots|,
\qquad
M:=\sum_{i=1}^K m_i=|T|.
\]
Assume the tabulation \(T\) is uniquely labeled, and let \(\ballot\) be the random ballot
obtained by first choosing batch \(i\) with probability \(r_i=m_i/M\) and then
choosing a ballot uniformly from that batch.

Suppose that the reported outcome is wrong and that the actual batch law \(\nu\)
is a directional \((p,\delta_{\mathrm{lower}},\delta_{\mathrm{upper}},\Delta_{\mathrm{lower}},\Delta_{\mathrm{upper}})\)-asymmetric manifest of the tabulated batch law
\(r\).

Then
\begin{align*}
\Exp[\disc_T(\ballot)]
&\ge
\min\!\left\{
  \frac{\mu^\tab}{1+\varepsilon_{\mathrm{upper}}(p)},
  \;
  \begin{aligned}\mu^\tab(1+\tau_{\mathrm{lower}}(p)) \\- \tau_{\mathrm{lower}}(p)\end{aligned}
\right\}\\
&-2\kappa
-4\Gamma_{\mathrm{tv}}(p).
\end{align*}
\end{theorem}

\begin{proof}
Let \(\pi_r\) denote the ballot law induced by the tabulated batch distribution
\(r\), and let \(\pi_\nu\) denote the ballot law induced by the actual batch
distribution \(\nu\). By Lemma~\ref{lem:batch-ballot-tv},
\[
d_{\mathrm{TV}}(\pi_r,\pi_\nu)=d_{\mathrm{TV}}(r,\nu).
\]
By Theorem~\ref{thm:batchwise-smooth},
\[
d_{\mathrm{TV}}(r,\nu)\le \Gamma_{\mathrm{tv}}(p)
\]
and
\[
\frac{1}{1+\varepsilon_{\mathrm{upper}}(p)}
\le
\frac{M}{N}
\le
1+\tau_{\mathrm{lower}}(p).
\]

The abstract discrepancy comparison from \appref{app:reweighting} then yields
\begin{align*}
\Exp[\disc_T(\ballot)]
&\ge
\min\!\left\{
  \frac{\mu^\tab}{1+\varepsilon_{\mathrm{upper}}(p)},
  \;
  \begin{aligned}
  \mu^\tab(1+\tau_{\mathrm{lower}}(p)) \\- \tau_{\mathrm{lower}}(p)\end{aligned}
\right\}\\
&-2\kappa
-4\Gamma_{\mathrm{tv}}(p).
\end{align*}
\end{proof}

\subsection{Upper-Bound-Only Specialization}
\label{app:upper-bound-manifests}

We now isolate the special case in which the declared manifest is only required
to be an upper bound on the true batch sizes. This is the regime
\[
  \Delta_{\mathrm{upper}} = 0
  \qquad\text{and}\qquad
  \Delta_{\mathrm{lower}} = \infty.
\]
In this setting, Theorem~\ref{thm:batchwise-smooth} already gives the stated
TV expression under the convention for
\(\Delta_{\mathrm{lower}}=\infty\). We state the test-specific one-sided
result separately because it also records the size and per-ballot sampling
guarantees used by the audit.

Let
\[
\begin{aligned}
  N &:= \sum_{i=1}^K n_i,
  &\qquad M &:= \sum_{i=1}^K m_i,\\
  u_i &:= \frac{n_i}{N},
  & r_i &:= \frac{m_i}{M}.
\end{aligned}
\]
As usual, \(u\) is the true size-proportional batch law and \(r\) is the
declared-size batch law. The corresponding ballot law \(\pi_r\) first draws a
batch \(i\) according to \(r_i\) and then draws a ballot uniformly from the
actual batch \(\Ballots_i\). Let \(U\) denote the uniform law on all physical
ballots.

Fix \(\delta\ge0\). The \emph{upper-bound condition} is
\[
  0<n_i\le m_i
  \qquad (1\le i\le K).
\]
Under this condition, define
\[
  \Bad_\delta
  :=
  \left\{i\in[K]:n_i<\frac{m_i}{1+\delta}\right\}.
\]
The statistical manifest test used below draws \(k\) batch indices
independently from \(r\), observes the true size of each sampled batch, and
accepts only if none of the sampled indices lies in \(\Bad_\delta\).

\begin{definition}[One-sided contaminated thinning]
\label{def:one-sided-thinning}
Let \(P\) and \(Q\) be probability distributions on a finite set \(X\). We say
that \(Q\) is a \((p,\delta)\)-one-sided contaminated thinning of \(P\) if
there is a function \(w:X\to(0,1]\) and a set \(\Bad\subseteq X\) such that
\[
  Q(x)=\frac{P(x)w(x)}{\sum_{y\in X}P(y)w(y)}
  \qquad\text{for all }x\in X,
\]
\[
  \frac{1}{1+\delta}\le w(x)\le 1
  \qquad\text{for all }x\notin\Bad,
\]
and
\[
  P(\Bad)\le p.
\]
\end{definition}

\begin{theorem}[Upper-bound manifest certificate]
\label{thm:upper-bound-certificate}
Assume the upper-bound condition.  Fix
\(\alpha_{\mathrm{ub}}\in(0,1)\) and \(p^\star\in(0,1)\).  If
\[
  k \ge
  \left\lceil
    \frac{\log(1/\alpha_{\mathrm{ub}})}{p^\star}
  \right\rceil,
\]
then, except with probability at most \(\alpha_{\mathrm{ub}}\), acceptance of
the statistical manifest test implies all of the following:
\begin{enumerate}[label=(\roman*)]
  \item The bad set has small declared mass:
  \[
    r(\Bad_\delta)\le p^\star.
  \]

  \item The true batch law \(u\) is a
  \((p^\star,\delta)\)-one-sided contaminated thinning of the declared batch
  law \(r\).

  \item The total actual size is controlled by
    \[
    \frac{1-p^\star}{1+\delta} M
    \le
    N
    \le
    M.
    \]

  \item The induced ballot law satisfies
  \[
    \TV(\pi_r,U)
    =
    \TV(r,u)
    \le
    \Gamma_{\mathrm{ub}}(p^\star,\delta),
    \]
    where
    \[
      \Gamma_{\mathrm{ub}}(p,\delta)
      :=
      \min\left\{
        1,\,
        \frac{\sqrt{1+\delta}-1}{\sqrt{1+\delta}+1}
        +
        \frac{(1+\delta)p}{1+\delta p}
      \right\}.
    \]

  \item Every physical ballot has sampling probability at least \(1/M\) under
  \(\pi_r\):
  \[
    \pi_r(b)\ge \frac{1}{M}
    \qquad\text{for every }b\in\Ballots.
  \]
\end{enumerate}
\end{theorem}

\begin{proof}
If \(p=r(\Bad_\delta)\ge p^\star\), then a single draw from \(r\) avoids
\(\Bad_\delta\) with probability at most \(1-p^\star\).  Hence the probability
that all \(k\) samples avoid \(\Bad_\delta\), and therefore that the test
accepts, is at most
\[
  (1-p^\star)^k
  \le
  \exp(-kp^\star)
  \le
  \alpha_{\mathrm{ub}}.
\]
Thus, except with probability at most \(\alpha_{\mathrm{ub}}\), acceptance
implies \(r(\Bad_\delta)\le p^\star\).

Define
\[
  w_i:=\frac{n_i}{m_i}.
\]
By assumption, \(0<w_i\le 1\).  For every \(i\notin\Bad_\delta\), the test
condition defining goodness gives
\[
  \frac{1}{1+\delta}\le w_i\le 1.
\]
Also,
\[
  \sum_j r_jw_j
  =
  \sum_j \frac{m_j}{M}\frac{n_j}{m_j}
  =
  \frac{N}{M}.
\]
Therefore
\[
  \frac{r_iw_i}{\sum_j r_jw_j}
  =
  \frac{(m_i/M)(n_i/m_i)}{N/M}
  =
  \frac{n_i}{N}
  =
  u_i,
\]
so \(u\) is a \((p^\star,\delta)\)-one-sided contaminated thinning of \(r\).

The same identity yields
\[
  \frac{N}{M}
  =
  \sum_i r_iw_i.
\]
Since \(w_i\le 1\) for all \(i\), \(N/M\le 1\).  Since \(w_i\ge 1/(1+\delta)\)
outside a set of \(r\)-mass at most \(p^\star\), and \(w_i>0\) on the bad set,
\[
  \frac{N}{M}
  =
  \sum_i r_iw_i
  \ge
  \sum_{i\notin\Bad_\delta}r_iw_i
  \ge
  \frac{1-p^\star}{1+\delta}.
\]

Next we bound total variation.  Let \(a:=1/(1+\delta)\).  Define an
intermediate weight
\[
  v_i :=
  \begin{cases}
    w_i, & i\notin\Bad_\delta,\\
    1, & i\in\Bad_\delta,
  \end{cases}
\]
and let
\[
  A:=\sum_i r_iv_i,
  \qquad
  q_i:=\frac{r_iv_i}{A}.
\]
Because \(a\le v_i\le 1\) for every \(i\),
Corollary~\ref{cor:sharp-reweight} gives
\[
  d_{\mathrm{TV}}(r,q)
  \le
  \frac{1-\sqrt{a}}{1+\sqrt{a}}
  =
  \frac{\sqrt{1+\delta}-1}{\sqrt{1+\delta}+1}.
\]
Now \(u\) is obtained from \(q\) by a further reweighting that is equal to \(1\)
off \(\Bad_\delta\) and equal to \(w_i\le 1\) on \(\Bad_\delta\).
Lemma~\ref{lem:general-reweight} gives
\[
  d_{\mathrm{TV}}(q,u)
  \le
  q(\Bad_\delta).
\]
Write $p:=r(\Bad_\delta)\le p^\star$. Since $v_i=1$ on
$\Bad_\delta$ and $v_i\ge 1/(1+\delta)$ elsewhere,
\[
  A=\sum_i r_iv_i
  \ge \frac{1-p}{1+\delta}+p
  =\frac{1+\delta p}{1+\delta}.
\]
Consequently,
\[
  q(\Bad_\delta)
  =\frac{p}{A}
  \le\frac{(1+\delta)p}{1+\delta p}
  \le\frac{(1+\delta)p^\star}{1+\delta p^\star},
\]
where the last inequality follows because
$x\mapsto (1+\delta)x/(1+\delta x)$ is increasing on $[0,1]$.
The triangle inequality proves the displayed bound on \(d_{\mathrm{TV}}(r,u)\), capped by
\(1\) since total variation is always at most \(1\).

Finally, for any ballot \(b\in\Ballots_i\),
\[
  \pi_r(b)
  =
  r_i\frac{1}{n_i}
  =
  \frac{m_i}{Mn_i}
  =
  \frac{1}{Mw_i}.
\]
Since \(w_i\le 1\), this is at least \(1/M\).  If \(i\notin\Bad_\delta\), then
\(w_i\ge 1/(1+\delta)\), so \(\pi_r(b)\le (1+\delta)/M\).
\end{proof}

\subsection{Bounded One-Sided Specialization}
\label{app:bounded-one-sided}

The upper-bound-only model is the limit of the following bounded one-sided
specialization as \(\Delta\to\infty\). The finite-\(\Delta\) version gives a
sharper TV bound.

\begin{corollary}[Bounded one-sided specialization]
\label{cor:bounded-one-sided}
Suppose that
\[
  \delta_{\mathrm{lower}}=\delta,\qquad
  \delta_{\mathrm{upper}}=\Delta_{\mathrm{upper}}=0,\qquad
  \Delta_{\mathrm{lower}}=\Delta.
\]
Set
\begin{align*}
  z_{\rightarrow}(p,\delta,\Delta)
  &:=
  \frac{1-p}{1+\delta}+\frac{p}{1+\Delta},
  \\
  \tau_{\rightarrow}(p,\delta,\Delta)
  &:=
  z_{\rightarrow}(p,\delta,\Delta)^{-1}-1,
  \\
  \Gamma_{\rightarrow}(p,\delta,\Delta)
  &:=
  \min\left\{
      \Theta(\Delta),
      \Theta(\delta)
      +\frac{(1+\delta)p}{1+\delta p}
       \frac{\Delta-\delta}{1+\Delta}
    \right\}.
\end{align*}
Then the following sharper TV bound is available in the bounded one-sided case:
\[
  \frac{N}{M}
  \ge
  z_{\rightarrow}(p,\delta,\Delta),
  \qquad
  \frac{M}{N}
  \le
  1+\tau_{\rightarrow}(p,\delta,\Delta),
\]
and
\[
  d_{\mathrm{TV}}(\pi_r,U)
  =
  d_{\mathrm{TV}}(r,u)
  \le
  \Gamma_{\rightarrow}(p,\delta,\Delta).
\]
\end{corollary}

\begin{proof}
The size bound follows by averaging the lower weight bounds. For TV, the
global interval \(w(i)\in[1/(1+\Delta),1]\) and
Corollary~\ref{cor:sharp-reweight} give
\[
  d_{\mathrm{TV}}(r,u)\le\Theta(\Delta).
\]
For the second bound, put \(a:=1/(1+\delta)\), define
\(v(i):=\max\{w(i),a\}\), and let \(q\) be the normalization of \(r_iv(i)\).
Then \(a\le v(i)\le1\), so
\(d_{\mathrm{TV}}(r,q)\le\Theta(\delta)\). Moreover, \(v=w\) off
\(\Bad\), and
\[
  q(\Bad)
  \le\frac{p}{a(1-p)+p}
  =\frac{(1+\delta)p}{1+\delta p}.
\]
The remaining reweighting \(h=w/v\) equals \(1\) off \(\Bad\) and satisfies
\[
  1-\frac{\Delta-\delta}{1+\Delta}\le h\le1
\]
on \(\Bad\). Lemma~\ref{lem:exceptional-reweight} and the triangle inequality
give the second term in the definition of \(\Gamma_{\rightarrow}\). Taking
the better of the global and contaminated bounds proves the claim.
\end{proof}

\subsection{Comparison of TV bounds}
\label{sec:compare-three-tv-bounds}

Throughout this section, we write the manifest parameters in directional form
\[
  (\delta_{\mathrm{lower}},\delta_{\mathrm{upper}},
   \Delta_{\mathrm{lower}},\Delta_{\mathrm{upper}}),
\]
with the symmetric case in the body recovered by setting
\[
  \delta_{\mathrm{lower}}=\delta_{\mathrm{upper}}=\delta
  \qquad\text{and}\qquad
  \Delta_{\mathrm{lower}}=\Delta_{\mathrm{upper}}=\Delta.
\]
The three manifest models considered in this paper differ only in the allowed
weight intervals:
\begin{enumerate}
\item the symmetric (SYM) model, where
  \(\delta_{\mathrm{lower}}=\delta_{\mathrm{upper}}=\delta\) and
  \(\Delta_{\mathrm{lower}}=\Delta_{\mathrm{upper}}=\Delta\);
\item the upper-bound-only (UB) model, where
  \(\delta_{\mathrm{lower}}=\delta\),
  \(\delta_{\mathrm{upper}}=\Delta_{\mathrm{upper}}=0\), and
  \(\Delta_{\mathrm{lower}}=\infty\); and
\item the bounded one-sided (OS) model, where
  \(\delta_{\mathrm{lower}}=\delta\),
  \(\delta_{\mathrm{upper}}=\Delta_{\mathrm{upper}}=0\), and
  \(\Delta_{\mathrm{lower}}=\Delta\).
\end{enumerate}

\paragraph{TV bounds achieved by the three models.}
For compactness, write
\[
  C(p,\delta):=\frac{(1+\delta)p}{1+\delta p}.
\]
The exact bound expressions established above, and used in the numerical
comparisons, are
\begin{align*}
  \Gamma_{\mathrm{SYM}}(p,\delta,\Delta)
  &:=
  \min\left\{
    1,\,
    \frac{\delta}{2+\delta}+C(p,\delta)\Delta
  \right\},\\
  \Gamma_{\mathrm{UB}}(p,\delta)
  &:=
  \min\left\{
    1,\,
    \Theta(\delta)+C(p,\delta)
  \right\},\\
  \Gamma_{\mathrm{OS}}(p,\delta,\Delta)
  &:=
  \min\left\{
    \Theta(\Delta),\,
    \Theta(\delta)
    +C(p,\delta)\frac{\Delta-\delta}{1+\Delta}
  \right\}.
\end{align*}\
The OS model gives a uniformly no-worse TV bound than either of the other two
models. Indeed,
\[
  \Theta(\delta)\le\frac{\delta}{2+\delta}
\]
and
\[
  \frac{\Delta-\delta}{1+\Delta}\le\Delta,
  \qquad
  \frac{\Delta-\delta}{1+\Delta}\le1.
\]
Consequently,
\begin{align*}
  \Gamma_{\mathrm{OS}}(p,\delta,\Delta)
  &\le \Gamma_{\mathrm{SYM}}(p,\delta,\Delta),\\
  \Gamma_{\mathrm{OS}}(p,\delta,\Delta)
  &\le \Gamma_{\mathrm{UB}}(p,\delta).
\end{align*}
The first improvement comes from ruling out upward reweighting, which reduces
the good-batch term from \(\delta/(2+\delta)\) to \(\Theta(\delta)\). The second
comes from imposing a finite lower bound on bad-batch weights: the exceptional
set is charged only for the additional thinning permitted between \(\delta\)
and \(\Delta\), rather than being charged as though all of its mass could
disappear. The independent cap \(\Theta(\Delta)\) records that even if
\(p=1\), all weights remain in the global interval \([1/(1+\Delta),1]\).

The upper-bound-only and two-sided bounds are not uniformly ordered. When
\(\Delta\ge1\), the displayed upper-bound-only bound is no larger than the
corresponding two-sided bound. For \(0\le\Delta<1\), the upper-bound-only
bound is smaller exactly when
\[
  \frac{(1+\delta)p}{1+\delta p}(1-\Delta)
  \le
  \frac{\delta}{2+\delta}-\Theta(\delta),
\]
ignoring the harmless cap at \(1\). Thus the upper-bound-only model has a
better good-batch constant but a worse exceptional-set term when \(\Delta<1\).

For small \(p\), \(\delta\), and \(\Delta\), the three bounds have the
first-order behavior
\begin{align*}
  \Gamma_{\mathrm{SYM}}(p,\delta,\Delta)
  &=\frac{\delta}{2}+p\Delta
    +O(\delta^2+\delta p\Delta),\\
  \Gamma_{\mathrm{UB}}(p,\delta)
  &=\frac{\delta}{4}+p
    +O(\delta^2+\delta p),\\
  \Gamma_{\mathrm{OS}}(p,\delta,\Delta)
  &=\min\left\{
      \frac{\Delta}{4},
      \frac{\delta}{4}+p(\Delta-\delta)
    \right\}\\
    &+O(\delta^2+\Delta^2+\delta p\Delta).
\end{align*}
In particular, the bounded one-sided model charges the bad set in proportion
to the \emph{additional} distortion \(\Delta-\delta\). 

\section{Applications to Polling and Comparison Audits}
\label{app:asymmetric-manifests}

This appendix records the directional special cases that specialize the
batchwise theorem in \appref{app:auditor-analysis}. The body uses the
case where the upper and lower directional error bounds agree,
\begin{align*}
\delta_{\mathrm{lower}}&=\delta_{\mathrm{upper}}=\delta,\\
\Delta_{\mathrm{lower}}&=\Delta_{\mathrm{upper}}=\Delta,
\end{align*}
so that the directional quantities collapse to the symmetric expressions used
in the main text. \appref{app:auditor-analysis} keeps the two sides
separate so that the symmetric case is just one specialization of the same
proof.

Statistical manifests yield a directional ``$(p, \delta_{\mathrm{lower}}, \delta_{\mathrm{upper}}, \Delta_{\mathrm{lower}}, \Delta_{\mathrm{upper}})$-contaminated'' accounting of batch sizes: every true batch size lies between $1/(1+\Delta_{\mathrm{lower}})$ and $1+\Delta_{\mathrm{upper}}$ times its reported size, and batches comprising all but a $p$-fraction of the reported ballot mass satisfy the tighter bounds $1/(1+\delta_{\mathrm{lower}})$ and $1+\delta_{\mathrm{upper}}$. 

\subsection{Polling}

A standard polling audit is defined by introducing the
function $A: \Ballots \rightarrow \R$ given by $A(\ballot)=\winner_\ballot-\loser_\ballot$.
Then
\[
  \Exp[A(\ballot_U)] = \frac{\winner^{\act}-\loser^{\act}}{|\Ballots|},
\]
where $\ballot_U$ is drawn from the uniform distribution $U$ on $\Ballots$.
Standard results from probability theory guarantee that if $\ballot_1, \ldots, \ballot_k$
are drawn uniformly (with replacement) from $\Ballots$, then the empirical
mean converges to the expected value,
\[
  \frac{1}{k}\sum_i A(\ballot_i) \approx \Exp_{\ballot \leftarrow \Ballots}[A(\ballot)],
\]
and, furthermore, the error in this approximation satisfies suitable tail
bounds; for example,
\[
  \Pr\left[\left|\frac{1}{k}\sum_i A(\ballot_i) - \Exp_{\ballot \leftarrow \Ballots}[A(\ballot)]\right| \geq \lambda\right] \leq 2\exp\left(-\frac{\lambda^2k}{2M^2}\right),
\]
where $M = \max \{|A(\ballot)|\}$.

By Lemma~\ref{lem:batch-ballot-tv}, a total-variation bound between the
tabulated and actual batch laws is exactly the same bound between their induced
ballot laws. The natural question is therefore what
can be said about this expectation when sampling is performed with the distorted
distribution. Specifically, observe that if ballots are instead drawn from a
distorted distribution $\pi$ on $\Ballots$, the same tail bounds apply to the
expected value
\[
  \Exp[A(\ballot_\pi)],
\]
where $\ballot_\pi$ denotes a ballot drawn from $\pi$ (rather than the uniform
distribution). In light of Lemma~\ref{lem:tv-bounded},
\[
  \left| \Exp[A(\ballot_\pi)] - \Exp[A(\ballot_U)] \right| \leq 2\cdot M \cdot d_{\mathrm{TV}}(\pi,U).
\]
Recalling Theorem~\ref{thm:batchwise-smooth}, this distance in total variation
$d_{\mathrm{TV}}(\pi,U)$ is no more than the joint conservative bound
\[
  \Gamma_{\mathrm{tv}}(p)
  =
  \min\left\{
    1,
    \Theta(\eta_{\mathrm{good}})
    +c_{\mathrm{bad}}(p)\chi_{\mathrm{bad}}
  \right\}.
\]
Here
\begin{align*}
  \eta_{\mathrm{good}}
  &:=(1+\delta_{\mathrm{lower}})(1+\delta_{\mathrm{upper}})-1,\\
  c_{\mathrm{bad}}(p)
  &:=\frac{(1+\delta_{\mathrm{lower}})p}
           {1+\delta_{\mathrm{lower}}p},\\
  \chi_{\mathrm{bad}}
  &:=\max\left\{
      \Delta_{\mathrm{upper}},
      \frac{\Delta_{\mathrm{lower}}}{1+\Delta_{\mathrm{lower}}}
    \right\}.
\end{align*}
Here the lower and upper bounds enter jointly because both kinds of reweighting
affect the same normalizing constant.
In the symmetric case used in the body, this reduces to the familiar
\[
  d_{\mathrm{TV}}(\pi,U)
  \leq
  \min\left\{1,
  \frac{\delta}{2+\delta}+\frac{(1+\delta)p\Delta}{1+p\delta}\right\}.
\]
This provides an immediate means for adapting a conventional polling audit to a
setting with a distorted sampling distribution given by a statistical manifest.

\balance
\subsection{Comparison}

A standard comparison audit draws CVR rows uniformly
from the tabulated CVR. A conclusion of Fuller, Harrison, and Russell~\cite{harrison2022adaptive}
is that this process is risk-limiting so long as the total number of CVR rows is
at least the total number of ballots. (In particular, duplicated identifiers
cannot compromise risk.) In this setting, the only quantity of direct relevance
is the total number of ballots, an upper bound for which is directly provided by
Theorem~\ref{thm:batchwise-smooth} in the context of a directional manifest.
Specifically,
\[
  |\Ballots| \leq \size^{\tab}\bigl(1+\varepsilon_{\mathrm{upper}}(p)\bigr),
\]
where
\[
  \varepsilon_{\mathrm{upper}}(p)=(1-p)\delta_{\mathrm{upper}}+p\Delta_{\mathrm{upper}}.
\]
Lindeman et al.~\cite{lindeman2012bravo} suggested adding dummy rows and marking them when sampled with the maximum discrepancy. This is a risk-limiting audit with parameters that depend on the gap between the size of the CVR and the upper bound on the size of the ballot population. As discussed in the body, for standard parameterization of Kaplan-Markov, obtaining a comparison audit that terminates requires $\varepsilon_{\mathrm{upper}}(p)< \mu/5$. Kaplan-Markov is designed for nonzero values to be rare, so these added items can greatly inflate the number of samples required to complete.

Instead, for a CVR $\cvr$ and a row $r$ of $\cvr$, consider the value $\disc_{\cvr}^{\Ballots}(r)$ defined to be the minimum---taken over all ballots $\ballot$ that match the CVR row identifier---of $\disc_\cvr(\ballot)$; following Harrison, Fuller, and Russell~\cite{harrison2022adaptive}, in the event that there is no matching ballot (or this is a ``dummy'' row of the CVR), define the discrepancy as though the row $r$ is matched with a ballot holding a vote for the loser, which is to say that $\loser_\ballot - \winner_\ballot = 1$. It is easy to argue~\cite{harrison2022adaptive} that by this convention
\[
 \Exp[\disc_{\cvr}^{\Ballots}(r)] \geq \mu^\tab,
\]
for any invalid election whenever the CVR has at least as many rows as there are ballots. (Here $r$ is drawn uniformly from the CVR rows.)

Assuming, for simplicity, that the ballots are organized into a single batch, write $\epsilon=\varepsilon_{\mathrm{upper}}(p)$ and let $\cvr'$ be a tabulated CVR expanded by adding $\size^{\tab}\cdot \epsilon$ additional dummy rows. As described above, this ensures that the size of the CVR is, at least, the total number of ballots. The added rows have no reported votes, so the reported margin is unchanged while the diluted margin of the expanded CVR is $\mu^\tab/(1+\epsilon)$. Then
\begin{align*}
\frac{\mu^\tab}{1+\epsilon} &\leq \Exp_{r}[\disc_{\cvr'}^{\Ballots}(r)]\\
&= \Exp_{r\mid r \in \text{ Added Rows}}[\disc_{\cvr'}^{\Ballots}(r)]\Pr[r \in \text{ Added Rows}] \\
&\quad+ \Exp_{r\mid r \in \cvr}[\disc_{\cvr}^{\Ballots}(r)]\Pr[r \in \cvr]
\end{align*}
and therefore
\[
\frac{\mu^\tab}{1+\epsilon} \leq \frac{\epsilon}{1+\epsilon} + \frac{\Exp_{r\mid r \in \cvr}[\disc_{\cvr}^{\Ballots}(r)]}{1+\epsilon}.
\]
So instead of computing on the expanded $\cvr'$, we use Kaplan-Markov to test whether
\[
\Exp_{r \in \cvr}[\disc_{\cvr}^{\Ballots}(r)] \geq \mu^\tab - \epsilon.
\]
This provides satisfactory results as long as $\mu^\tab > \varepsilon_{\mathrm{upper}}(p)$.
As discussed in the body, our technique can be seen as controlling the margin ``sacrificed'' by adding these rows so they do not unduly cause the audit to output $\inconclusive$.

\fi


\fi
\end{document}